\documentclass[reqno,12pt]{amsart}

\usepackage[a4paper, 
            left=1in,
            right=1in,
            top=1in,
            bottom=1in]{geometry}

\usepackage[utf8]{inputenc}
\usepackage{microtype}
\usepackage{graphicx}
\usepackage{amssymb} 
\usepackage[colorlinks = true,
            linkcolor = blue,
            urlcolor  = blue,
            citecolor = blue,
            anchorcolor = blue]{hyperref}

\usepackage{mathtools}  
\usepackage{thmtools} 
\usepackage{thm-restate}

\usepackage{microtype}

\theoremstyle{plain}
\newtheorem{thm}{}[section]
\newtheorem{lemma}[thm]{Lemma}
\newtheorem{proposition}[thm]{Proposition} 
\newtheorem{conjecture}[thm]{Conjecture}
\newtheorem{theorem}[thm]{Theorem}

\newtheorem{mainquestion}{Main Question} 
\newtheorem{openquestion}[thm]{Open Question}
\theoremstyle{remark}
\newtheorem{claim}[thm]{Claim}

 \newtheorem{example}[thm]{Example}

\usepackage{algorithmic}

\newenvironment{claimproof}
  {\proof[Proof (claim)]}
  {\endproof}

\usepackage{centernot}  
\usepackage{xspace}  
\usepackage{float}

\usepackage[scr=boondoxo]{mathalpha}  
\newcommand{\coloneqq}{\mathrel{\mathop:}\mathrel{\mkern-1.2mu}=} 

  \newcommand{\acts}{\curvearrowright} 
\newcommand{\astlarge}{\circledast}

\newcommand{\Pol}{\ensuremath{\mathrm{Pol}}} 
 
\newcommand{\struct}[1]{\mathfrak{#1}}    
\newcommand{\blowup}[1]{(\mathbb{Q}\wr #1,{\neq})}     
\newcommand{\blowupi}[1]{(\mathbb{Q}\wr #1,{\neq},I_4)}  
\newcommand{\ordblowup}[1]{\big((\mathbb{Q};<) \wr #1,{\neq}\big)}  
\newcommand{\ordblowupi}[1]{\big((\mathbb{Q};<) \wr #1,{\neq},I_4\big)}

\newcommand{\CSP}{\ensuremath{\mathrm{CSP}}\xspace}   
   
\newcommand{\PCSP}{\ensuremath{\mathrm{PCSP}}\xspace}   

\newcommand{\age}{\ensuremath{\mathrm{age}}\xspace}

\newcommand{\Aut}
{\ensuremath{\mathrm{Aut}}\xspace}  
\newcommand{\End}{\ensuremath{\mathrm{End}}\xspace}

\newcommand{\Sym}{\ensuremath{\mathrm{Sym}}\xspace}

\usepackage{stmaryrd}   
\usepackage{booktabs}

\newcommand{\code}[1]{\ensuremath{\llbracket#1\rrbracket}\xspace}

\newcommand{\orbeq}[2]{#1/\text{\small$#2$}}
\newcommand{\of}[1]{\raisebox{0.25pt}{\textup{{\kern-0.025em{\footnotesize[\raisebox{-0.05pt}{$#1$}}\kern-0.025em{\footnotesize]}\kern0.025em}}}}

\newcommand{\fpwr}[2]{{#1}^{[#2]}}

\usepackage{xcolor}
\usepackage{pdfcomment}
\usepackage{mathtools}   
\usepackage{forest}

\usepackage{lineno}

\begin{document}

\title[Fundamental Questions in Infinite-Domain Constraint Satisfaction]{Three Fundamental Questions in Infinite-Domain Constraint Satisfaction}

\author[M. Pinsker]{Michael Pinsker}
\author[J. Rydval]{Jakub Rydval}
\author[M. Sch\"{o}bi]{Moritz Sch\"{o}bi}
\author[C. Spiess]{Christoph Spiess}
\author[P. Winkler]{Paul Winkler}
\thanks{\emph{Michael Pinsker, Jakub Rydval, and Moritz Sch\"{o}bi}: This research was funded in whole or in part by the Austrian Science Fund (FWF) [I 5948, ESP 1571225]. For the purpose of Open Access, the authors have applied a CC BY public copyright licence to any Author Accepted Manuscript (AAM) version arising from this submission. 
\smallskip \newline \emph{Michael Pinsker, Christoph Spiess, and Paul Winkler}: This research is  funded by the European Union (ERC, POCOCOP, 101071674). Views and opinions expressed are however those of the author(s) only and do not necessarily reflect those of the European Union or the European Research Council Executive Agency. Neither the European Union nor the granting authority can be held responsible for them.
\smallskip\newline A conference version of this material appeared in the Proceedings of the 50th International Symposium on Mathematical Foundations of Computer Science (MFCS), \#83, 1--20, 2025~\cite{pinsker2025three}.
}

\address{Institut f\"{u}r Diskrete Mathematik und Geometrie, FB Algebra, TU Wien, Austria}
\email{\{michael.pinsker,jakub.rydval,moritz.schoebi,christoph.spiess\}@tuwien.ac.at}  
\address{Institut f\"{u}r Algebra, TU Dresden, Germany}
\email{paul.winkler1@tu-dresden.de}

 \begin{abstract} The Feder-Vardi dichotomy conjecture for Constraint Satisfaction Problems (CSPs) with finite templates, confirmed independently by Bulatov and Zhuk, has an extension to certain well-behaved infinite templates due to Bodirsky and Pinsker which remains wide open. 
 We formulate three fundamental questions on the scope of the Bodirsky-Pinsker conjecture and provide positive answers to them.
 
Our first two main results provide two simplifications of this scope, one of  structural, and the other one of algebraic nature. 
The former simplification implies that the conjecture is equivalent to its restriction to templates without algebraicity, a crucial assumption in the most powerful classification methods. The latter yields that the higher-arity invariants of any template within its scope can be assumed to be essentially injective, and any algebraic condition characterizing any complexity class within the conjecture closed under Datalog reductions must be satisfiable by injections, thus lifting the mystery of the better applicability of certain algebraic conditions over others.
 
Our third main result uses the first one to show that any non-trivially tractable template within the scope serves, up to a Datalog-computable modification of it, as the witness of the tractability of a non-finitely tractable finite-domain Promise Constraint Satisfaction Problem (PCSP) by the so-called sandwich method. This provides a particularly strong connection between the Bodirsky-Pinsker conjecture and finite-domain PCSPs.

In the light of the third main result, we initiate a new case study\textemdash of phylogeny CSPs\textemdash which we investigate from the perspective of descriptive complexity.
Within this study, 
we show that there exists a tractable phylogeny CSP that pp-constructs a finite-domain PCSP inexpressible in fixed-point logic with counting but does not pp-construct any finite-domain CSP with this property.
\end{abstract}

\maketitle 
 
\renewcommand{\linenumberfont}{\normalfont \scriptsize\color{gray}}

\section{Introduction}\label{section:introduction}

 \subsection{The finite-domain CSP dichotomy theorem} 
 Fixed-template \emph{Constraint Satisfaction Problems} are computational problems parametrized by  structures $\struct{A}$ with a finite relational signature, called \emph{templates}; they are denoted by $\CSP(\struct{A})$ and ask whether a given structure $\struct{X}$ with the same signature as $\struct{A}$ admits a homomorphism to $\struct{A}$. 
The general CSP framework is incredibly rich, in fact, it contains all decision problems up to polynomial-time equivalence~\cite{bodirsky2008non}.
For this reason, CSPs are typically studied under additional structural restrictions on the template $\struct{A}$.
The restriction that has received most attention is requiring the domain of $\struct{A}$ to be finite, and the problems arising in this way are called \emph{finite-domain} CSPs.
Already in this setting, one obtains many well-known problems such as HORN-SAT, 2-SAT, 3-SAT, or 3-COLORING.

Consider the example of 3-COLORING. This NP-complete problem can be modeled as $\CSP(\struct{K}_3)$, where $\struct{K}_3$ stands for the complete graph on three vertices: indeed, a homomorphism to $\struct{K}_3$ from a given structure $\struct{X}$ in the same signature (i.e.~from a  graph $\struct{X}$) is simply a  mapping that does not map any two adjacent vertices of $\struct{X}$ to the same vertex of $\struct{K}_3$.
A key observation in the theory of CSPs is that the NP-completeness of $\CSP(\struct{K}_3)$ can be extended to CSPs of other structures which can ``simulate'' $\struct{K}_3$ over their domain or over a finite power thereof using \emph{primitive positive} (pp) formulas (such simulations are called \emph{pp-interpretations}).
Enhancing this by the additional observation that \emph{homomorphically equivalent}  templates have equal CSPs, one arrives at the notion of \emph{pp-constructibility}~\cite{barto2018wonderland}.  
If $\struct{A}'$ is pp-constructible from $\struct{A}$, then $\CSP(\struct{A}')$ is reducible to $\CSP(\struct{A})$ in logarithmic space; moreover, this reduction can be formulated in the logic programming language \emph{Datalog}.
This uniform  reduction between CSP templates is sufficiently powerful, as we will see below, to describe all NP-hardness amongst finite-domain CSPs.

In the early 2000s, the field of finite-domain constraint satisfaction quickly rose in fame, in particular due to the discovery of  tight  connections  with universal algebra.
The most basic of these connections is that two structures with identical domains have the same sets of pp-definable relations if and only if they have the same sets of polymorphisms~\cite{jeavons1997closure,jeavons1998algebraic}.
Here, a \emph{polymorphism} of a relational structure $\struct{A}$ is simply a homomorphism from a finite power of $\struct{A}$ into $\struct{A}$ itself; we denote by $\Pol(\struct{A})$ the set of all polymorphisms of $\struct{A}$, the \emph{polymorphism clone} of $\struct{A}$.
Since taking expansions by pp-definable relations does not lead to an increase in complexity, the study of finite-domain CSPs is subsumed by the study of finite algebras.

Over the past two decades, the link between universal algebra and CSPs was  gradually refined.
After an intermediate stop at %
pp-interpretations~\cite{bulatov2005classifying}, today we know that also pp-constructibility between finite relational structures is fully encoded in their polymorphisms~\cite[Thm.~1.3]{barto2018wonderland}: $\struct{A}'$ is pp-constructible from  $\struct{A}$ if and only if $\Pol(\struct{A}')$ satisfies all \emph{height-1 identities} of $\Pol(\struct{A})$.
By a breakthrough result of Bulatov~\cite{bulatov2017dichotomy} and Zhuk~\cite{zhuk2017proof,zhuk2020proof}, the pp-construction of $\struct{K}_3$ is the unique source of NP-hardness for finite-domain CSPs (if P$\neq$NP), and in fact, they provided a polynomial-time algorithm for any such CSP which does not  pp-construct $\struct{K}_3$. 
Besides its complexity-theoretic significance (in particular, of confirming the dichotomy conjecture of Feder and Vardi~\cite{federvardi1998}), this result has the additional appeal that the border for tractability can be described by a neat universal-algebraic condition on polymorphism clones.
We state the theorem of Bulatov and Zhuk in a formulation which takes into account the results in~\cite{barto2018wonderland,Siggers_2010}.

\begin{theorem}[Bulatov and Zhuk~\cite{bulatov2017dichotomy,zhuk2017proof}] \label{thm:finite_domain_CSP} Let $\struct{A}$ be a finite relational structure.
Either $\struct{A}$ pp-constructs $\struct{K}_3$ and $\CSP(\struct{A})$ is NP-complete, or $\Pol(\struct{A})$ satisfies the \emph{Siggers identity} $ s(x,y,z,x,y,z) \approx s(y,z,x,z,x,y)$ and $\CSP(\struct{A})$ is solvable in polynomial time.  
\end{theorem} 
 
\subsection{The infinite-domain CSP tractability conjecture}\label{subsection:intro_conjecture}
 
Arguably, the two main reasons for the popularity  of finite-domain CSPs over CSPs which require an infinite template are immediate containment of such problems in the class NP (a homomorphism can be guessed and verified in polynomial time) as well as the above-mentioned applicability of algebraic methods  which   had been developed independently of CSPs for decades. Yet, even within the class NP, finite-domain CSPs are of an extremely restricted kind: already simple problems such as ACYCLICITY of directed graphs (captured by the template $(\mathbb Q;<)$),  various natural coloring problems for graphs such as NO-MONOCHROMATIC-TRIANGLE, and more generally the model-checking problem for  natural restrictions of existential second-order logic such as the logic MMSNP of Feder and Vardi~\cite{federvardi1998}, are beyond its primitive scope despite their containment in NP. 

The quest for a CSP framework including such problems whilst staying within NP and allowing for an algebraic approach akin to the one for finite templates  started with the popularization of \emph{$\omega$-categoricity} by Bodirsky and Ne\v{s}et\v{r}il~\cite{bodirsky2006constraint} as a sufficient structural restriction ensuring the latter:
for countable $\omega$-categorical structures, pp-definability is determined by their polymorphisms and, as was subsequently shown, so are the general reduction of pp-interpretability~\cite{Topo-Birk} as well as the pp-constructibility of $\struct{K}_3$~\cite{barto2018wonderland}.
The  $\omega$-categoricity of $\struct{A}$ can be interpreted as $\struct{A}$ being ``finite modulo automorphisms'':
more precisely, the  action of its  automorphism group $\Aut(\struct{A})$ on $d$-tuples has only finitely many orbits for all finite $d\geq 1$. 
It is, however, a mathematical property with little computational bearing: the CSPs of $\omega$-categorical structures can be monstrous from this perspective~\cite{GJKMP-conf,gillibert2022symmetries}, namely, complete for a variety of complexity classes of arbitrarily high complexity. 
For this reason, and based on strong empirical evidence, Bodirsky and Pinsker~\cite{bodirsky2021projective} identified  a proper subclass of countable $\omega$-categorical structures as a candidate for an algebraic complexity dichotomy extending the theorem of Bulatov and Zhuk which does not seem to suffer from this deficiency.

Their first requirement on $\struct{A}$ is that the orbit under $\Aut(\struct{A})$ of any $d$-tuple is determined by the relations that hold on it: we call a relational structure $\struct{A}$ \emph{homogeneous} if every isomorphism between its finite substructures extends to an automorphism of $\struct{A}$.
This gives an effective way of representing orbits, which are central to $\omega$-categoricity.
They then additionally impose an effective way of determining whether a given finite substructure is contained in $\struct{A}$, which places the CSP in NP:
a relational structure $\struct{A}$ over a finite signature $\tau$ is \emph{finitely bounded}~\cite{macpherson2011survey} if there exists a finite set $\mathcal{N}$ of finite $\tau$-structures (\emph{bounds}) such that a finite $\tau$-structure embeds into $\struct{A}$ if and only if it does not embed any member of $\mathcal{N}$;
equivalently,  the finite substructures are up to isomorphism precisely the finite models of some universal first-order sentence (see e.g.~\cite[Lem.~2.3.14]{bodirsky2021complexity}).
\begin{example} \label{ex:properties_of_q}
    A standard example of a countable finitely bounded homogeneous structure is $(\mathbb{Q};<)$. Its finite substructures are all finite strict linear orders, which can be axiomatized by irreflexivity, totality, and transitivity:
    \[ 
    \forall x,y,z \big( \neg( x<x) \wedge ( x<y \vee x=y \vee  y<x) \wedge  (x<y \wedge y<z \Rightarrow x<z) \big).
    \]
    Moreover, every isomorphism between two such finite substructures can be extended to an automorphism of $(\mathbb{Q};<)$, e.g.~by a piecewise affine transformation.
\end{example}

Finite boundedness and homogeneity taken together are a very strong assumption.
Bodirsky and Pinsker observed that for all practical purposes, in particular the applicability of polymorphisms to complexity as well as containment in NP, it is sufficient that the template $\struct{A}$ be a \emph{first-order reduct} of a structure $\struct{B}$ enjoying these, i.e.  first-order definable therein. 
 This yields  sufficient flexibility to model a huge class of computational problems which includes in particular the ones mentioned above.  Requiring $\struct{A}$ to be a \emph{reduct} of a finitely bounded homogeneous structure $\struct{B}$, i.e.~obtained from $\struct{B}$ by forgetting relations,  turns out to be equivalent (see~\cite[Prop.~7]{baader_rydval}), and we shall find this  approach convenient. 
 \begin{example} \label{ex:properties_of_q1} The CSP of the template $(\mathbb Q;\{(x,y,z) \mid x<y<z \vee z<y<x\})$, first-order definable in $(\mathbb{Q};<)$, is the classical NP-complete \emph{betweenness problem} from the 1970s~\cite{opatrny}. 
 \end{example}
 The following formulation of the original conjecture from 2012 takes into account later progress~\cite{barto2018wonderland,barto_pinsker_proc,barto_pinsker_journal,barto2019equations,BKOPP}.

\begin{conjecture}[Bodirsky and Pinsker~\cite{bodirsky2021projective}]  
  \label{conj:bodirsky_pinsker} Let $\struct{A}$ be a reduct of a countable finitely bounded
homogeneous structure $\struct{B}$. Then exactly one of the following holds.
\begin{itemize}
    \item $\struct{A}$ pp-constructs $\struct{K}_3$ (and consequently, $\CSP(\struct{A})$ is NP-complete);
    \item $\struct{A}$ does not pp-construct $\struct{K}_3$; in this case $\Pol(\struct{A})$ satisfies the \emph{pseudo-Siggers identity}  \[\alpha \circ s(x,y,z,x,y,z) \approx \beta \circ s(y,z,x,z,x,y)\] and  $\CSP(\struct{A})$  is solvable in polynomial time.
\end{itemize} 
\end{conjecture} 
 
 It is important to note that the open part of Conjecture~\ref{conj:bodirsky_pinsker} is only the consequence of polynomial-time tractability: the negation of the first item does yield the satisfaction of the pseudo-Siggers identity by a theorem due to Barto and Pinsker~\cite{barto_pinsker_journal}. It is even equivalent to it for \emph{model-complete cores}~\cite[Thm.~1.3]{barto2019equations}: an $\omega$-categorical model-complete core is a template $\struct{A}$ whose endomorphisms preserve all orbits of $\Aut(\struct{A})$; an $\omega$-categorical model-complete core exists for all CSPs with an $\omega$-categorical template and is unique up to isomorphism~\cite{bodirsky2007cores}.

\subsection{Promise CSPs}

One of the currently most vibrant branches of research in constraint satisfaction is the study of \emph{Promise CSPs} (PCSPs).
For two relational structures $\struct{S}_1$ and $\struct{S}_2$ such that $\struct{S}_1$ maps homomorphically to $\struct{S}_2$, the problem  $\PCSP(\struct{S}_1,\struct{S}_2)$ asks to decide whether a given finite  structure $\struct{X}$ in the same signature as $\struct{S}_1$ and $\struct{S}_2$ maps homomorphically to $\struct{S}_1$ or does not even map homomorphically to $\struct{S}_2$; the instance $\struct{X}$ is  promised to satisfy one of those two cases. Thus, the problem $\PCSP(\struct{S}_1,\struct{S}_2)$ asks to compare $\CSP(\struct{S}_1)$ to a relaxation $\CSP(\struct{S}_2)$ of it, and  
since $\PCSP(\struct{S}_1,\struct{S}_1)=\CSP(\struct{S}_1)$, promise constraint satisfaction problems generalize the CSP framework. This generalization is vast, and contains many well-known computational problems such as  APPROXIMATE GRAPH COLORING \cite{approximatecolouring}, $(2+\varepsilon)$-SAT \cite{2plusepsilonSAT}, and HYPERGRAPH COLORING~\cite{hypergraphcol}. 
Although formally not necessary,  research has focused on finite templates $(\struct{S}_1,\struct{S}_2)$: this is justified for example by the fact that even for structures over a Boolean domain, there is no complete complexity classification yet. 
Hence, the two extensions considered here of finite-domain CSPs to $\omega$-categorical templates, or alternatively to PCSPs, are somewhat orthogonal, and it is natural to wonder whether there exist any connections, in particular of algorithmic nature.

\subsection{Contributions}\label{section:contributions}

We formulate and affirmatively answer three fundamental questions around Conjecture~\ref{conj:bodirsky_pinsker} that have been implicitly present in the field for some years now.
The first  concerns the scope of Conjecture~\ref{conj:bodirsky_pinsker}; the second the algebraic invariants of templates therein; and the third  its connection with the rapidly evolving  field of (finite-domain) PCSPs. 

 \begin{mainquestion}\label{mq1}
Are there significant additional structural assumptions, perhaps of model-theoretic nature, that  can be imposed onto the structures of Conjecture~\ref{conj:bodirsky_pinsker} without loss of generality, i.e.~without affecting the truth of the conjecture?
\end{mainquestion}

\begin{mainquestion}\label{mq2}
    Are there significant algebraic assumptions that can be imposed on  the polymorphisms of the structures of Conjecture~\ref{conj:bodirsky_pinsker} without loss of generality?
\end{mainquestion}

\begin{mainquestion}\label{mq3}
Are there significant connections between CSPs from  Conjecture~\ref{conj:bodirsky_pinsker} and finite-domain PCSPs?
\end{mainquestion}

\subsubsection{Algebraicity is irrelevant} 
\label{section:introduction_no_alg}

 Conjecture~\ref{conj:bodirsky_pinsker} has been confirmed for many different subclasses. Mottet and Pinsker speak of first- and second-generation classifications~\cite{mottet2024smooth}:
 those of the former kind  can  roughly be described as exhaustive case-by-case analyzes, using Ramsey theory, of the available polymorphisms for the first-order reducts of a fixed finitely bounded homogeneous structure $\struct{B}$ -- extensive complexity classifications such as those for temporal CSPs~\cite{ComplOfTempCSPs}, Graph-SAT problems~\cite{BodPin-Schaefer-both}, Poset-SAT problems~\cite{posetCSP16}, and CSPs with templates first-order definable in arbitrary homogeneous graphs~\cite{BMPP16} were obtained this way.
 Second-generation classifications 
 take a more structured approach mimicking advanced algebraic methods for finite domains; they were  employed, for example,  to achieve classifications for the logic MMSNP~\cite{MMSNP-Journal}, Hypergraph-SAT problems~\cite{mottet_et_al:LIPIcs.ICALP.2024.148}, and certain graph orientation problems~\cite{feller2024algebraic,bitter2024completion}.
 However, even the most advanced methods as described in~\cite{mottet2024smooth} require the assumption of additional abstract structural properties of the template.
 The first aim of this paper is to investigate whether any such property can be assumed without loss of generality when investigating whether Conjecture~\ref{conj:bodirsky_pinsker} holds.

 One such assumption is that the template $\struct{A}$, or even the structure $\struct{B}$ in which it is first-order defined, has \emph{no algebraicity}.
 It is present virtually everywhere in the infinite-domain CSP
literature;
 and in particular an important prerequisite in most of the general results of the theory of \emph{smooth approximations}  developed by Mottet and Pinsker~\cite{mottet2024smooth}, as well as in all results of Bodirsky and Greiner~\cite{bodirsky2020complexity,bodirsky2021tractable} about CSPs of \emph{generic superpositions} of theories.
We will show that the following  strengthening of this property (for model-complete cores, see Proposition~\ref{prop:CSP_injective_no_algebraicity}) can be assumed when resolving Conjecture~\ref{conj:bodirsky_pinsker}: a structure $\struct{A}$ is \emph{CSP-injective} if every finite structure that homomorphically maps to $\struct{A}$ also does so injectively; in other words, if there is a solution to an instance of $\CSP(\struct{A})$, then there is also an injective one.
CSP-injectivity played an essential role in the universal-algebraic proof of the complexity dichotomy for Monotone Monadic SNP~\cite{MMSNP-Journal} and, more recently, certain finitely bounded homogeneous uniform hypergraphs~\cite{mottet_et_al:LIPIcs.ICALP.2024.148}.

In our first main result, Theorem~\ref{thm:removing_algebraicity}, we show that any structure $\struct{A}$ in the scope of Conjecture~\ref{conj:bodirsky_pinsker} can be transformed to a CSP-injective structure $\blowup{\struct{A}}$ such that there are Datalog-reductions in both directions between $\CSP(\struct A)$ and $\CSP\blowup{\struct{A}}$. We further show that $\struct{A}$ pp-constructs $\struct{K}_3$ if and only if $\blowup{\struct{A}}$ pp-constructs $\struct{K}_3$ and essential model-theoretic properties of $\struct{A}$
are passed on to $\blowup{\struct{A}}$.
In particular, $\blowup{\struct{A}}$ is again in the scope of Conjecture~\ref{conj:bodirsky_pinsker} and both tractability
and
the hardness criterion of the conjecture are preserved.

This is essentially done by blowing up every point in $A$ to a countably infinite set and then defining relations on the union of those infinite sets in a way that remembers the structure of $\struct{A}$. Formally, it is achieved by the so-called \emph{relational wreath product} of an infinite set with the structure $\struct{A}$. 

 \medskip 
\begin{center}
\framebox{\parbox{0.85\textwidth}{Throughout the present article, the notation $\struct{C}\wr \struct{A}$ stands for the relational wreath product of two structures $\struct{C}$ and $\struct{A}$, introduced formally in Section~\ref{section:wreath}.}}
\end{center} 

\medskip

The only assumption necessary to guarantee Datalog-interreducibility of 
$\CSP(\struct{A})$ and $\CSP\blowup{\struct{A}}$
is that $\struct{A}$ is \emph{non-contractible}, i.e.~its CSP is not the CSP of a one-element structure.

\begin{restatable}{theorem}{maintheorem}   \label{thm:removing_algebraicity}
Let $\struct{B}$ be a countable structure over a finite relational signature. Then 
\begin{enumerate}
    \item \label{item:Bremovingalgebraicity} $\ordblowup{\struct{B}}$ is a countably infinite structure without algebraicity over a finite relational signature.
    \item \label{item:Bomegacat} If $\struct{B}$ is $\omega$-categorical, then so is $\ordblowup{\struct{B}}$.
    \item \label{item:BhomogRamsey} If $\struct{B}$ is homogeneous, then so is $\ordblowup{\struct{B}}$. In this case, $\struct{B}$ is Ramsey if and only if  $\ordblowup{\struct{B}}$ is Ramsey.
    \item \label{item:Bfinbounded} If $\struct{B}$ is finitely bounded, then so is $\smash{\ordblowup{\struct{B}}}$.
\end{enumerate}

Moreover, if $\struct{A}$ is a non-contractible reduct of $\struct{B}$, then 
\begin{enumerate}\addtocounter{enumi}{4}
    \item \label{item:ACSPinj} $\blowup{\struct{A}}$ is a non-contractible CSP-injective reduct of $\ordblowup{\struct{B}}$.
    \item \label{item:ADatalog} $\CSP(\struct{A})$ and $\CSP\blowup{\struct{A}}$ are Datalog-interreducible.
    \item \label{Amodelcomplcore} If $\struct{A}$ is $\omega$-categorical, then so is $\blowup{\struct{A}}$, and $\struct{A}$ is a  model-complete core if and only if  $\blowup{\struct{A}}$ is. 
    \item \label{AppconstructK3}  If the number of orbits of $d$-tuples of $\struct{A}$  grows slower than $2^{2^{d}}$, then $\struct{A}$ pp-constructs $\struct{K}_3$ if and only if $\blowup{\struct{A}}$ pp-constructs $\struct{K}_3$. 
\end{enumerate}
 \end{restatable} 

 It is well-known and  not hard to see that item~\ref{AppconstructK3} of Theorem~\ref{thm:removing_algebraicity} applies to reducts of homogeneous structures in a finite relational signature~\cite{macpherson2011survey}, in particular  to all structures in the scope of Conjecture~\ref{conj:bodirsky_pinsker}.

\subsubsection{Polymorphisms are injective}
 
Comparing Theorem~\ref{thm:finite_domain_CSP} with Conjecture~\ref{conj:bodirsky_pinsker}, one notices the replacement of the Siggers identity by its (weaker) pseudo-variant. This is not just an inconvenience arising from the proof of~\cite[Thm.~1.6]{barto_pinsker_journal} by Barto and Pinsker (or the more recent proof in~\cite{BBKMP23}), but there are examples showing the necessity to weaken the condition. 
\begin{example} \label{ex:example_I4}
Consider the template $(\mathbb Q;I_4)$, where $I_4$ is defined by
\[I_4 \coloneqq \{ (x,y,u,v) \mid x=y \Rightarrow u=v\}.\] Its CSP is solvable in Datalog (Sections~\ref{section:fp} and~\ref{section:everything_together}). 
Since $\CSP(\struct{K}_3)$ is not, %
and pp-constructions provide Datalog-reductions~\cite{atserias2009affine}, %
it follows that $(\mathbb{Q};I_4)$ does not pp-construct $\struct{K}_3$. Hence, its polymorphisms satisfy the pseudo-Siggers identity by the above-mentioned theorem of Barto and Pinsker.
However, any polymorphism of $I_4$ is injective up to dummy variables~\cite{BodChenPinsker}, and hence cannot satisfy the Siggers identity. This makes it clear that the characterization of tractability of CSPs by the Siggers identity cannot be lifted from finite to infinite domains.
\end{example}

One prominent example of a condition that can actually be lifted is the satisfiability of \emph{quasi near-unanimity identities}, which captures \emph{bounded strict width} of CSPs both in the finite and in the infinite~\cite{bodirsky2013datalog}; bounded strict width corresponds to solvability in a particular proper fragment of Datalog that is not closed under the Datalog-reductions in Theorem~\ref{thm:removing_algebraicity}.

In order to shed light on which algebraic criteria might be applicable to characterize tractability (or perhaps finer complexity classes) within the scope of Conjecture~\ref{conj:bodirsky_pinsker}, we investigate whether one can impose any assumptions on the polymorphisms of the templates therein without loss of generality. We say that a template is \emph{Pol-injective} if all of its polymorphisms are \emph{essentially injective}, i.e.~injective up to dummy variables -- this is the case if and only if the relation $I_4$ is invariant under its polymorphisms.  
 As it turns out, the relation $I_4$ can be added to the templates constructed in Theorem~\ref{thm:removing_algebraicity} at no computational cost (up to Datalog reductions). 
 It thus follows that any algebraic condition on polymorphisms capturing a complexity class closed under Datalog-reductions for CSPs of structures within the scope of Conjecture~\ref{conj:bodirsky_pinsker} must necessarily be satisfiable by essentially injective operations.
Until now, various  concrete examples of natural templates all of whose polymorphisms are essentially injective had pointed towards  a statement of this kind (see, e.g.~\cite{barto2019equations}); we provide a rigorous confirmation.

\begin{restatable}{theorem}{maintheoremtwo}    \label{thm:polinjective}
The structures $ \blowupi{\struct{A}} $ and  $ \ordblowupi{\struct{B}} $ enjoy all properties of $\blowup{\struct{A}}$ and $\ordblowup{\struct{B}}$ from Theorem~\ref{thm:removing_algebraicity} except that CSP-injectivity is replaced by Pol-injectivity.   
\end{restatable} 

For finite structures the only essentially injective polymorphisms are at most unary up to dummy variables.
It is therefore of no surprise that so few universal-algebraic conditions for CSPs from the finite have been successfully lifted to the infinite. 
Examples of conditions for polymorphism clones that can be satisfied by essentially injective operations are any \emph{pseudo height-1 identities} %
(e.g., the pseudo-Siggers identity), but also the \emph{dissected near-unanimity identities}~\cite{gillibert2022symmetries} which were proposed in~\cite{bodirsky2022descriptive} as a possible candidate for solvability of CSPs in \emph{Fixed-Point logic with Counting} (FPC).

Marimon and Pinsker have recently provided a negative answer, within the scope of  Conjecture~\ref{conj:bodirsky_pinsker}, to a question of Bodirsky~\cite[Q.~14.2.6(27)]{bodirsky2021complexity} asking whether every $\omega$-categorical CSP template without algebraicity solvable in Datalog has a binary injective polymorphism~\cite[Cor.~8.10]{MarimonPinsker23}.  We contrast this negative result with a corollary to Theorem~\ref{thm:polinjective} that implies that up to Datalog reductions the answer is positive.
 
\begin{restatable}{corollary}{binaryInj}  \label{cor:binaryInj} 
  Let $\struct{A}$ be a relational structure that 
  does not pp-construct $\struct{K}_3$. Then $\blowupi{\struct{A}}$ has a binary injective polymorphism. 
\end{restatable}

\subsubsection{Everything is a cheese}
While  research on both infinite-domain CSPs and finite-domain PCSPs has been very fruitful in recent years, few connections between these two branches have been discovered so far. 
 One natural link is the use of  CSP templates as cheeses in the so-called sandwich method to prove tractability of PCSPs, as follows. 
Note that if $\struct{S}_1$ maps homomorphically to $\struct{A}$ and $\struct{A}$ maps homomorphically to $\struct{S}_2$, then $\PCSP(\struct{S}_1,\struct{S}_2)$ trivially reduces to $\CSP(\struct{A})$; we call $\struct{A}$ a \emph{cheese sandwiched} by the template $(\struct{S}_1,\struct{S}_2)$. 
This simple observation is more powerful than it might first appear: in fact, so far, every known tractable (finite-domain)  PCSP can be reduced to the tractable CSP of some cheese structure via the sandwich method~\cite{PromiseandInfDomCSPs,sandwichespromiseconstraintsatisfaction}. 
On the other hand, it can be necessary to pick large finite  cheeses even for PCSPs on a Boolean domain~\cite{smallpcspsreducetolarge}, and Barto~\cite{promisesmakefiniteproblemsinfinitary} even provided an example of a 
PCSP template sandwiching an infinite polynomial-time tractable cheese while not being \emph{finitely tractable}, i.e.~such that every finite-domain cheese has an NP-complete CSP~\cite{promisesmakefiniteproblemsinfinitary}. Since Barto's cheese is not $\omega$-categorical, he raised the question whether such an example could also be constructed with an $\omega$-categorical cheese (\cite[Q.~IV.1]{promisesmakefiniteproblemsinfinitary}). Subsequently the question arose, and was promoted by Zhuk at the CSP World Congress~2023, whether structures within the scope of Conjecture~\ref{conj:bodirsky_pinsker} could \emph{ever} serve as tractable cheeses in the absence of a finite one. 

The first two examples of tractable $\omega$-categorical cheeses for non-finitely tractable finite-domain PCSPs were recently obtained in~\cite{Mottet_2025}; these structures fall within the scope of Conjecture~\ref{conj:bodirsky_pinsker}.
In the present paper, we strengthen this result by creating non-finitely tractable finite-domain PCSPs from virtually every structure $\struct{A}$ within the scope of Conjecture~\ref{conj:bodirsky_pinsker}. 
Our construction consists of two steps. 
First, we show that such a construction is possible in general under the additional assumption that $\struct{A}$ is equipped with an inequality predicate and that it is a reduct of a structure $\struct{B}$ which is linearly ordered. 
We then combine this construction  with Theorem~\ref{thm:removing_algebraicity}.
Modulo an extra step consisting of taking \emph{generic superpositions} (Proposition~\ref{prop:generic_superpositions}) of structures with $(\mathbb{Q};<)$ (for which we need the property of no algebraicity), this finishes our proof.

\begin{restatable}{theorem}{sandwichthm} \label{thm:main_theorem_sandwiches}  
Let $\struct{A}$ be a non-contractible reduct of a countable finitely bounded homogeneous structure $\struct{B}$.
Then there exists a reduct $\struct{A}'$ of a countable finitely bounded homogeneous structure $\struct{B}'$ together with Datalog-reductions $\mathcal{J}$ and $\mathcal{J}'$ from $\CSP(\struct{A})$ to $\CSP(\struct{A}')$ and vice versa such that, for every $n\in \mathbb{N}$, there exists a finite PCSP template $(\struct{S}_1,\struct{S}_2)$ with the following properties:  
 \begin{enumerate} 
  \item \label{item:sandwichthm5} $\struct{A}'$ is a cheese for $\PCSP(\struct{S}_1,\struct{S}_2)$;   thus, $\PCSP(\struct{S}_1,\struct{S}_2)$ reduces to $\CSP(\struct{A})$ via $\mathcal{J}'$.  
       \item \label{item:sandwichthm6}  $\mathcal{J}$ is correct as a reduction from $\CSP(\struct{A})$ to  $\PCSP(\struct{S}_1,\struct{S}_2)$ when limited to instances $\struct{X}$ with the property that $\struct{X}$ homomorphically maps to $\struct{A}$ if and only if $\struct{X}$ homomorphically maps to a substructure of $\struct{A}$ of size $\leq n$. 
       \item \label{item:sandwichthm4} Every finite cheese for $\PCSP(\struct{S}_1,\struct{S}_2)$ pp-constructs $\struct{K}_3$; thus $\PCSP(\struct{S}_1,\struct{S}_2)$ is not finitely tractable.  
   \item \label{item:sandwichthm1} $\struct{A}$ pp-constructs $\struct{K}_3$ if and only if  $\struct{A}'$ pp-constructs $\struct{K}_3$.  
 \item \label{item:sandwichthm2} If $\struct{A}$ is a model-complete core, then $\struct{A}'$ is a model-complete core as well.
        \item \label{item:sandwichthm3} If $\struct{B}$ is Ramsey, then $\struct{B}'$ is Ramsey as well. 
   \end{enumerate} 
   
 \end{restatable}  

Item~\ref{item:sandwichthm6} is not necessary for obtaining non-finitely tractable PCSPs from Conjecture~\ref{conj:bodirsky_pinsker}, but provides an additional link that makes it possible to transfer logical inexpressibility results from infinite-domain CSPs to PCSPs, e.g., for FPC (see Section~\ref{section:inexpressibilitytransfer}).

By omitting the part of the proof of Theorem~\ref{thm:main_theorem_sandwiches} which relies on Theorem~\ref{thm:removing_algebraicity} (i.e.~removing algebraicity and adding a linear order), we additionally get that $\struct{A}$ and $\struct{A}'$ are pp-bi-interpretable (which leads to \emph{gadget reductions} between their CSPs) at the cost of dropping the requirement that the resulting PCSP is non-finitely tractable. However, this construction still generalizes the two concrete examples of non-finitely tractable PCSPs from~\cite{Mottet_2025}.

\begin{restatable}{theorem}{sandwichthmbasic}   \label{thm:main_theorem_sandwiches_basic}  
Let $\struct{A}$ be a reduct of a countable finitely bounded homogeneous structure $\struct{B}$.
Then there exists a reduct $\struct{A}'$ of a countable finitely bounded homogeneous structure $\struct{B}'$ such that 
$\struct{A}$ and $\struct{B}$ are pp-bi-interpretable with $\struct{A}'$ and $\struct{B}'$, respectively, and such that for all  $n\in \mathbb{N}$, there exists a finite PCSP template $(\struct{S}_1,\struct{S}_2)$ and gadget-reductions $\mathcal{J}$ and $\mathcal{J}'$ from $\CSP(\struct{A})$ to $\CSP(\struct{A}')$ and vice versa such that
\begin{enumerate} 
  \item \label{item:sandwichbasic1} $\struct{A}'$ is a cheese for $\PCSP(\struct{S}_1,\struct{S}_2)$;   thus, $\PCSP(\struct{S}_1,\struct{S}_2)$ reduces to $\CSP(\struct{A})$ via  $\mathcal{J}'$;
  \item \label{item:sandwichbasic2} $\mathcal{J}$ is correct as a reduction from $\CSP(\struct{A})$ to  $\PCSP(\struct{S}_1,\struct{S}_2)$ when limited to instances $\struct{X}$ with the property that $\struct{X}$ homomorphically maps to $\struct{A}$ if and only if it homomorphically maps to a substructure of $\struct{A}$ of domain size $\leq n$.  
\end{enumerate} 
   
\end{restatable}  

By a well-known result due to Barto, Kozik, Larose, Z\'{a}dori, Atserias, Bulatov, and Dawar, a finite-domain CSP is expressible in FPC if and only if it does not pp-construct any non-trivial finite Abelian group~\cite[Thm.~1.1]{bodirsky2022descriptive}.
This correspondence does not extend to the CSPs in the scope of Conjecture~\ref{conj:bodirsky_pinsker};
in fact, it already fails within \emph{temporal CSPs}\textemdash the CSPs of the first-order reducts of $(\mathbb{Q};<)$~\cite[Thm.~1.5]{bodirsky2022descriptive}. 
However, using Theorem~\ref{thm:main_theorem_sandwiches_basic} and the results in~\cite{bodirsky2022descriptive}, one can show that a temporal CSP is expressible in FPC if and only if it does not pp-construct any finite-domain PCSP inexpressible in FPC.
We remark that this fact can essentially already be derived from the results in~\cite{bodirsky2022descriptive,Mottet_2025}.
\begin{theorem}\label{thm:temporal_situation} A temporal CSP is expressible in FPC if and only if it does not pp-construct any finite-domain PCSP inexpressible in FPC.
\end{theorem}

\subsubsection{Phylogeny CSPs}
We initiate a new case study\textemdash of phylogeny CSPs\textemdash from the perspective of descriptive complexity.
Phylogeny CSPs are the CSPs of structures first-order definable in the (finitely bounded) homogeneous $C$-relation, and hence fall into the scope of Conjecture~\ref{conj:bodirsky_pinsker}.
As the name suggests, phylogeny CSPs can be used to model certain decision problems arising in evolutionary biology.
We show that similarly as in the temporal case, there exists a (tractable) phylogeny CSP which pp-constructs a finite-domain PCSP inexpressible in FPC but it does not pp-construct any such finite-domain CSP.
This provides further evidence that the 
complexity of the CSPs in the scope of Conjecture~\ref{conj:bodirsky_pinsker} is intimately connected to the complexity of finite-domain PCSPs.

\begin{restatable}{theorem}{tpppconstruct}   \label{thm:temporal_phylo_pp_constr}  
 There is a tractable phylogeny CSP that:
\begin{enumerate}  
    \item \label{item:finite_reason_1} pp-constructs a finite-domain PCSP inexpressible in FPC; 
    \item \label{item:finite_reason_2} does  not pp-construct any finite-domain CSP inexpressible in FPC.
\end{enumerate}
   
\end{restatable}

\subsection{Related work}
Our results in Theorem~\ref{thm:main_theorem_sandwiches} compare to the above-mentioned examples of tractable $\omega$-categorical cheeses for finite-domain PCSPs that are not finitely tractable from~\cite[Props.~35 and~36]{Mottet_2025} as follows. 

First, the PCSPs from~\cite[Thm.~2]{Mottet_2025} are reducible to the associated $\omega$-categorical templates under gadget reductions.
The main advantage of gadget reductions over general Datalog-reductions is that, in many settings, gadget reductions between (P)CSPs are captured by the transfer of height-1 identities between polymorphism clones (or more generally \emph{polymorphism minions}); see~\cite{dalmau2024local}. 
While Theorem~\ref{thm:main_theorem_sandwiches_basic} generalizes these cases, its construction fails to ensure that the resulting PCSP is not finitely tractable in general. To address this, Theorem~\ref{thm:main_theorem_sandwiches} employs Datalog-reductions rather than gadget-reductions. Notably, it can be shown that in many cases, the resulting finite-domain PCSPs do not admit a gadget reduction to the original $\omega$-categorical template (Example~\ref{ex:two_ex}). 

Second,~\cite{Mottet_2025} also contains a result (Theorem~1) showing that every CSP within the scope of the Bodirsky-Pinsker conjecture is polynomial-time equivalent to the PCSP of a half-infinite template, more precisely a template $(\struct{S}_1,\struct{S}_2)$ such that $\struct{S}_1$ is in the scope of the Bodirsky-Pinsker conjecture and $\struct{S}_2$ is finite.
In contrast, our Theorem~\ref{thm:main_theorem_sandwiches} provides non-finitely tractable finite-domain PCSP templates $(\struct{S}_1,\struct{S}_2)$ sandwiching a given structure $\struct{A}$ from Conjecture~\ref{conj:bodirsky_pinsker} up to Datalog interreducibility.
Since our PCSP templates $(\struct{S}_1,\struct{S}_2)$ are finite and not half-infinite, we cannot guarantee polynomial-time equivalence between $\PCSP(\struct{S}_1,\struct{S}_2)$ and $\CSP(\struct{A})$. 
However, we may still approximate $\CSP(\struct{A})$ to an arbitrary degree in the sense of Theorem~\ref{thm:main_theorem_sandwiches}(\ref{item:sandwichthm6}).

\subsection{Organization of this paper}
In Section~\ref{section:preliminaries}, we provide definitions and general notation. In Section~\ref{section:algebraicity_irrelevant}, we show that one can without loss of generality assume that the templates in Conjecture~\ref{conj:bodirsky_pinsker} have no algebraicity and essentially injective polymorphisms by providing proofs of Theorem~\ref{thm:removing_algebraicity}, Theorem~\ref{thm:polinjective}, and Corollary~\ref{cor:binaryInj}. Theorem~\ref{thm:removing_algebraicity} is then used in Section~\ref{section:sandwiches} to construct sandwiches of non-finitely tractable finite-domain PCSPs from all structures in the scope of the aforementioned conjecture via Theorem~\ref{thm:main_theorem_sandwiches}.
At the end of Section~\ref{section:sandwiches}, we also provide a full proof of Theorem~\ref{thm:temporal_situation}.
Finally, in Section~\ref{section:phylo}, we investigate the class of phylogeny CSPs from the perspective of the descriptive complexity. We provide an example of a phylogeny problem that is not definable in FPC and show that this inexpressibility result transfers to the finite-domain PCSP constructed from this template via Theorem~\ref{thm:main_theorem_sandwiches_basic}, ultimately proving Theorem~\ref{thm:temporal_phylo_pp_constr}. 
At the end of Section~\ref{section:phylo} and in Section~\ref{section:conclusion} we moreover use the opportunity to pose new open questions which we believe need to be addressed in order to assure continued progress on Conjecture~\ref{conj:bodirsky_pinsker}. Section~\ref{section:conclusion}  also provides a more detailed discussion on specific aspects of our results.

\section{Preliminaries} \label{section:preliminaries} 

The set $\{1,\dots,n\}$ is denoted by $[n]$, and we use the bar notation $\bar{t}$ for tuples.
The \emph{component-wise action} of a function $f\colon A^n \rightarrow B$ on $k$-tuples is given by \[f\big((x_{1,1},\dots, x_{1,k}),\dots, (x_{n,1},\dots, x_{n,k}) \big)\coloneqq \big(f(x_{1,1},\dots,x_{n,1}),\dots, f(x_{1,k},\dots,x_{n,k}) \big).\] 

\subsection{Relational structures} 
A (\emph{relational}) \emph{signature} $\tau$ is a set of \emph{relation symbols}, each $R\in\tau$ with an associated natural number called \emph{arity}. 
A (\emph{relational}) \emph{$\tau$-structure} $\struct{A}$ consists of a set $A$ (the \emph{domain}) together with the relations $R^{\struct{A}}\subseteq A^{k}$ for each $R\in \tau$ with arity $k$.
An \emph{expansion} of $\struct{A}$ is a $\sigma$-structure $ \struct{B}$ with $A=B$ such that $ \tau\subseteq \sigma$ and $R^{\struct{B}}=R^{\struct{A}}$ for each relation symbol $R\in \tau$. Conversely, we call $\struct{A}$ a \emph{reduct} of $\struct{B}$ and denote it by $\struct{B}|^{\tau}$.   
We denote by $(\struct{A},R)$ the expansion of a structure $\struct{A}$ by a relation $R$ over its domain. 
The \emph{substructure} of a $\tau$-structure $\struct{A}$ on a subset $B\subseteq A$ is the $\tau$-structure $\struct{B}$ with domain $B$ and relations $R^{\struct{B}}=R^{\struct{A}}\cap B^k$ for every $R\in \tau$ of arity $k$; we use the notation $\struct{A}[B]\coloneqq \struct{B}$.

The \emph{factor} of a $\tau$-structure $\struct{A}$ through an equivalence relation $E\subseteq A^2$ is the $\tau$-structure $\struct{A}/{E}$ with domain $A/E$ and relations $R^{\struct{A}/{E}}= q_E(R^{\struct{A}})$, where $q_E$ denotes the factor map $x\mapsto [x]_E$. %
Let $\sigma \subset \tau$ be arbitrary. 
We say that $E$ is a \emph{relational congruence} on $\struct{A}$ if the definitions of the corresponding relations of $\struct{A}/E$ do not depend on the choices of the representatives of the equivalence classes of $E$:  for every $R\in \tau$ and all tuples $\bar{t}$ we have $q_E(\bar{t})\in R^{\struct{A}/{E}}$ if and only if $\bar{t}\in R^{\struct{A}}$.

For a positive integer $d$ and a $\tau$-structure $\struct{A}$, the \emph{$d$-th power} of $\struct{A}$ is the $\tau$-structure $\struct{A}^d$ with the domain $A^d$ and relations as follows. For each $R\in \tau$ of arity $k$:
\begin{align*}
    R^{\struct{A}^d}=\{\big((a_{1,1},\dots, a_{1,d}),\dots, & (a_{k,1},\dots, a_{k,d})\big)\in (A^d)^k \\ & \mid  (a_{1,1},\dots, a_{k,1}),\dots, (a_{1,d},\dots, a_{k,d}) \in R^{\struct{A}}\}
\end{align*}

  A \emph{homomorphism} $h\colon \struct{A} \rightarrow \struct{B}$ for $\tau$-structures $\struct{A}$ and $\struct{B}$ is a mapping $h\colon  A\rightarrow B$ that \emph{preserves} each $\tau$-relation, i.e.,~if $ \bar{t} \in R^{\struct{A}}$ for some  relation symbol $R\in \tau$, then $h(\bar{t})\in R^{\struct{B}}$.
We write $\struct{A} \rightarrow \struct{B}$ if $\struct{A}$ maps homomorphically to $\struct{B}$; formally, 
\[\CSP(\struct{A}) \coloneqq \{\struct{X} \text{ finite} \mid \struct{X} \rightarrow \struct{A}\}.\]
 If $\struct{A} \rightarrow \struct{B}$ and $\struct{B} \rightarrow \struct{A}$, then we say that $\struct{A}$ and $\struct{B}$ are \emph{homomorphically equivalent}. 
 An \emph{endomorphism} of $\struct{A}$ is a homomorphism from $\struct{A}$ to itself; we denote by $\End(\struct{A})$ the set (monoid) of all endomorphisms of $\struct{A}$.  
 We say that $\struct{A}$ is \emph{contractible}~\cite{bodirsky2010categorical} if it has a constant endomorphism; otherwise \emph{non-contractible}. 
Clearly, $\struct{A}$ is contractible if and only if its CSP is the CSP of a one-element structure, and hence, solved by an algorithm that only checks whether certain fixed relations are empty in an instance.

An \emph{embedding} is an injective homomorphism $h\colon \struct{A} \rightarrow \struct{B}$ that additionally satisfies the following condition: for every $k$-ary relation symbol $R\in \tau$ and $\bar{t}\in A^{k}$ we have $h(\bar{t})\in R^{\struct{B}}$ only if $\bar{t}\in R^{\struct{A}}.$ 
The \emph{age} of $\struct{A}$, denoted by $\age(\struct{A})$, is the class of all finite structures which embed into $\struct{A}$.  
 An \emph{isomorphism} is a surjective embedding. Two structures $\struct{A}$ and $\struct{B}$ are \emph{isomorphic} if there exists an isomorphism from $\struct{A} $ to $\struct{B}$.
 A \emph{partial isomorphism} between two structures $\struct{A}$ and $\struct{B}$ is an isomorphism from a substructure of $\struct{A}$ to a 
substructure of $\struct{B}$; if $\struct{A}=\struct{B}$, then we speak of a partial isomorphism on $\struct{A}$.
An \emph{automorphism} is an isomorphism from $\struct{A}$ to $\struct{A}$; we denote by $\Aut(\struct{A})$ the set (group) of all automorphisms of $\struct{A}$.
The \emph{orbit} of a tuple $\bar{t}\in A^{k}$ in $\struct{A}$ is the set $\{g(\bar{t}) \mid g \in \Aut(\struct{A})\}$.  Any tuples belonging to the same orbit satisfy precisely the same first-order formulas over $\struct{A}$.  
The \emph{orbit equivalence} relation provides a natural way to factorise relational structures: 
given a relational structure $\struct{A}$ and a subgroup $G\subseteq \Aut(\struct{A})$,
we denote by $\orbeq{\struct{A}}{G}$ the structure $\struct{A}/E$ for \[E\coloneqq \{(x,y)\in A^2 \mid \exists g\in G\colon g(x)=y\}.\]

We say that a structure $\struct{A}$ has \emph{no algebraicity} (in the group-theoretic sense) if, for every $k\geq 1$ and every $\bar{a}\in A^k$, the automorphisms of $\struct{A}$ which stabilize $\bar{a}$ do not stabilize any $a'\notin \bar{a}$  (by this notation we mean that $a'$ does not appear as an entry of $\bar{a}$).
In $\omega$-categorical structures this is the case if and only if for every $k\geq 1$ and every tuple $\bar{a}\in A^k$  it is not possible to first-order define any finite set of elements outside $\bar{a}$ using $\bar{a}$ as parameters. In other words, any non-trivial first-order property  relative to $\bar{a}$ is always satisfied by infinitely many elements. %

Recall from the introduction that an $\omega$-categorical structure $\struct{A}$ is a model-complete core if the endomorphisms of $\struct{A}$ preserve all orbits of $\Aut(\struct{A})$.
A reformulation of this definition that will be more convenient for us is that for every $e\in \End(\struct{A})$ and every finite $F\subseteq A$, there exists $\alpha \in \Aut(\struct{A})$ such that $e|_F=\alpha|_F$ (see e.g.~\cite{bodirsky2021complexity}).

\subsection{The KPT correspondence}\label{sec:kpt}

For structures $\struct{S}$ and $\struct{M}$, we denote by $\smash{\binom{\struct{M}}{\struct{S}}}$ the set of all embeddings of $\struct{S}$ into $\struct{M}$.
A class $\mathcal{K}$ of structures over a common signature $\tau$ has the \emph{Ramsey property} (RP) if,  for all $\struct{S},\struct{M}\in \mathcal{K}$ and $k\in \mathbb{N}$, there exists $\struct{L}\in \mathcal{K}$ such that, for every map $f\colon \binom{\struct{L}}{\struct{S}} \rightarrow [k]$, there exists $e\in \smash{\binom{\struct{L}}{\struct{M}}}$ such that $f$ is constant on the set 
\[\{e\circ u\;|\; u\in \textstyle\binom{\struct{M}}{\struct{S}}\} \subseteq \textstyle\binom{\struct{L}}{\struct{S}}.\]
It is a folklore fact that considering the special case $k=2$ is enough for verifying the RP~\cite{hubickanesetril2019}, and we call any  $\struct{L}$ satisfying the above in this case a \emph{Ramsey witness} for $(\struct{S},\struct{M})$.  
A homogeneous structure is \emph{Ramsey} if its age has the Ramsey property. 

In the proof of one of our auxiliary results (Proposition~\ref{prop:basic_properties_wreath_products}), we use the Kechris-Pestov-Todor\v{c}evi\'c (KPT) correspondence, which links the Ramsey property for homogeneous structures over countable signatures to a  property of their automorphism group viewed as a topological group. 
A \emph{topological group} is a pair $(G,\mathcal{T})$, where $G$ is a group and $(G;\mathcal{T})$ a topological space, such that both taking inverses and composition in the group $G$ are continuous with respect to $\mathcal{T}$. 
A topological group $(G,\mathcal{T})$ is called \emph{extremely amenable} if every continuous action $\alpha$ of $(G,\mathcal{T})$ on a  compact Hausdorff topological space $(X,\mathcal{T}')$ has a fixed point, i.e.~there exists $x\in X$ such that $\alpha(g)(x)=x$ for every $g\in G$; the action $\alpha$ being continuous means that  the map $(g,x)\mapsto \alpha(g)(x)$ is.
A bijective map $f$ between two topological groups is a \emph{topological isomorphism} if it is a group isomorphism and  both $f$ and $f^{-1}$ are continuous.
Clearly, topological group isomorphisms preserve extreme amenability. 
The formulation of the KPT correspondence below (Theorem~\ref{thm:kpt}) refers to the extreme amenability of $\Aut(\struct{B})$, without mentioning any specific topology: this is because it is understood to be equipped with the \emph{topology of pointwise convergence}.
This is the smallest topology on $\Aut(\struct{B})$ containing all sets of the form $\{g\in \Aut(\struct{B}) \mid g(\bar{t}) =\bar{s}\}$ for some tuples $\bar{t},\bar{s}$ over $B$, the \emph{basic open sets} of the topology.
\begin{theorem}[Kechris, Pestov, and Todorcevic~\cite{kechris2005fraisse}] \label{thm:kpt} Let $\struct{B}$ be a countable homogeneous %
structure over a countable  relational signature. 
Then $\struct{B}$ is Ramsey if and only if  $\Aut(\struct{B})$ is extremely amenable.
\end{theorem}  

\subsection{Universal algebra}
An (\emph{equational}) \emph{condition} is a set of \emph{identities}, i.e.~formal
expressions of the form $s \approx t$, where $s$ and $t$ are terms over
a common set of function symbols. 
An equational condition is \emph{height-1} if it contains neither nested terms nor terms consisting of a single variable.
A \emph{polymorphism} of a relational structure $\struct{A}$ is a  homomorphism from $\struct{A}^k$ into $\struct{A}$ for some $k\in \mathbb{N}$; the set of all polymorphisms of  $\struct{A}$ is denoted by $\Pol(\struct{A})$.
We say that $\Pol(\struct{A})$ (or some set $S$ of operations on %
$A$) \emph{satisfies} an %
equational condition $\Sigma$ if the function symbols can be interpreted as elements of $\Pol(\struct{A})$ (of $S$) so that, for each identity $s \approx t$ in $\Sigma$, the equality $s = t$ holds for any evaluation of variables in $A$. 
An example is the \emph{cyclic identity} of arity $n\geq 2$ given by \[f(x_1,\dots, x_n) \approx f(x_2,\dots,x_n, x_1),\] for a symbol $f$ of arity $n\geq 2$. 
 An operation is called \emph{cyclic} if it satisfies the cyclic identity.

\begin{theorem}[\cite{Barto_2012}]\label{thm:barto12}Let $\struct{A}$ be a finite relational structure. If $\struct{A}$ does not pp-construct $\struct{K}_3$, then $\Pol(\struct{A})$ contains a cyclic operation.
\end{theorem}
 
A \emph{pseudo-version} of a height-1 identity $s \approx t$ is of the form $\alpha \circ s \approx \beta \circ t$ for fresh unary function symbols $\alpha$ and $\beta$.
An example is the pseudo-Siggers identity from Conjecture~\ref{conj:bodirsky_pinsker}. 
Allowing for different unary symbols in different identities, this definition naturally extends to height-1 conditions.

We extend the notion of a polymorphism to PCSP templates $(\struct{S}_1,\struct{S}_2)$ in the usual way:
a \emph{polymorphism} of $(\struct{S}_1,\struct{S}_2)$  is a homomorphism from a finite power of $\struct{S}_1$ to $\struct{S}_2$;
we denote by $\Pol(\struct{S}_1,\struct{S}_2)$ the set (minion) of all such polymorphisms. 
Also the satisfaction of height-1 conditions can be generalized to the PCSP setting in a straightforward manner.
The same is not true for nested identities, because the composition of polymorphisms between different structures is not defined.

We say that a height-1 condition $\Sigma$ is \emph{satisfied in  $\Pol(\struct{S}_1,\struct{S}_2)$ modulo a set $H$ of unary operations on $S_2$} if the function symbols in $\Sigma$ can be interpreted as elements of $\Pol(\struct{S}_1,\struct{S}_2)$ such that for each identity $s \approx t$ in $\Sigma$ there exist $\alpha,\beta\in H$ such that $\alpha\circ s =\beta \circ t$ holds for any evaluation of variables in $S_1$.

\subsection{First-order logic} \label{section:fo}
We say that a first-order formula is \emph{$k$-ary} if it has $k$ free variables.
For a first-order formula $\phi$, we use the notation $\phi(\bar{x})$ to indicate that the free variables of $\phi$ are among $\bar{x}$.
This does not mean that the truth value of $\phi$ depends on each entry in $\bar{x}$.
We assume that equality $=$ as well as the nullary predicate symbol $\bot$  for falsity are always available when building first-order formulas. 
Thus, \emph{atomic $\tau$-formulas} over a relational signature $\tau$ are of the form $\bot$, $(x=y)$, and $R(\bar{x})$ for some $R\in \tau$ and a tuple $\bar{x}$ of first-order variables matching the arity of $R$.
A first-order $\tau$-formula $\phi$ is \emph{primitive positive} (pp) if it is of the form $\exists x_{1},\dots,x_{m}  (\phi_{1}\wedge \dots \wedge \phi_{n})$, where each $
\phi_{i}$ is atomic.
Let $\struct{A}$ and $\struct{B}$ be relational structures and $d$ a natural number.
A \emph{$d$-dimensional First-Order} (FO) \emph{interpretation} of $\struct{B}$ in $\struct{A}$ is a partial surjection $\mathcal{I}\colon A^d \rightarrow B$ with the property that, for every relation $R$ defined by a $k$-ary atomic formula in $\struct{B}$, the $dk$-ary (preimage) relation 
\[
\mathcal{I}^{-1}(R) \coloneqq \bigl\{(a_1^1,\dots, a_d^1, \dots, a_1^k,\dots, a_d^k) \mid \big(\mathcal{I}(a_1^1,\dots, a_d^1),\dots, \mathcal{I}(a_1^k,\dots, a_d^k)\big) \in R \bigr\}
\]
has a first-order definition in $\struct{A}$.
Since the equality is allowed as an atomic formula, there must in particular be first-order formulas defining the domain and the kernel of $\mathcal{I}$.    

If the defining formulas of an interpretation can be chosen primitive positive, then we call it a \emph{pp-interpretation}.  
We say that $\struct{A}$ and $\struct{B}$ are \emph{first-order bi-interpretable} (\emph{pp-bi-interpretable}) if there exist first-order interpretations (pp-interpretations) $\mathcal{I}_{\struct{B}}$ and $\mathcal{I}_{\struct{A}}$ of $\struct{B}$ in $\struct{A}$ and vice versa such that the relations 
\[\{(\bar{x},y) \mid \mathcal{I}_{\struct{A}}\circ \mathcal{I}_{\struct{B}}(\bar{x})= y \} \qquad  \text{and} \qquad \{(\bar{x},y) \mid \mathcal{I}_{\struct{B}} \circ \mathcal{I}_{\struct{A}}(\bar{x}) = y \}\] 
are first-order definable (pp-definable) in $\struct{A}$ and $\struct{B}$, respectively.
Two countable $\omega$-categorical structures are first-order bi-interpretable if and only if their automorphism groups are isomorphic as topological groups~\cite{AHLBRANDT198663} and they are pp-bi-interpretable if and only if their polymorphism clones are isomorphic as topological clones~\cite{Topo-Birk}.

The property of being a model-complete core for an $\omega$-categorical structure is encoded in the polymorphism clone (the closure of the automorphism group under the topology of pointwise convergence is the endomorphism monoid).
Consequently, if $\struct{A}$ and $\struct{B}$ are $\omega$-categorical and pp-bi-interpretable, then $\struct{A}$ is a model-complete core if and only if $\struct{B}$ is. 
Moreover, by Theorem~\ref{thm:kpt}, first-order bi-interpretations between homogeneous structures preserve the Ramsey property. 
 
A \emph{pp-construction} between CSP templates is a pp-interpretation composed
with homomorphic equivalence~\cite{barto2018wonderland}: $\struct{A}$ pp-constructs $\struct{B}$ if $\struct{A}$ pp-interprets a structure $\struct{A}'$ which is homomorphically equivalent to $\struct{B}$; this notion naturally extends to PCSP templates.
For a full definition, we refer the reader to~\cite{barto2021algebraic}; in the present article, we need the following characterization: a structure $\struct{A}$ pp-constructs the template $(\struct{S}_1,\struct{S}_2)$ if
there exists a structure $\struct{A}'$ such that $\struct{A}$ pp-constructs $\struct{A}'$ and $\struct{S}_1 \rightarrow \struct{A}' \rightarrow \struct{S}_2$.

 \subsection{Fixed-Point logic}\label{section:fp}
 \emph{Fixed-Point logic} (FP) is defined by adding formation rules to first-order logic whose semantics is defined with inflationary fixed-points of definable operators:
every first-order formula is an FP formula and, if $\phi(\bar{x})$ is an FP formula over some relational signature $\tau\cup \{R\}$ with $R\notin \tau$, then 
$[\mathsf{ifp}_{R}\,\phi](\bar{x})$ is an FP formula over the signature $\tau$ whose semantics is given as follows.
For a $\tau$-structure $\struct{A}$ and a tuple $\bar{a}$ over $A$ matching the arity of $\bar{x}$, say $k$, we have $\struct{A}\models [\mathsf{ifp}_{R}\, \phi](\bar{a})$ if and only if $\bar{a}$ is contained in the inflationary fixed-point of the operator 
\[
F^{\struct{A}}_{\phi,R}(X) \coloneqq \{\bar{x}\in A^k \mid  (\struct{A},X)\models \phi(\bar{x})\},
\] 
i.e., the limit of the sequence $X_0\coloneqq \emptyset$ and $X_{i+1} \coloneqq X_i \cup F^{\struct{A}}_{\phi,R}(X_i)$. 
For example, the formula \[[\mathsf{ifp}_T\, E(x,z) \vee (\exists y \ldotp T(x,y) \wedge T(y,z) )](x,z)\] computes the transitive closure of $E$. 

The logic programming language~\emph{Datalog} can be viewed as the existential positive fragment of FP.
It is not hard to see that every Datalog formula is specified by a finite set of rules of the form $R(\bar{x})\Leftarrow R_1(\bar{x}_1)\wedge \cdots \wedge R_m(\bar{x}_m)$ where $R$ is a fixed-point variable.
In the case of the above Datalog formula computing the transitive closure of $E$, for example,  the rules are 
$T(x,z) \Leftarrow E(x,z)$ and $T(x,z) \Leftarrow T(x,y) \wedge T(y,z)$. 

\emph{Fixed-Point logic with Counting} FPC is obtained from FP by adding a mechanism for counting the members of definable queries.
Here, we only need to know  the following well-known game-theoretic invariant.
The \emph{bijective $k$-pebble-game} is played by two players, the \emph{Spoiler} and the \emph{Duplicator}, on a pair of structures $\struct{A}$ and $\struct{B}$.
The players have a set of pairs of pebbles $\{(p_1^{\struct{A}},p_1^{\struct{B}}),\dots, (p_k^{\struct{A}},p_k^{\struct{B}})\}$.
In each move:

\smallskip
\begin{enumerate} 
    \item The Spoiler chooses $i\in [k]$. 
     \item The Duplicator selects a bijection $f\colon A \rightarrow B$.
    \item The Spoiler places $\smash{p_i^{\struct{A}}}$ on any element $a\in A$ and $\smash{p_i^{\struct{B}}}$ on $f(a)$. 
\end{enumerate}  
\smallskip

We write $\struct{A}  \equiv_{k} \struct{B}$ if the Duplicator can ensure for arbitrarily many rounds that $p_i^{\struct{A}} \mapsto p_i^{\struct{B}}$ specifies a partial isomorphism between the pebbled elements. 

\begin{theorem}[Immerman and Lander~\cite{otto2017bounded}] \label{thm:pebble_games}
For every FPC sentence $\phi$, there exists $k\in \mathbb{N}$ such that, for all finite structures $\struct{A}$ and $ \struct{B}$, we have that  $\struct{A} \equiv_{k} \struct{B}$ implies $\struct{A} \models \phi \Leftrightarrow \struct{B} \models \phi$.  
\end{theorem}

\subsection{Descriptive complexity of (P)CSPs}

For a relational $\tau$-structure $\struct{A}$ and a logic $\mathcal{L}$ (in the sense of Gurevich~\cite{gurevich1988logic,grohe2008quest}), we say that $\CSP(\struct{A})$ is \emph{solvable} or \emph{expressible in $\mathcal{L}$} if there exists an  $\mathcal{L}$-sentence defining the class of all finite $\tau$-structures which do \emph{not} homomorphically map to $\struct{A}$.
Here, the extra negation is included to account for certain logics which are not closed under negation, e.g.~Datalog.
More generally, we say that $\PCSP(\struct{S}_1,\struct{S}_2)$ is solvable in  $\mathcal{L}$ if there exists an $\mathcal{L}$-sentence $\Phi$ that is true in all negative instances, and false in all positive instances.
A typical strategy to show that $\PCSP(\struct{S}_1,\struct{S}_2)$
is not solvable in FPC is to use Theorem~\ref{thm:pebble_games}:  it suffices to find, for all $k\in \mathbb{N}$,
a YES-instance $\struct{A}_k$ and a NO-instance $\struct{B}_k$ such that $\struct{A}_k \equiv_{k} \struct{B}_k$.
We call such families \emph{fooling instances}.
If, in addition, there is a finite substructure $\struct{S}$ of $\struct{S}_1$ such that for all $k$, $\struct{A}_k$ maps homomorphically to $\struct{S}$, we call the family \emph{fixed-image fooling}.

\begin{example}
    $\CSP(\mathbb{Q};<)$ is solvable in Datalog, as witnessed by the sentence \[\exists u [\mathsf{ifp}_T\, (x<z) \vee (\exists y \ldotp T(x,y) \wedge T(y,z) )](u,u).\]
\end{example}

First-order interpretations defined in Section~\ref{section:fo} naturally generalize to the notion of an $\mathcal{L}$-interpretation for arbitrary logics $\mathcal{L}$ by allowing $\mathcal{L}$-formulas for defining the preimages of the relations. 
Let $\tau$ and $\tau'$ be two finite relational signatures and $\mathcal{R}$ a tuple consisting of a $dk$-ary $\tau$-formula $\phi_{\alpha}$ in some logic $\mathcal{L}$ for every $k$-ary atomic $\tau'$-formula $\alpha$.  
Let $\textnormal{dom}(\mathcal{R})\coloneqq \{\bar{x}\in X^d:\phi_{x=x}(\bar{x})\}$ and $\sim_{\mathcal{R}}\coloneqq\{(\bar{x},\bar{y})\in(X^d)^2:\phi_{x=y}(\bar{x},\bar{y})\}$. Then we can assign to any $\tau$-structure $\struct{X}$ a unique $\tau'$-structure $\mathcal{R}(\struct{X})$ on $\textnormal{dom}(\mathcal{R})/\sim_{\mathcal{R}}$ with an $\mathcal{L}$-interpretation in $\struct{X}$ as given by the formulas in $\mathcal R$.
For a $\tau$-structure $\struct{A}$ and $\tau'$-structure $\struct{A}'$, we say that $\mathcal{R}$ is an \emph{$\mathcal{L}$-reduction} from $\CSP(\struct{A})$ to $\CSP(\struct{A}')$ if 
\[
\struct{X} \rightarrow \struct{A} \quad \text{if and only if} \quad \mathcal{R}(\struct{X}) \rightarrow \struct{A}'.
\] 
Logical reductions compose whenever the associated interpretations do.
For any logic $\mathcal{L}$ closed under conjunction of formulas and existential quantification over first-order variables,  $\mathcal{L}$-reductions preserve the expressibility of CSPs in $\mathcal{L}$ and other logics extending $\mathcal{L}$.
Some authors require the kernel of the associated interpretation to be trivial (cf.~\cite[Def.~1]{atserias2009affine} and~\cite[Sec.~3.1]{dalmau2024local}). We choose not to impose this restriction, since this is only a convention stemming from the fact that factorization is not necessary in most (if not all) 
use cases of logical reductions.
For example, any pp-interpretation of a structure $\struct{A}'$ in $\struct{A}$ yields a factor-free Datalog-reduction from $\CSP(\struct{A}')$ to $\CSP(\struct{A})$~\cite{atserias2009affine}.\footnote{Datalog reductions in~\cite{atserias2009affine} are additionally allowed to have finitely many exceptions and use finitely many parameters in the defining formulas.}
Let $(\struct{A}_1,\struct{A}_2)$ and $(\struct{A}_1',\struct{A}'_2)$ be PCSP templates over the signatures $\tau$ and $\tau'$, respectively.
We say that $\mathcal{R}$ is an \emph{$\mathcal{L}$-reduction} from $\PCSP(\struct{A}_1,\struct{A}_2)$ to $\PCSP(\struct{A}_1',\struct{A}'_2)$ if 
\[
\struct{X} \rightarrow \struct{A}_1 \text{ implies } \mathcal{R}(\struct{X}) \rightarrow \struct{A}'_1 \quad\text{and}\quad \struct{X} \nrightarrow \struct{A}_2 \text{ implies } \mathcal{R}(\struct{X}) \nrightarrow \struct{A}'_2.
\] 
More generally, we say that  $\PCSP(\struct{A}_1,\struct{A}_2)$ \emph{$\mathcal{L}$-reduces} to $\PCSP(\struct{A}_1',\struct{A}'_2)$ if there exists an $\mathcal{L}$-reduction $\mathcal{R}$ from $\PCSP(\struct{A}_1,\struct{A}_2)$ to $\PCSP(\struct{A}_1',\struct{A}'_2)$.

\section{Simplifying the Bodirsky-Pinsker conjecture} \label{section:algebraicity_irrelevant}

In the present section, we give proofs of Theorem~\ref{thm:removing_algebraicity}, Theorem~\ref{thm:polinjective}, and Corollary~\ref{cor:binaryInj}.

\subsection{Wreath products} \label{section:wreath}

Given groups $G$ and $H$ acting on sets $A$ and $B$, respectively, their \emph{wreath product} $G\wr H$ is given by $G^B\times H$ acting on $A\times B$ via \[\left((g_b)_{b\in B},h)\right)(a,b)\coloneqq (g_{h(b)}(a),h(b)).\]
The wreath product $G\wr H$ again forms a group, with the neutral element 
$((e_G)_{b\in B},e_H)$ and the group operation 
\[((g_b)_{b\in B},h))\cdot ((g'_b)_{b\in B},h'))\coloneqq ((g_b\cdot g'_{h^{-1}(b)})_{b\in B},h\cdot h')).\]
Note that the permutation group on  $A\times B$ induced by the component-wise action of the Cartesian product $G\times H$ is contained by that induced by the action of $G\wr H$ thereon: the permutations induced stemming from the former are represented in $G\wr H$ by all elements of the form $((g)_{b\in B},h)$.  

Below we give an important construction on structures (here also called wreath product) which corresponds to group-theoretic wreath products in terms of their automorphism groups. 
Let $\struct{A}$ and $\struct{B}$ be structures over disjoint relational signatures $\tau$ and $\sigma$, and let $E$ be a fresh binary symbol. 
The \emph{wreath product} $\struct{A}\wr\struct{B}$ is the structure over the signature $\tau\cup \{E\}\cup \sigma$ with domain $A\times B$ whose relations are defined as follows.
First, we set \[E^{\struct{A}\wr\struct{B}}\coloneqq \{((x_1,y_1),(x_2,y_2))\in (A\times B)^2 \mid y_1=y_2\}.\]
Then, for $R\in \tau$ and $S\in \sigma$ of arities $k$ and $m$: 
\begin{align*}
    R^{\struct{A}\wr\struct{B}}\coloneqq \ & \{((x_1,y_1),\dots, (x_k,y_k))\in (A\times B)^k \mid (x_1,\dots, x_k)\in R^{\struct{A}}\text{ and } y_1=\cdots=y_n\},  \\ 
S^{\struct{A}\wr\struct{B}} \coloneqq \ & \{((x_1,y_1),\dots, (x_m,y_m))\in (A\times B)^m \mid (y_1,\dots, y_m)\in S^{\struct{B}}\}. 
\end{align*} 

\begin{restatable}[Basic properties of wreath products]{proposition}{wreathproducts} \label{prop:basic_properties_wreath_products}Let $\struct{A}$ and $\struct{B}$ be structures over disjoint finite relational signatures $\tau$ and $\sigma$. Then:
    \begin{enumerate}
        \item \label{item:aut_wreath}  $\Aut(\struct{A} \wr \struct{B})=\Aut(\struct{A}) \wr \Aut(\struct{B})$.
        \item \label{item:algebraicity}  If $\struct{A}$ has no algebraicity, then $\struct{A}\wr\struct{B}$ has no algebraicity.
        \item \label{item:aut_wreath_omega_cat}  If $\struct{A}$ and $\struct{B}$ are $\omega$-categorical, then $\struct{A}\wr\struct{B}$ is $\omega$-categorical. 
        \item \label{item:wreath_homogeneous} If $\struct{A}$ and $\struct{B}$ are homogeneous, then $\struct{A}\wr\struct{B}$ is homogeneous.
        \item \label{item:finite_boundedness} If $\struct{A}$ and $\struct{B}$ are finitely bounded, then $\struct{A}\wr\struct{B}$ is finitely bounded.
        \item \label{item:Ramsey} If $\struct{A}$ and $\struct{B}$ are homogeneous Ramsey, then $\struct{A}\wr\struct{B}$ is homogeneous Ramsey.
    \end{enumerate} 

 \end{restatable}  
 
Proofs of items~\ref{item:wreath_homogeneous},~\ref{item:finite_boundedness}, and~\ref{item:Ramsey} can be found in~\cite[Lem.~2.2.49]{bodor_thesis} (see also~\cite[Lem.~2.33]{bodor2024classification}). The statement in item~\ref{item:Ramsey} can also be found in~\cite[Thm.~5.13]{scow2021ramsey}.
We include a full proof of Proposition~\ref{prop:basic_properties_wreath_products} for the convenience of the reader.
\begin{proof} For~\eqref{item:aut_wreath}, we first show that $\Aut(\struct{A}) \wr \Aut(\struct{B})$ is a subgroup of $\Aut(\struct{A}\wr\struct{B})$.
    Clearly, the wreath product is a subgroup of $\Sym(A\times B)$ and preserves $E^{\struct{A}\wr\struct{B}}$.
    Moreover, as every group element $\left((g_b)_{b\in B},h\right)$ simply acts like $h$ in the second component, all $\sigma$-relations are preserved.
    Concerning $\tau$-relations, note that all pairs occurring within a tuple $((a_1,b_1),\dots,(a_k,b_k))\in R^{\struct{A}\wr\struct{B}}$ share the same second entry $b=b_1=\cdots =b_k$.
    Therefore, the same $g_{h(b)}\in\Aut(\struct{A}))$ is picked to act on all the first entries, and the image $((g_{h(b)}(a_1),h(b)),\dots,(g_{h(b)}(a_k),h(b)))$ lies within $R^{\struct{A}\wr\struct{B}}$.
    
    To prove the equality of the two groups, it remains to show that every $w\in\Aut(\struct{A}\wr\struct{B})$ can be written in the above form.
    As $w$ acts on the equivalence classes of $E^{\struct{A}\wr\struct{B}}$, it induces an automorphism $h'$ on the $\sigma$-reduct of $(\struct{A}\wr\struct{B})/E^{\struct{A}\wr\struct{B}}$, which is isomorphic to $\struct{B}$.
    Further, for any $b\in B$, $w$ also induces an automorphism $g'$ on the $\tau$-reduct of the substructure induced by $A\times \{b\}$, which is in turn isomorphic to $\struct{A}$. 
    Denoting the corresponding automorphisms of $\struct{A}$ and $\struct{B}$ by $g_{h(b)},b\in B$ and $h$, it is now easy to check that $((g_b)_{b\in B},h)$ and $w$ coincide.

    For~\eqref{item:algebraicity},  we must show that fixing any tuple $\bar{t}$ of parameters does not yield any nontrivial fixed point of the automorphism group when stabilizing $\bar{t}$. Let $(a',b')$ be any element outside $\bar{t}$. Let $A'$ be the set of all elements $a\in A$ such that $(a,b)$ appears in $\bar{t}$ for some $b\in B$. Move $a'$ to a different element via some $g\in\Aut(\struct{A})$   whilst fixing $A'$; this is possible because $\struct{A}$ has no algebraicity. Then $\left((g_b)_{b\in B},id\right)$ with $g_{b}\coloneqq g$ if $b=b'$ and $id$ otherwise is an automorphism of $\Aut(\struct{A}\wr\struct{B})$ fixing every entry in $\bar{t}$ and moving $(a',b')$. 

For~\eqref{item:aut_wreath_omega_cat}, recall that, by item~\eqref{item:aut_wreath}, we have $\Aut(\struct{A} \wr \struct{B})=\Aut(\struct{A}) \wr \Aut(\struct{B})$.
But note that $\Aut(\struct{A}) \wr \Aut(\struct{B})$ in particular contains the Cartesian product $\Aut(\struct{A})\times \Aut(\struct{B})$. 
Since the Cartesian product already witnesses that there are only finitely many orbits of tuples of each arity, the statement follows. 

For~\eqref{item:wreath_homogeneous}, let $\bar{t}$ and $\bar{s}$ be tuples of the same length such that the map sending $\bar{t}$ to $\bar{s}$ is a partial isomorphism. 
By reordering, we split $\bar{t}$ into subtuples $(\bar{t}_{b_1},\dots,\bar{t}_{b_n})$ for distinct elements $b_i\in B$ such that the second  projection of all entries of $\bar{t}_{b_i}$ is  $b_i$. Similarly, we obtain $(\bar{s}_{c_1},\dots,\bar{s}_{c_n})$. Pick $g_{c_i}\in\Aut(\struct{A})$  sending the first projections of $\bar{t}_{b_i}$ to those of $\bar{s}_{c_i}$. Pick $h\in\Aut(\struct{B})$ sending $b_i$ to  $c_i$. The automorphism $((g_b)_{b\in B}),h)$, where $g_b=id$ if hitherto undefined, sends $\bar{t}$ to $\bar{s}$.

For~\eqref{item:finite_boundedness}, let $\Phi_{\struct{A}}$ and $\Phi_{\struct{B}}$ be any universal sentences defining $\age(\struct{A})$ and $\age(\struct{B})$, respectively. 
Denote their quantifier-free parts by $\phi_{\struct{A}}$ and $\phi_{\struct{B}}$.
For two tuples $\bar{x}_1,\bar{x}_2$ of the same arity $k$, we use $E(\bar{x}_1,\bar{x}_2)$ as a shortcut for the formula stating that, for every $i\in [k]$, the $i$-the entry of $\bar{x}_1$ and the $i$-the entry of $\bar{x}_2$ are $E$-equivalent.
Then the following sentence $\Phi_{\struct{A}} \wr \Phi_{\struct{B}}$ defines $\age(\struct{A}\wr\struct{B})$:
\begin{align*}
 \forall x,y,z  \Big( E(x,x) \wedge  \big(E(x,y) \Rightarrow E(y,x)\big) \wedge  \big(E(x,y) \wedge E(y,z) \Rightarrow E(x,z) \big)
 \\   
  {}  \wedge \bigwedge\nolimits_{R\in \sigma}   \forall \bar{x}_1,\bar{x}_2   \Big( E(\bar{x}_1,\bar{x}_2) \Rightarrow \big(R(\bar{x}_1)  \Leftrightarrow  R(\bar{x}_2)  \big) 
 \Big)  \\  
   {}  \wedge  \forall \bar{x}  \Big(   \phi_{\struct{B}}(\bar{x})\vee \bigvee\nolimits_{x_1,x_2\in \bar{x}} \big(E(x_1,x_2) \wedge (x_1\neq x_2)\big)    \Big) \\
   {} \wedge \bigwedge\nolimits_{R\in \tau} \forall \bar{x} \Big( R(\bar{x}) \Rightarrow \bigwedge\nolimits_{x_1,x_2\in \bar{x}} E(x_1,x_2)\Big) \\
{}  \wedge \forall \bar{x}  \Big(   \phi_{\struct{A}}(\bar{x})\vee \bigvee\nolimits_{x_1,x_2\in \bar{x}} \neg E(x_1,x_2) \Big)\Big). 
\end{align*} 
To see this, note that:
\begin{itemize}
    \item the first line ensures that $E$ is an equivalence relation,
    \item the second line ensures that $E$ is a relational congruence on $(\struct{A}\wr\struct{B})|^\sigma$,
    \item the third line ensures that the factor by $E$ satisfies $\Phi_{\struct{B}}$, 
    \item the fourth line ensures that $\tau$-relations do not stretch over different $E$-classes,
    \item the fifth line ensures that each $E$-class satisfies $\Phi_{\struct{A}}$.
\end{itemize}

For~\eqref{item:Ramsey}, recall that, by Theorem~\ref{thm:kpt}, $\Aut(\struct{A})$ and $\Aut(\struct{B})$ are extremely amenable. 
By Proposition~\ref{prop:basic_properties_wreath_products}\eqref{item:aut_wreath}, we have $\Aut(\struct{A} \wr \struct{B})=\Aut(\struct{A}) \wr \Aut(\struct{B})$.
Hence, we can use the folklore fact that $  \Aut(\struct{A}) \wr \Aut(\struct{B})$ is a semidirect product $\Aut(\struct{A})^B \rtimes \Aut(\struct{B})$.
In particular, we immediately get that $ \Aut(\struct{A})^B$ is a normal subgroup of $\Aut(\struct{A} \wr \struct{B})$.
We claim that $\Aut(\struct{A})^B$ is closed with respect to the topology of pointwise convergence.
This follows from the fact that the complement of $\Aut(\struct{A})^B$ can be written as \[\bigcup\nolimits_{x \in B, y\in B\setminus\{x\}}\{\left((g_b)_{b\in B},h)\right)\in \Aut(\struct{A}\wr \struct{B}) \mid h(x)=y \},\] which is a union of basic open sets.
By Lemma~6.7(iii) in~\cite{kechris2005fraisse}, products of extremely amenable groups are extremely amenable, and hence $\Aut(\struct{A})^B$ is extremely amenable.
If we can show that also $\Aut(\struct{A} \wr \struct{B})/\Aut(\struct{A})^B$ is extremely amenable, then it follows from Lemma~6.7(ii) in~\cite{kechris2005fraisse} that $\Aut(\struct{A} \wr \struct{B})$ is extremely amenable as well.

The composition $\pi \circ \iota$ of the natural embedding $\iota \colon \Aut(\struct{B})\rightarrow \Aut(\struct{A} \wr \struct{B})$ with the natural quotient map $\pi \colon \Aut(\struct{A} \wr \struct{B}) \rightarrow \Aut(\struct{A} \wr \struct{B})/\Aut(\struct{A})^B$ is a topological isomorphism between $\Aut(\struct{B})$ and $\Aut(\struct{A} \wr \struct{B})/\Aut(\struct{A})^B$. 
Hence, the claim follows from the fact that $\Aut(\struct{B})$ is extremely amenable.
Recall that, by item~\ref{item:wreath_homogeneous}, we have that $\struct{A} \wr \struct{B}$ is homogeneous.
Since $\Aut(\struct{A} \wr \struct{B})$ is extremely amenable, we conclude using Theorem~\ref{thm:kpt} that $\struct{A} \wr \struct{B}$ is homogeneous Ramsey.  
\end{proof}

\subsection{Removing algebraicity and adding injectivity} \label{section:everything_together}

In the context of Conjecture~\ref{conj:bodirsky_pinsker}, the concepts of CSP-injectivity and having no algebraicity are closely related.
On the one hand, if $\struct{B}$ is a homogeneous structure with no algebraicity and its relations only contain  tuples with pairwise distinct entries, then it is CSP-injective~\cite[Lem.~4.3.6]{bodirsky2021complexity}.
On the other hand, we can prove the following.

\begin{restatable}{proposition}{corealgebraicity} \label{prop:CSP_injective_no_algebraicity}
 CSP-injective $\omega$-categorical model-complete cores have no algebraicity. 
\end{restatable}
\begin{proof} Let $\struct{B}$ be a CSP-injective $\omega$-categorical model-complete core.
We show that $\struct{B}$ has no algebraicity.
Suppose, on the contrary, that there exists $\bar{b}\in B^k$ and  $b'\notin \bar{b}$ such that the unary relation $\{b'\}$ is first-order definable in $\struct{B}$ using $\bar{b}$ as parameters.
By~\cite[Thm.~4.5.1 and Lemma~4.5.5]{bodirsky2021complexity}, the expansion $(\struct{B},\{\bar{b}\})$ is an $\omega$-categorical model-complete core, and there exists a primitive positive formula $\phi(x)$ defining $\{b'\}$ in $(\struct{B},\{\bar{b}\})$.
Also, by~\cite[Thm.~4.5.1]{bodirsky2021complexity}, there exists a primitive positive formula $\psi(\bar{y})$ defining the orbit of $\bar{b}$ in $\struct{B}$.
We may assume that $\psi$ does not contain any equality atoms, otherwise we eliminate them by identifying variables.
Let $\phi'$ be the primitive positive formula obtained by replacing the constants $\bar{b}$ in $\phi$ with the variables $\bar{y}$, and set \[ \phi'' \coloneqq   \psi(\bar{y}) \wedge \phi'(x_1)\wedge \phi'(x_2);\] we may again assume that $\phi'$ and $\phi''$ do not contain any equality atoms. 
 We use $\phi''$ to create an instance $\struct{A}$ of $\CSP(\struct{B})$ that maps homomorphically to $\struct{B}$ but has no injective such homomorphism.
 
 The domain of $\struct{A}$ consists of the variables of $\phi''$, and the  relations are determined by the atomic subformulas of $\phi''$. 
Clearly, $\struct{A}$ has a homomorphism to $\struct{B}$, because we can substitute $\bar{b}$ for $\bar{y}$ in $\phi''$ and $b$ for both $x_1$ and $x_2$.
However, $\struct{A}$ has no injective homomorphism to $\struct{B}$.
Indeed, if this was the case, then there would exist distinct witnesses $\bar{y}$, $x_1$, and $x_2$ for the satisfaction of $\phi''$ in $\struct{B}$.
By definition, there would exist $\alpha\in \Aut(\struct{B})$ with $\alpha(\bar{y})=\bar{b}$.
But then $\alpha(x_1)=\alpha(x_2)=c$, a contradiction to the injectivity of $\alpha$.
We conclude that $\struct{A}$ does not have any injective homomorphism to $\struct{B}$, and hence  $\struct{B}$ is not CSP-injective, a contradiction to our assumption. 
    Hence, $\struct{B}$ has no algebraicity.
\end{proof} 
We shall now see that CSP-injectivity and Pol-injectivity cannot hold simultaneously in $\omega$-categorical structures; yet, we shall observe immediately thereafter that CSP-injectivity can be traded against Pol-injectivity.  

\begin{restatable}{proposition}{clash} \label{prop:pol_and_CSP_injectivity} CSP-injective $\omega$-categorical structures cannot be Pol-injective. On the other hand, the expansion $\struct{A}_{I_4}$ of any CSP-injective structure $\struct{A}$ by the relation \[I_4 \coloneqq \{ (x,y,u,v) \mid x=y \Rightarrow u=v\}\]  is Pol-injective, and  $\CSP(\struct{A})$ and $\CSP(\struct{A}_{I_4})$  are Datalog-interreducible. 
 \end{restatable}
 \begin{proof}
For the first part, we show that an $\omega$-categorical $\Pol$-injective structure cannot be CSP-injective. This leads to the same conclusion as the original statement, namely, that the formal intersection between the two properties is empty.

    Let $\struct{B}$ be an $\omega$-categorical Pol-injective structure. 
    Note that every essentially injective operation preserves the relation $I_4$.
    Hence, every polymorphism of $\struct{B}$ preserves $I_4$, and so $I_4$ is pp-definable in $\struct{B}$.
    Let $\phi(x,y,u,v)$ be a primitive positive formula that defines $I_4$ in $\struct{B}$.
    Since $I_4$ contains all $4$-tuples with pairwise distinct entries, we may assume that $\phi$ contains no equality atoms, otherwise we eliminate them by identifying variables.
    But now we can use $\phi$ to create an instance $\struct{A}$ of $\CSP(\struct{B})$ that maps homomorphically to $\struct{B}$ but has no injective such homomorphism.
    More specifically, the domain of $\struct{A}$ consists of the variables of $\phi$, except that we identify $x$ and $y$ (but not $u$ and $v$), and the relations are determined by the atomic subformulas of $\phi$. We conclude that $\struct{B}$ is not CSP-injective.   

Regarding the second part, the Pol-injectivity of $\struct{A}_{I_4}$ follows from the known fact that every polymorphism of $I_4$ is essentially injective~\cite{BodChenPinsker} (also~\cite[Lem.~7.5.1]{bodirsky2021complexity}). We proceed with the proof of Datalog-interreducibility.
We first give a Datalog-reduction $\mathcal{I}^{\ast}$ from $\CSP(\struct{A}_{I_4})$ to $\CSP(\struct{A})$.
Let $\tau$ be the signature of $\struct{A}$.
Given an instance $\struct{X}$ of $\CSP(\struct{A}_{I_4})$, the relation $\approx$ is defined in $\struct{X}$ by the following Datalog program:

\smallskip
\begin{itemize}
    \item $(x\approx y) \Leftarrow     (x=y)$,
    \item $(x\approx y)  \Leftarrow     (y\approx x)$,
    \item $(x\approx z)  \Leftarrow     (x\approx y)  \wedge (y\approx z)$,
    \item $(x\approx y)  \Leftarrow     I_4(u,v,x,y) \wedge  (u\approx v)$.
\end{itemize} 
\smallskip

By definition, this relation is an equivalence relation on $X$; denote by $q_{\approx}$ the induced quotient map $X\rightarrow X/{\approx}, x\mapsto [x]_{\approx} $.
We define $\mathcal{I}^{\ast}(\struct{X})$ as the $\tau$-structure with domain $X$ and whose relations are the preimages of all $\tau$-relations in $\struct{X}/{\approx}$ under $q_{\approx}$.
It remains to verify that  $\mathcal{I}^{\ast}$ is a reduction. 

First, suppose that $\mathcal{I}^{\ast}(\struct{X}) \rightarrow \struct{A}$.
Note that $(\struct{X}/{\approx})|^{\tau}$ is a substructure of $\mathcal{I}^{\ast}(\struct{X})$, and hence $(\struct{X}/{\approx})|^{\tau}\rightarrow \struct{A}$.
Since $\struct{A}$ is CSP-injective, there exists an injective homomorphism $h \colon (\struct{X}/{\approx})|^{\tau} \rightarrow \struct{A}$.
By the definition of $\approx$,  
    if we have $[x_1]_{\approx}=[x_2]_{\approx}$ and $([x_1]_{\approx},[x_2]_{\approx},[x_3]_{\approx},[x_4]_{\approx})\in I_4^{\struct{X}/{\approx}}$, then we also have $[x_3]_{\approx}=[x_4]_{\approx}$. 
Since $h$ is injective, it follows that $h$ is also a homomorphism from $\struct{X}/{\approx}$ to $\struct{A}_{I_4}$.
Then the composition $h\circ q_{\approx} $ is a homomorphism from $\struct{X}$ to $\struct{A}_{I_4}$.

Now suppose that there exists a homomorphism $g\colon \struct{X}\rightarrow\struct{A}_{I_4}$. 
Note that, by definition, for every $k$-ary $R\in\tau$, we have $(s_{i_1},\dots, s_{i_k}) \in R^{\mathcal{I}^{\ast}(\struct{X})}$ if and only if there are $j_1\in [s_{i_1}]_{\approx}, \dots, j_k\in [s_{i_k}]_{\approx}$ such that  $(j_1,\dots,j_k) \in R^{\struct{X}}$.
By the definition of $\approx$, every function preserving $I_4$ must identify $\approx$-equivalent elements, and therefore $(g(s_{i_1}),\dots,g(s_{i_k}))=(g(j_1),\dots,g(j_k))$, which implies $(g(s_{i_1}),\dots,g(s_{i_k})) \in  R^{\struct{A}_{I_4}}$. 
Therefore, $g$ is a homomorphism from $\mathcal{I}^{\ast}(\struct{X})$ to $\struct{A}$.

The reduction in the opposite direction is trivial. Given an instance $\struct{X}$ of $\CSP(\struct{A})$, we define $\mathcal{I}(\struct{X})$ as the $(\tau\cup \{I_4\})$-expansion of $\struct{X}$ by an empty relation.
\end{proof}

 Finally, we give the proofs of Theorem~\ref{thm:removing_algebraicity} and Theorem~\ref{thm:polinjective}, restated below.
 \maintheorem*
 \maintheoremtwo*
 As the two theorems only differ in the presence of the relation $I_4$, 
large parts of the proofs of Theorem~\ref{thm:removing_algebraicity} and Theorem~\ref{thm:polinjective} are essentially the same.
For this reason, we conduct both proofs simultaneously, and differentiate whenever necessary. 

\begin{proof}[Proof of Theorem~\ref{thm:removing_algebraicity} and Theorem~\ref{thm:polinjective}]
Let $\tau$ and $\sigma$ be the signatures of $\struct{A}$ and $\struct{B}$, respectively.
We first define the auxiliary structure $\struct{B}'\coloneqq  (\mathbb{Q};<) \wr \struct{B}$. 
By Example~\ref{ex:properties_of_q1}, the structure $(\mathbb{Q};<)$ is finitely bounded homogeneous. It is also CSP-injective:  a digraph maps homomorphically to $(\mathbb Q;<)$ if and only if it is acyclic; in this case there is always an injective homomorphism.
It is a model-complete core, and hence has no algebraicity (Proposition~\ref{prop:CSP_injective_no_algebraicity}).
It follows immediately from Proposition~\ref{prop:basic_properties_wreath_products}\eqref{item:algebraicity}  that $\struct{B}'$ has no algebraicity.
The expansion $(\struct{B}',{\neq})$ of $\struct{B}'$ by the binary inequality has no algebraicity because taking expansions by first-order definable relations does not change the automorphism group, proving item~\ref{item:Bremovingalgebraicity}; the same is true for the structure $(\struct{B}',{\neq}, I_4)$.

\smallskip{\textit{Proof of item~\ref{item:Bomegacat}.}} This follows directly from Proposition~\ref{prop:basic_properties_wreath_products}\eqref{item:aut_wreath_omega_cat} applied to $\struct{B}'$ and the fact that we take an expansion by relations which are preserved by all bijections, and in particular all automorphisms of $\struct{B}'$. 

\smallskip{\textit{Proof of item~\ref{item:BhomogRamsey}.}} Since $\ordblowupi{\struct{B}}$ is an expansion of $\struct{B}'$ by two relations that are preserved by all embeddings, we can ignore these relations and only prove the statement for $\struct{B}'$.
Recall that $\struct{B}'=(\mathbb{Q};<)\wr\struct{B}$.
The first part of~item~\ref{item:BhomogRamsey} follows directly from Proposition~\ref{prop:basic_properties_wreath_products}\eqref{item:wreath_homogeneous} because $(\mathbb{Q};<)$ is homogeneous.
We continue with the second part. 
There the forward direction follows directly from Proposition~\ref{prop:basic_properties_wreath_products}\eqref{item:Ramsey} because $(\mathbb{Q};<)$ is homogeneous Ramsey \cite{kechris2005fraisse}.
We continue with the backward direction.

Suppose that $\age\ordblowupi{\struct{B}}$ has the RP.
Then clearly also $\age(\struct{B}')$ has the RP.
Let $\struct{S},\struct{M}\in \age(\struct{B})$ be arbitrary.
We transform them into $\struct{S}',\struct{M}'\in \age(\struct{B}')$ by letting $E$ interpret as the diagonal relation and $<$ as the empty relation.
Since $\age(\struct{B}')$ has the RP, there exists a Ramsey witness $\struct{L}'\in \age(\struct{B}')$ for $(\struct{S}',\struct{M}')$.
We set $\struct{L}\coloneqq (\struct{L}'/E^{\struct{L}'})|^{\sigma}$.
Let $f\colon \binom{\struct{L}}{\struct{S}} \rightarrow [2]$ be arbitrary. 
We define $f'\colon \binom{\struct{L}'}{\struct{S}'} \rightarrow [2]$ as follows.
For an embedding $e'\in \binom{\struct{L}'}{\struct{S}'} $, we set $f'(e')\coloneqq f(q_E\circ e')$, where $q_E$ denotes the factor map $x\mapsto [x]_E$.
We have that $q_E\circ e' \in \binom{\struct{L}}{\struct{S}}$ because $E$ interprets as the diagonal relation in $\struct{S}'$ and as a relational congruence on $\struct{L}'|^\sigma$.
Since $\struct{L}'$  is a Ramsey witness for $(\struct{S}',\struct{M}')$, there exists $e'\in \binom{\struct{L}'}{\struct{M}'} $ such that $f'$ is constant on  $\smash{\bigl\{e'\circ u'\;|\; u'\in \binom{\struct{M}'}{\struct{S}'}\bigr\} \subseteq \binom{\struct{L}'}{\struct{S}'}}$.
Since $E$ interprets as the diagonal relation in
$\struct{M}'$ and as a relational congruence on $\struct{L}'|^\sigma$, the map $e\coloneqq q_E\circ e'$ is contained in  $\binom{\struct{L}}{\struct{M}}$. 
By construction, $f$ is constant on  $\smash{\bigl\{e\circ u\;|\; u\in \binom{\struct{M}}{\struct{S}}\bigr\} \subseteq \binom{\struct{L}}{\struct{S}}}$.

\smallskip{\textit{Proof of item~\ref{item:Bfinbounded}.}}  
The finite boundedness of $ \ordblowupi{\struct{B}}$ can be deduced from Proposition~\ref{prop:basic_properties_wreath_products}\eqref{item:finite_boundedness}, and the fact that we take an expansion by relations definable by a Boolean combination of equality atoms.
More specifically, as explained in the introduction, up to isomorphism, the finite substructures of a finitely bounded structure are just the finite models of a universal sentence. 
By adding new clauses, that define the relations $\neq$ and $I_4$ in terms of equalities, to the universal sentence for $\age(\struct{B}')$, 
we obtain a universal sentence for $\age\ordblowupi{\struct{B}}$.

\smallskip{\textit{Proof of items~\ref{item:ACSPinj}~and~\ref{item:ADatalog}.}}  
Note that removing all symbols in $\sigma\setminus \tau$ as well as the symbol $<$ from $(\struct{B}',{\neq})$ yields the structure $\blowup{\struct{A}}$.  
We claim that $\blowup{\struct{A}}$ is CSP-injective. 
Let $\struct{X}$ be an arbitrary finite structure for which there exists a homomorphism $h\colon \struct{X} \rightarrow \blowup{\struct{A}}$.
We define $\struct{X}_{\pi}$ as the $(\tau \cup \{E\})$-structure with domain $X$ such that, for every $R\in \tau\cup \{E\}$, we have $\bar{x}\in R^{\struct{X}_{\pi}}$ if and only if $h(\bar{x})\in R^{\blowup{\struct{A}}}$.
It is easy to see that $\struct{X}_{\pi}$ embeds into the $(\tau \cup \{E\})$-reduct of $\struct{B}'$ (if multiple elements in $X$ are mapped to $(q,a)$ for some $q\in\mathbb{Q}$ and $a\in A$, we can always map them to different elements in $\mathbb{Q}\times \{a\}$ since this does not affect the relations in $\tau\cup \{E\}$).
Hence, $\struct{X}_{\pi}$ expanded by the binary inequality predicate embeds into $\blowup{\struct{A}}$. This yields an injective homomorphism from $\struct{X}$.
We conclude that $\blowup{\struct{A}}$ is CSP-injective.

Next, we informally describe two Datalog-reductions $\mathcal{I}^{\ast}$ and $\mathcal{I}$: from $\CSP\blowup{\struct{A}}$ to $\CSP(\struct{A})$ and back, starting with the former. 
Let $\struct{X}$ be an instance of $\CSP\blowup{\struct{A}}$.
Denote by $\sim$ the equivalence closure of $E^{\struct{X}}$, and let $q_{\sim}\colon \struct{X}\rightarrow (\struct{X}/{\sim}), x \mapsto [x]_{\sim}$ be the induced quotient map.
We define $\mathcal{I}^{\ast}(\struct{X})$ as the $\tau$-structure with domain $X$ and whose relations are the preimages of all $\tau$-relations in $\struct{X}/{\sim}$ under $q_{\sim}$, except that we additionally add, for every $x\in X$ such that $(x,x)\in {\neq}^{\struct{X}}$, the tuple $(x,\dots, x)$ to every $\tau$-relation.
It is easy to see that $\mathcal{I}^{\ast}(\struct{X})$ is definable in Datalog; we verify that $\mathcal{I}^{\ast}$ is a reduction from $\CSP\blowup{\struct{A}}$ to $\CSP(\struct{A})$. 
Clearly, if $\struct{X}$ maps homomorphically to $\blowup{\struct{A}}$, then $\mathcal{I}^{\ast}(\struct{X})$ maps homomorphically to $\struct{A}$.
Suppose that there exists a homomorphism $h\colon \mathcal{I}^{\ast}(\struct{X})\rightarrow\struct{A}$.
Since by non-triviality $\struct{A}$ does not have a constant endomorphism, there is no $x\in X$ such that $(x,x)\in {\neq}^{\struct{X}}$: otherwise, the map on $\struct{A}$ sending every element to $h(q_\sim(x))$ would provide a constant endomorphism. 
Denote by $\struct{X}'$ the structure obtained from $\struct{X}$ by removing all tuples from ${\neq}^{\struct{X}}$.
Now, choosing representatives $x_1,\dots,x_k$ for all classes of $\struct{X}/{\sim}$, the function sending the entirety of $[x_i]_\sim$ to $(0,h(x_i))$ for all $i\leq k$ is a homomorphism from $\struct{X}'$ to $\blowup{\struct{A}}$.
By the CSP-injectivity of $\blowup{\struct{A}}$, it can be replaced by an injective homomorphism, which is then also a homomorphism from $\struct{X}$ to $\blowup{\struct{A}}$.
We continue with the reduction from $\CSP(\struct{A})$ to $\CSP\blowup{\struct{A}}$, which is trivial.
Let $\struct{X}$ be an instance of $\CSP(\struct{A})$.
We define $\mathcal{I}(\struct{X})$ as the $(\tau\cup \{E,\neq\})$-expansion of $\struct{X}$ by empty relations.
By CSP-injectivity, there exists a homomorphism $h\colon \mathcal{I}(\struct{X}) \rightarrow \blowup{\struct{A}}$ if and only if there exists an injective such homomorphism $i$.
This is the case if and only if the relation $E^{\blowup{\struct{A}}}$ restricted to $i(X)$ describes the kernel of a homomorphism from  $\struct{X}$ to $\struct{A}$.
Therefore, $\struct{X}\rightarrow\struct{A}$ if and only if $\mathcal{I}(\struct{X})\rightarrow\blowup{\struct{A}}$. 
 
We now consider $\blowupi{\struct{A}}$. 
By Proposition~\ref{prop:pol_and_CSP_injectivity}, we have that $\blowupi{\struct{A}}$ is Pol-injective and $\CSP\blowupi{\struct{A}}$ and $\CSP\blowup{\struct{A}}$ are Datalog-interreducible.
That $\CSP\blowupi{\struct{A}}$ and $\CSP(\struct{A}$) are Datalog-interreducible then follows from the fact that Datalog-reductions can be composed.
 
Finally, we verify the properties in items~\ref{Amodelcomplcore}--\ref{AppconstructK3}.
We only explicitly cover the case where $I_4$ is present,  our arguments clearly also hold in the case without $I_4$.

\smallskip{\textit{Proof of item~\ref{Amodelcomplcore}.}}
By item~\ref{item:Bomegacat} applied to $\struct{A}$, $\ordblowup{\struct{A}}$ is $\omega$-categorical, and therefore so is its reduct $\blowup{\struct{A}}$. 
Next, assume that $\struct{A}$ is a model-complete core and suppose that $e^{\ast}$ is an endomorphism of $\blowupi{\struct{A}}$. Let $S^{\ast}$ be an arbitrary finite subset of $\mathbb{Q}\times A$.
Since $e^{\ast}$ preserves $\neq$, it is injective. Since $e^{\ast}$ preserves $E^{\blowupi{\struct{A}}}$, which is an equivalence relation, it naturally acts on its classes.
Let $e$ be the map on $(\mathbb{Q}\times A)/E^{\blowupi{\struct{A}}}$ mapping each $[x]_{E^{\blowupi{\struct{A}}}}$ to $[e^{\ast}(x)]_{E^{\blowupi{\struct{A}}}}$, and let $S\coloneqq \{[x]_{E^{\blowupi{\struct{A}}}} \mid x\in S^{\ast}\}$.
Similarly as in Proposition~\ref{prop:basic_properties_wreath_products}\eqref{item:aut_wreath}, we have that $e$ is an endomorphism of $\struct{A}$. 
Since $\struct{A}$ is a model-complete core, there exists an automorphism $\alpha$ of $\struct{A}$ such that $\alpha|_S=e|_S$.
We claim that there exists an automorphism ${\alpha}^{\ast}$ of $\blowupi{\struct{A}}$ such that ${\alpha}^{\ast}|_{S^{\ast}}=e^{\ast}|_{S^{\ast}}$.
To see this, note that we can simply fix a  permutation  of each equivalence class of $E^{\blowupi{\struct{A}}}$ and extend $\alpha$, which is a permutation of the equivalence classes, according to it.
Since $S^{\ast}$ is finite, we can moreover select the permutations so that ${\alpha}^{\ast}$ coincides with $e^{\ast}$ on $S^{\ast}$.

Next, suppose that $e$ is an endomorphism of $\struct{A}$, and let $S$ be an arbitrary finite subset of $A$.
Define $S^{\ast}\coloneqq\{0\}\times S$. 
Let $e^{\ast}$ be the map on $\mathbb{Q}\times A$ defined via $e^{\ast}(q,a)\coloneqq(q,e(a))$.
This function is an endomorphism of $\blowupi{\struct{A}}$.   
Since $\blowupi{\struct{A}}$ is a model-complete core, there exists an automorphism ${\alpha}^{\ast}$ of 
$\blowupi{\struct{A}}$ such that ${\alpha}^{\ast}|_{S^{\ast}}=e^{\ast}|_{S^{\ast}}$.
Let $\alpha$ be the map on $A$ mapping each $[x]_{E^{\blowupi{\struct{A}}}}$ to $[e^{\ast}(x)]_{E^{\blowupi{\struct{A}}}}$.
As before, we have that  $\alpha$ is a well-defined automorphism of $\struct{A}$ extending $e$ on $S$.

\smallskip{\textit{Proof of item~\ref{AppconstructK3}.}}
One direction is simple  (and does not use the assumption on the orbit growth), because $\blowupi{\struct{A}}$ clearly pp-constructs $\struct{A}$;
the pp-construction is given by $\struct{A} =(\blowupi{\struct{A}}/E^{\blowupi{\struct{A}}})|^{\tau}$. Thus, if $\struct{A}$ pp-constructs $\struct{K}_3$, then so does $\blowupi{\struct{A}}$. 

For the other direction, we use  a characterization via the pseudo-Siggers identity from~\cite[Thms.~1.3 and~3.4]{barto2019equations}.
Suppose that $\struct{A}$ does not pp-construct $\struct{K}_3$.
Let $\struct{C}$ be the model-complete core of $\struct{A}$.
Since $\struct{A}$ and $\struct{C}$ are homomorphically equivalent, $\struct{C}$ also does not pp-construct 
$\struct{K}_3$.
Then, by Theorem~1.3 in~\cite{barto2019equations}, there exist $\alpha,\beta,s\in \Pol(\struct{C})$ witnessing the pseudo-Siggers identity. 
Consider the structure $\blowupi{\struct{C}}$.
By item~\ref{Amodelcomplcore}, $\blowupi{\struct{C}}$ is an $\omega$-categorical model-complete core.
Note that we can lift the homomorphic equivalence between $\struct{A}$ and $\struct{C}$ to a homomorphic equivalence between $\blowupi{\struct{A}}$ and $\blowupi{\struct{C}}$ by setting, for any homomorphism $h$ between $\struct{A}$ and $\struct{C}$ and irrespectively of its direction, $h^{\ast}(x,y)\coloneqq (e(x,y),h(y))$, where $e$ is an arbitrary injection from the Cartesian product of $\mathbb{Q}$ with the domain of $h$ into $\mathbb{Q}$.
The inclusion of the injection in the first coordinate ensures that $h^{\ast}$ preserves $I_4$ and $\neq$ (because they are preserved by all injective operations).
Since $\blowupi{\struct{A}}$ and $\blowupi{\struct{C}}$ are homomorphically equivalent, $\blowupi{\struct{C}}$ is the model-complete core of $\blowupi{\struct{A}}$.
Next, we lift the pseudo-Siggers identity from $\Pol(\struct{C})$ to $\Pol\blowupi{\struct{C}}$.
Let $e$ be an arbitrary injection from $\mathbb{Q}\times C$ into $\mathbb{Q}$.
The structure $(\mathbb{Q};\neq, I_4)$ is an $\omega$-categorical model-complete core that does not pp-construct  $\struct{K}_3$~\cite[Thms.~12.0.1,~12.7.3,~12.9.2, and Cor.~6.4.4]{bodirsky2021complexity}. 
Hence, by Theorem~1.3 in~\cite{barto2019equations} (since pp-constructions are more general than pp-interpretations with parameters), there are  $\alpha',\beta',s'\in \Pol(\mathbb{Q};\neq, I_4)$ witnessing the pseudo-Siggers identity.
Set:

\smallskip 
\begin{itemize}
    \item $\alpha^{\ast}(x,y)\coloneqq   (\alpha'(x),\alpha(y))$,  
    \item $\beta^{\ast}(x,y)\coloneqq  (\beta'(x),\beta(y))$, and 
    \item $ s^{\ast}((x_1,y_1),\dots, (x_6,y_6))\coloneqq  (s'(e(x_1,y_1),\dots, e(x_6,y_6)),s(y_1,\dots, y_6)).$
\end{itemize} 
\smallskip 

It is easy to verify that $\alpha^{\ast},\beta^{\ast},s^{\ast}$ are polymorphisms of $\blowupi{\struct{C}}$ and that these operations witness the pseudo-Siggers identity.
It moreover follows by inspection of the proof of  Proposition~\ref{prop:basic_properties_wreath_products}\eqref{item:aut_wreath_omega_cat} that also the structure $\blowupi{\struct{A}}$ has less than double exponential orbit growth, and therefore so does its model-complete core $\blowupi{\struct{C}}$ by the proof of \cite[Cor.~1.8]{barto2019equations}. 
Combining this with the pseudo-Siggers identity, it follows from Theorem~3.4 in~\cite{barto2019equations} (whose item~(ii) is equivalent to the pp-construction of $\struct{K}_3$ by Theorem~1.8 in~\cite{barto2018wonderland}) that $\blowupi{\struct{A}}$ does not pp-construct $\struct{K}_3$.
\end{proof}

With the proof of Theorem~\ref{thm:polinjective} at hand, we
can easily prove Corollary~\ref{cor:binaryInj}.
\binaryInj*  
\begin{proof}[Proof of Corollary~\ref{cor:binaryInj}]
Suppose that $\struct{A}$ does not pp-construct $\struct{K}_3$.
By Theorem~\ref{thm:polinjective}\eqref{AppconstructK3}, the structure $\blowupi{\struct{A}}$ does not pp-construct $\struct{K}_3$ either.
Let $\struct{C}$ be the model-complete core of $\struct{A}$.
By the proof of Theorem~\ref{thm:polinjective}\eqref{AppconstructK3}, we have that $\blowupi{\struct{C}}$ is the model-complete core of $\blowupi{\struct{A}}$.
Hence, also $\blowupi{\struct{C}}$ does not pp-construct $\struct{K}_3$.

By \cite[Thm.~3.14]{MarimonPinsker23}, we have that $\Pol\blowupi{\struct{C}}$ contains a binary essential operation $f$.    
Let $g$ and $h$ be any homomorphisms witnessing the homomorphic equivalence between $\blowupi{\struct{A}}$ and $\blowupi{\struct{C}}$, respectively. 
Then $(x,y) \mapsto h \circ f(g(x),g(y))$ is a binary essential polymorphism of 
$\blowupi{\struct{A}}$.
Since $\blowupi{\struct{A}}$ is Pol-injective, this operation must be injective. 
\end{proof}

 We conclude this section with an example. 

\begin{example} \label{ex:two_ex}  Consider the two structures 
\begin{align*}
\struct{A}_1 & \coloneqq (\{0,1\};\{1\}, \{(x,y,z)\in \{0,1\}^3 \mid x+y+z=0 \bmod 2\}), \\
\struct{A}_2 & \coloneqq (\mathbb{Q};<,\{(x,y,z) \in \mathbb{Q}^{3} \mid x\geq y \text{ or } x\geq z\}).
\end{align*} 
The former is a finitely bounded homogeneous model-complete core \emph{with} algebraicity (due to it being finite) and the latter is a finitely bounded homogeneous model-complete core \emph{without} algebraicity. 
Note that $\struct{A}_1$ is preserved by the symmetric operation \[(x_1,\dots,x_n)\mapsto x_1+\cdots+x_n \bmod 2\] for every odd $n\geq 2$ and $\struct{A}_2$ is preserved by the symmetric operation \[(x_1,\dots,x_n)\mapsto \min\{x_1,\dots,x_n\}\] for every $n\geq 2$. Thus, in both cases the polymorphism clone satisfies the cyclic identity of arity $n$ for some $n\geq 2$.
Since pp-constructions preserve the satisfaction of height-1 identities in polymorphism clones~\cite{barto2018wonderland} and no cyclic operation $f$ of arity $n\geq 2$ preserves $\neq$ over an infinite set, $\struct{A}_i$ does not pp-construct $\blowup{\struct{A}_i}$ for both $i\in [2]$. 
Hence, our Datalog-reduction in Theorem~\ref{thm:removing_algebraicity} is not subsumed by gadget reductions.
\end{example}

\section{\texorpdfstring{$\omega$}{omega}-categorical cheeses for PCSPs} \label{section:sandwiches}

In the present section, we give a proof of Theorem~\ref{thm:main_theorem_sandwiches}.

\subsection{Full powers}

    Our basic tool for creating $\omega$-categorical cheeses for PCSPs that are not finitely tractable are  \emph{full powers} (see, e.g.,~Bodirsky~\cite[Sect.~3.5]{bodirsky2021complexity}).
    Roughly speaking, a full power $\fpwr{\struct{A}}{d}$ is a higher-dimensional representation of $\struct{A}$ that is pp-bi-interpretable with $\struct{A}$.
    Its central feature is that its polymorphisms are the polymorphisms of $\struct{A}$ acting on $d$-tuples, thus allowing to factor it by the orbit-equivalence of $\Aut(\struct{A})$ on $d$-tuples (rather than on elements). 
     In other words, the orbits of $d$-tuples of $\struct{A}$ are represented by orbits of $1$-tuples (elements) of $\fpwr{\struct{A}}{d}$, which is crucial in applications of the sandwich method. 
    Let $\struct{A}$ be a relational structure with a signature $\tau$ and let $d\in\mathbb{N}$ be arbitrary.  
    The $d$-th \emph{full power} of $\struct{A}$, denoted $\fpwr{\struct{A}}{d}$, is the structure with domain $A^d$ and the following relations for every $k \in [d]$:
     \begin{itemize}
         \item for every $R\in \tau$ of arity $k$ and every injection $\iota\colon [k] \rightarrow [d]$, the unary relation 
         \begin{align*}
            R^{\fpwr{\struct{A}}{d}}_{\iota} \coloneqq \{ (a_1\dots, a_d)\in A^d \mid  (a_{\iota(1)},\dots,a_{\iota(k)}) \in R^\struct{A} \};
         \end{align*}
         \item  for every $R\in \tau$ of arity $k$ and every function $\iota\colon [k] \rightarrow [d]$, the $k$-ary relation
         \begin{align*}
             \widehat{R}^{ \fpwr{\struct{A}}{d}}_{\iota} \coloneqq \{ ((a^1_1,\dots, a^1_d),\dots, (a^k_1,\dots, a^k_d)) \in (A^d)^k \mid  (a^1_{\iota(1)},\dots,a^k_{\iota(k)})\in R^{\struct{A}}
               \};
         \end{align*}  
         \item for all functions $\iota,\iota' \colon [k] \rightarrow [d]$, the binary \emph{compatibility relation} 
         \begin{align*}
             E^{\fpwr{\struct{A}}{d}}_{\iota,\iota'} \coloneqq \{ ((a^1_1\dots, a^1_d),(a^2_1,\dots,a^2_d))\in (A^d)^2 \mid \forall i \in [k] \colon a^1_{\iota(i)}=a^2_{\iota'(i)}    \}.
         \end{align*}
          
    \end{itemize}  
 Below we give several important properties of full powers.
 For an operation $f\colon A^n \rightarrow A$, we denote the component-wise action of $f$ on $d$-tuples by $f\acts A^d$. For a set $F$ of operations on $A$, we set  $F\acts A^d \coloneqq \{f\acts A^d \mid  f\in F\}$.

\begin{restatable}[Basic properties of full powers]{proposition}{fullpower} \label{prop:basic_properties_full_powers}
Let $\struct{A}$ be a relational structure over a finite signature $\tau$ and let its relations be of arity $\leq d$. Then:
    \begin{enumerate}
        \item \label{item:fpwrpolequal} $\Pol(\fpwr{\struct{A}}{d})$ consists of the component-wise actions of $\Pol(\struct{A})$ on $d$-tuples.  
        \item \label{item:fpwrppconstr} The structures $\struct{A}$ and $\fpwr{\struct{A}}{d}$ are pp-bi-interpretable.  
        \item \label{item:fpwromegacat} If $\struct{A}$ is $\omega$-categorical, then $\fpwr{\struct{A}}{d}$ is $\omega$-categorical.
        \item \label{item:fpwrhomog} If $\struct{A}$ is homogeneous, then $\fpwr{\struct{A}}{d}$ is homogeneous.
        \item \label{item:fpwrmodcompl} If $\struct{A}$ is a model-complete core, then $\fpwr{\struct{A}}{d}$ is a model-complete core.
              \item \label{item:fpwrfinbou} If $\struct{A}$ is finitely bounded, then $\fpwr{\struct{A}}{d}$ is finitely bounded.
        \item \label{item:fpwrramsey} If $\struct{A}$ is homogeneous Ramsey, then $\fpwr{\struct{A}}{d}$ is homogeneous Ramsey.
        \item \label{item:fpwrreduct} If $\struct{A}$ is a reduct of $\struct{B}$, then $\fpwr{\struct{A}}{d}$ is a reduct of $\fpwr{\struct{B}}{d}$.
    \end{enumerate} 

 \end{restatable}  
 We include a full proof of Proposition~\ref{prop:basic_properties_full_powers} for the convenience of the reader.

 \begin{proof} For $i\in [d]$, we denote by $\iota_i$ the function from $[1]$ to $[d]$ defined by $\iota_i(1)\coloneqq i$.

        For~\eqref{item:fpwrpolequal}, let $f\in  \Pol(\fpwr{\struct{A}}{d})$ be arbitrary, and let $n$ be its arity.
        In what follows, we shall use the abbreviation $\bar{a}=(a_1,\dots,a_n)$.
        If $\bar{a}$ appears as the $i$-th row of $(\bar{x}_1,\dots,\bar{x}_n)\in (A^d)^n$ (viewed as a $d\times n$-matrix) and as the $j$-th row in $(\bar{y}_1,\dots,\bar{y}_n)\in (A^d)^n$, then, since $f$ preserves $E^{\fpwr{\struct{A}}{d}}_{\iota_i,\iota_j}$, the $i$-th entry of $f(\bar{x}_1,\dots,\bar{x}_n)$ equals the $j$-th entry of $f(\bar{y}_1,\dots,\bar{y}_n)$.
        This ensures that there exists a well-defined operation $f'\colon A^n\to A$ such that $f=f'\acts A^d$.
        The definition of $R^{\fpwr{\struct{A}}{d}}_{\mathrm{id}}$ for every $k$-ary relation $R^\struct{A}$, where $\mathrm{id}$ is the identity map on $[k]$, guarantees that $f'\in \Pol(\struct{A})$.
        Thus, we have that $\Pol(\fpwr{\struct{A}}{d})\subseteq\Pol(\struct{A})\acts A^d$. 
        The other inclusion is trivial.
        
        For~\eqref{item:fpwrppconstr}, we proceed as in the proof of \cite[Lem.~3.5.4]{bodirsky2021complexity}; the identity map yields a $d$-dimensional pp-interpretation of $\fpwr{\struct{A}}{d}$ in $\struct{A}$ and the first projection gives a 1-dimensional pp-interpretation of $\struct{A}$ in $\fpwr{\struct{A}}{d}.$
        It is easy to verify that this gives a pp-bi-interpretation between $\struct{A}$ and $\fpwr{\struct{A}}{d}$. 
        
        For~\eqref{item:fpwromegacat}, note that $\Aut(\fpwr{\struct{A}}{d})=\Aut(\struct{A})\acts A^d$ by item~\eqref{item:fpwrpolequal}.
        Since every $n$-orbit in $\fpwr{\struct{A}}{d}$ corresponds to a $dn$-orbit in $\struct{A}$, there are finitely many $n$-orbits in $\fpwr{\struct{A}}{d}$ for each $n\in \mathbb{N}$. 
        Hence, $\fpwr{\struct{A}}{d}$ is $\omega$-categorical as well.
        
        For~\eqref{item:fpwrhomog}, let $\alpha:\struct{D}_1\to\struct{D}_2$ be an isomorphism between two finite substructures of $\fpwr{\struct{A}}{d}$ and let $A_1,A_2$ be the sets of elements of $A$ appearing as entries in some $d_1\in D_1$ or $d_2\in D_2$, respectively.
        Then the relations $E^{\fpwr{\struct{A}}{d}}_{\iota_i,\iota_j}$ and $\widehat{R}^{\,\fpwr{\struct{A}}{d}}$ for $R\in\tau$ guarantee that $\alpha$ induces an isomorphism $\beta$ between $\struct{A}_1$ and $\struct{A}_2$.
        Since $\struct{A}$ is homogeneous, we can extend $\beta$ to some automorphism $\beta^\ast$ of $\struct{A}$.
        Now $\alpha^\ast\coloneqq\beta^\ast\acts A^d$ is an automorphism of $\fpwr{\struct{A}}{d}$ extending $\alpha$.
        
        For~\eqref{item:fpwrmodcompl}, let $e\in \End(\fpwr{\struct{A}}{d})$ be arbitrary and let $D$ be an arbitrary finite subset of $A^d$.
        Define $A_D$ as the set of all elements of $A$ appearing as entries in some $d\in D$.
        By item~\eqref{item:fpwrpolequal}, we have $e=e'\acts A^d$ for some $e'\in\End(\struct{A})$.
        Since $\struct{A}$ is a model-complete core, there exists $\alpha'\in\Aut(\struct{A})$ such that $\alpha'|_{A_D}=e'|_{A_D}$.
        Then $\alpha\coloneqq\alpha'\acts A^d$ is an automorphism of $\fpwr{\struct{A}}{d}$ satisfying $\alpha|_{D}=e|_{D}$. 
        
        Item~\eqref{item:fpwrfinbou} can be shown in a similar fashion as \cite[Lem.~3.5.4]{bodirsky2021complexity}.
        
        For~\eqref{item:fpwrramsey}, note that, by item~\eqref{item:fpwrpolequal} and since the topology of pointwise convergence is determined by the behaviour of functions on finite subsets, the mapping $\alpha \mapsto \alpha \acts A^d$ is a topological isomorphism from $\Aut(\struct{A})$ to $\Aut(\fpwr{\struct{A}}{d})$. 
        By Theorem~\ref{thm:kpt}, we have that $\fpwr{\struct{A}}{d}$ is homogeneous Ramsey.
        
        Item~\eqref{item:fpwrreduct} follows directly from the definition of $\fpwr{\struct{A}}{d}$.
    \end{proof}

\subsection{Our sandwich recipe}

 To construct a finite PCSP sandwich for a structure $\struct{A}$
it is sufficient to select any finite substructure $\struct{S}_1$ and finite factor $\struct{S}_2$ of $\struct{A}$.
The issue is that, in many cases, $\PCSP(\struct{S}_1,\struct{S}_2)$ will be finitely tractable, possibly because already $\CSP(\struct{S}_1)$ or $\CSP(\struct{S}_2)$ is tractable.
Proposition~\ref{prop:sandwiches} describes a general condition for structures in the scope of Conjecture~\ref{conj:bodirsky_pinsker} under which this does not happen, i.e., where $\PCSP(\struct{S}_1,\struct{S}_2)$ is provably not finitely tractable.
Intuitively, in order to obtain a non-finitely tractable sandwich for a structure in the scope of Conjecture~\ref{conj:bodirsky_pinsker}, it is enough to take a sufficiently large full power thereof, a finite substructure $\struct{S}_1$ whose polymorphisms preserve the inequality, and a finite factor $\struct{S}_2$ by an arbitrary subgroup of automorphisms preserving a linear order and acting with finitely many orbits.  

\begin{restatable}{proposition}{sandwiches} \label{prop:sandwiches}
Let $\struct{A}$ be a reduct of a linearly ordered finitely bounded homogeneous structure $\struct{B}$ such that one of the relations of $\struct{A}$ interprets as $\neq$.
    Let $d\in\mathbb{N}$ be such that the bounds of $\struct{B}$ have size $\leq d$ and its relations are of arity $\leq d-1$. 
    Then, for every finite substructure $\struct{S}$ of $\struct{A}$ with $|S|\geq 3$, every finite cheese for $\PCSP(\fpwr{\struct{S}}{d},\orbeq{\fpwr{\struct{A}}{d}}{\Aut(\struct{B})})$ pp-constructs $\struct{K}_3$.  
   
 \end{restatable}
    
We split the proof of Proposition~\ref{prop:sandwiches} into several   statements. 
We start with the following easy  observation (cf.~ also~\cite{Mottet_2025}) about pseudo-cyclic polymorphisms.

\begin{restatable}{lemma}{pseudocyclic}  
\label{lemma:pseudocyclic}  
 Let $\struct{A}$ be a reduct of a linearly ordered structure $\struct{B}$ and $\struct{S}$ a finite substructure of $\struct{A}$. If $f\in \Pol(\struct{S},\struct{A})$ is pseudo-cyclic modulo $\Aut(\struct{B})$, then $f$ is cyclic.
 \end{restatable}   
 
    \begin{proof}
        Let $n$ be the arity of $f$ and let $\alpha_1,\alpha_2\in \Aut(\struct{B})$ be such that 
        \begin{align*}
            \alpha_1\circ f(s_1,\dots,s_n)=\alpha_2\circ f(s_2,\dots,s_n,s_1) 
        \end{align*}
        for all $s_1,\dots, s_n\in S$.
        Suppose, on the contrary, that $f$ is not cyclic. 
        Then there are $s_1,\dots,s_n\in S$ such that $f(s_1,\dots,s_n)\neq f(s_2,\dots,s_n,s_1)$. 
        Let $<$ denote the linear order of $\struct{B}$ and set $\Tilde{s_i} \coloneqq (s_i,\dots,s_n,s_1,\dots,s_{i-1})$ for $i\leq n$.
        By assumption, we have $f(\Tilde{s_1})\neq f(\Tilde{s_2})$.
        Assume without loss of generality that $f(\Tilde{s_1})<f(\Tilde{s_2})$. 
        Since $\alpha_1,\alpha_2$ preserve $<$ and its complement, we have 
        \begin{align*}
            f(\Tilde{s_1})<f(\Tilde{s_2})\Leftrightarrow  \alpha_1(f(\Tilde{s_1}))<\alpha_1(f(\Tilde{s_2})) 
            \Leftrightarrow \alpha_2(f(\Tilde{s_2}))<\alpha_2(f(\Tilde{s_3})) 
            \Leftrightarrow&\enspace f(\Tilde{s_2})<f(\Tilde{s_3})\\
           \vdots \  &  \\
            \Leftrightarrow&\enspace f(\Tilde{s_n})<f(\Tilde{s_1}),
        \end{align*}
        a contradiction to $<$ being a linear order. Thus, $f$ is cyclic.
    \end{proof} 
 
The next lemma explains how height-1 identities can be lifted from full powers.
Lemma~\ref{lemma:lifting}\eqref{item:liftpol} essentially states that $\mathcal{I}(\struct{X})\coloneqq \fpwr{\struct{X}}{d}$ is a reduction from $\CSP(\struct{A})$ to $\CSP(\orbeq{\fpwr{\struct{A}}{d}}{\Aut(\struct{B})})$; we will, however, use it in the more general setting of PCSPs.
The fact that this reduction exists is a folklore result in infinite-domain constraint satisfaction, which was popularized by~\cite[Thm.~3]{bodirsky2016reducts} but in fact goes back to~\cite[Sect.~5.5]{bodirsky_pinsker_tsankov}.
We include a full proof of Lemma~\ref{lemma:lifting}; an analogous statement can be found in~\cite[Lem.~27]{Mottet_2025}. 
 
\begin{restatable}{lemma}{lifting}\label{lemma:lifting}  
Let $\struct{A}$ be a reduct of a finitely bounded homogeneous structure $\struct{B}$  and let $\struct{S}$ be a finite structure over the signature of $\struct{A}$.
Suppose that the bounds of $\struct{B}$ have size $\leq d$ and its relations
are of arity $\leq d-1$.  
\begin{enumerate}
    \item \label{item:liftpol} If there is a homomorphism $\smash{h\colon \fpwr{\struct{S}}{d}\to \orbeq{\fpwr{\struct{A}}{d}}{\Aut(\struct{B})}}$, then there exists a homomorphism $g\colon\struct{S}\to\struct{A}$ such that $h(\bar{x})=\orbeq{g(\bar{x})}{\Aut(\struct{B})}$ for every $\bar{x}\in S^d$. 
    \item \label{item:liftidentity} If 
    $\Sigma$ is a height-1 condition  satisfied in $\Pol(\fpwr{\struct{S}}{d},\orbeq{\fpwr{\struct{A}}{d}}{\Aut(\struct{B})})$, then $\Sigma$ is satisfied in $\Pol(\struct{S},\struct{A})$ modulo $\Aut(\struct{B})$.
\end{enumerate} 
 \end{restatable}  
 \begin{proof} 
We start with item~\ref{item:liftpol}.
        Let $\struct{S}$ and $h$ be as in the statement of the lemma. 
        
        First, we define an auxiliary structure $\struct{S}'$ on $S$ with the same signature as $\struct{S}$.
        For every $k$-ary relation $R^{\struct{B}}$ of $\struct{B}$, the relation $R^{\struct{S}'}$ is defined as the preimage of the relations $R_\iota^{\smash{\orbeq{\fpwr{\struct{B}}{d}}{_{\Aut(\struct{B})}}}}$ in the following sense:
        \begin{align}
              R^{\struct{S}'} \coloneqq \{   (x_{\iota(1)},\dots,x_{\iota(k)})  \in S^k \mid     \iota\colon [k]\to[d], &(x_1,\dots,x_d)\in S^d ,\nonumber\\       h&(x_1,\dots,x_d)\in R_\iota^{\smash{\orbeq{\fpwr{\struct{B}}{d}}{_{\Aut(\struct{B})}}}}\}.  \label{eq:exists}
        \end{align}  
        By the definition of full powers, for every $\iota\colon [k]\to[d]$ and all $x_1,\dots,x_d,y_{k+1},\dots,y_d\in S$, we have 
        \begin{align*}
             \big((x_{\iota(1)},\dots,x_{\iota(k)},y_{k+1}\dots,y_d),(x_{1}\dots,x_{k},x_{k+1}\dots,x_{d})\big)\in E^{\fpwr{\struct{S}}{d}}_{\mathrm{id},\iota},
        \end{align*}
        where $\mathrm{id}$ denotes the identity map on $[k]$. 
        This and the fact that $h$ is a homomorphism yield
        \begin{align}
            R^{\struct{S}'} = \{ (x_1,\dots,x_k)\in S^k \mid  \forall  x_{k+1},\dots,x_d\in S\colon  h(x_1,\dots,x_d)       \in R_{\mathrm{id}}^{\ \smash{\orbeq{\fpwr{\struct{B}}{d}}{_{\Aut(\struct{B})}}}}      \}.  \label{eq:forall}
        \end{align} 
        Next, we define an auxiliary equivalence relation $\sim$ on $S$ by relating two elements if they appear as entries of some $\bar{z}=(z_1,\dots, z_d)\in S^d$ that is mapped via $h$ to an orbit where the corresponding entries are equal:
        \begin{align}
            x\sim y :\Leftrightarrow   \ &          \exists\bar{z}\in S^d  \,\exists i_1,i_2\in [d]  \big( x=z_{i_1},\ y=z_{i_2}, \textnormal{ and }    \forall\bar{a}\in h(\bar{z})\colon  a_{i_1}=a_{i_2}\big). \label{eq:noreformulation} 
        \end{align}
        Here, $\bar{a} \in h(\bar{z})$ is to be understood in the sense that $h(\bar{z})$ represents the orbit of $\bar{a}$ under $\Aut(\struct{B})$. Note that, since $h$ preserves $E_{\iota,\iota'}$  for all $\iota,\iota'\colon [2]\to[d]$, eq.~\eqref{eq:noreformulation} is equivalent to
        \begin{align}
            x\sim y \Leftrightarrow   \forall \bar{z}\in S^d  \, \forall i_1,i_2\in [d]: \big(  x=z_{i_1} \text{ and }y=z_{i_2} \big) \Rightarrow \big(\forall\bar{a}\in h(\bar{z})\colon  a_{i_1}=a_{i_2}\big) .\label{eq:reformulation} 
        \end{align} 
        \vspace{-1em}
        \begin{claim} \label{claim:congruence}
            $\sim$ is a relational congruence on $\struct{S}'$.
        \end{claim}
        \begin{claimproof} 
            We must show that, for every $k$-ary relation $R^{\struct{S}'}$, every $i\in [k]$, and all $ y_1,\dots,y_{i-1},$ $y_{i+1},\dots,y_k,x_1,x_2\in S$, we have that: if  $x_1 \sim x_2$, then
         \begin{align*}
             R^{\struct{S}'}(y_1,\dots,y_{i-1},x_1,y_{i+1},\dots,y_k)\Leftrightarrow R^{\struct{S}'}(y_1,\dots,y_{i-1},x_2,y_{i+1},\dots,y_k) .
         \end{align*}
        To this end, define  
        \begin{align*}
            & \bar{y_1}  \coloneqq (y_1,\dots,y_{i-1},x_1,y_{i+1},\dots,y_k),  \quad\bar{z_1}  \coloneqq (y_1,\dots,y_{i-1},x_1,y_{i+1},\dots,y_k,x_2,\dots,x_2), \\ 
            &  \bar{y_2}  \coloneqq (y_1,\dots,y_{i-1},x_2,y_{i+1},\dots,y_k), \quad   
            \bar{z_2}  \coloneqq (y_1,\dots,y_{i-1},x_2,y_{i+1},\dots,y_k,x_1,\dots,x_1).  
        \end{align*}
        Here, the arity of $\bar{z}_1$ and $\bar{z}_2$ is $d$.
        As before, we denote by $\mathrm{id}$ the identity map on $[k]$.
        Additionally, $(i\mapsto k+1)$ denotes the map from $[k]$ to $[d]$ sending $i$ to $k+1$ and fixing all other elements.
        This map is well-defined since all relations are of arity $\leq d-1$.
        By eq.~\eqref{eq:reformulation}, for all $\bar{a}=(a_1,\dots, a_d)\in h(\bar{z}_1)\cup h(\bar{z}_2)$, we have  $a_i=a_{k+1}$. 
        Hence: 
        
        \[ \setlength{\arraycolsep}{0.15em}
        \begin{array}{cclcll} 
           \bar{y_1} \in  R^{\struct{S}'} &   
           \overset{\eqref{eq:exists}\&\eqref{eq:forall}}{\Longleftrightarrow} & h(\bar{z_1}) \in R_{\mathrm{id}}^{\, \smash{\orbeq{\fpwr{\struct{B}}{d}}{_{\Aut(\struct{B})}}}}  &  & \\ 
           &    \overset{\text{def.}}{\Longleftrightarrow}  & \exists \bar{a} \in h(\bar{z_1}) \cap R_{\mathrm{id}}^{\fpwr{\struct{B}}{d}} &
    \overset{\scriptscriptstyle{a_i=a_{k+1}}}{\Longleftrightarrow} &  \exists \bar{a}\in h(\bar{z_1}) \cap R_{(i\mapsto k+1)}^{\fpwr{\struct{B}}{d}}  &   \\ 
  &  & &    \overset{\text{def.}}{\Longleftrightarrow}  & h(\bar{z_1}) \in R_{(i\mapsto k+1)}^{\orbeq{\fpwr{\struct{B}}{d}}{_{\Aut(\struct{B})}}} &     \\
          &  &  &   \overset{\text{def.}}{\Longleftrightarrow}  &  h(\bar{z_2}) \in R_{\mathrm{id}}^{\, \smash{\orbeq{\fpwr{\struct{B}}{d}}{_{\Aut(\struct{B})}}}}  &    \\
        &  &  &     \overset{\mathclap{\eqref{eq:exists}\&\eqref{eq:forall}}}{\Longleftrightarrow} & \bar{y_2}\in  R^{\struct{S}'}.  &  
        \end{array}  
     \]
        This is what we had to show.
        \end{claimproof}  
 
     Finally, we combine the two constructions from before by setting $\Tilde{\struct{S}} \coloneqq \struct{S}' /{\sim}$.
  
        \begin{claim} \label{claim:embedding} $\Tilde{\struct{S}}$ embeds into $\struct{B}$.
        \end{claim} 
        \begin{claimproof} 
        Suppose, on the contrary, that this is not the case.
        Then there exists some bound $\struct{F}$ of $\struct{B}$ that embeds into $\Tilde{\struct{S}}$.
        Without loss of generality, $F=\{[x_1]_\sim,\dots,[x_n]_\sim\}\subseteq \Tilde{S}$.
        Since all bounds have size $\leq d$, we have $n\leq d$.
        As before, we shall use the abbreviations $\bar{x}\coloneqq (x_1,\dots, x_d)$ and $\bar{a}\coloneqq (a_1,\dots, a_d)$.
        Now, for every $k$-ary relation $R^{\Tilde{\struct{S}}}$ and every $\iota\colon [k]\to[n]$, we have  
        \begin{align*}
            ([x_{\iota(1)}]_\sim,\dots,[x_{\iota(k)}]_\sim)\in R^{\Tilde{\struct{S}}}  &\Leftrightarrow \exists x_1'\in[x_{\iota(1)}]_\sim,\dots,x_k'\in[x_{\iota(k)}]_\sim\colon 
 (x_1',\dots,x_k') \in R^{\struct{S}'} \\
            &\Leftrightarrow \smash{(x_{\iota(1)},\dots,x_{\iota(k)}) \in R^{\struct{S}'}} \\
            &\Leftrightarrow \forall x_{n+1},\dots,x_d\in S\colon    h(\bar{x}) \in R_{\iota}^{\smash{\orbeq{\fpwr{\struct{B}}{d}}{_{\Aut(\struct{B})}}}} \\
            &\Leftrightarrow \forall x_{n+1},\dots,x_d\in S\,\forall \bar{a} \in h(\bar{x})\colon \bar{a} \in  R_{\iota}^{\fpwr{\struct{B}}{d}} \\
            &\Leftrightarrow \forall x_{n+1},\dots,x_d\in S\,\forall \bar{a}\in h(\bar{x})\colon  (a_{\iota(1)},\dots,a_{\iota(k)}) \in R^{\struct{B}}.
        \end{align*}
        By fixing some $x_{n+1},\dots,x_d\in S$ and setting $e([x_i]_\sim) \coloneqq a_i$ for $i\leq n$ and some $\bar{a}\in h(\bar{x})$, we have that $e$ is an embedding of $\struct{F}$ into $\struct{B}$, a contradiction.
        \end{claimproof}
         
        Let $e'$ be an embedding from $\Tilde{\struct{S}}$ to $\struct{B}$, which exists by Claim~\ref{claim:embedding}.
        It follows from Claim~\ref{claim:congruence} that $e'$ induces a homomorphism $g\colon \struct{S}'\to\struct{B}$ via $g(x) \coloneqq e'([x]_\sim)$.
        \begin{claim} \label{claim:orbits}
            $h(\bar{x})=\orbeq{g(\bar{x})}{\Aut(\struct{B})}$ for every $\bar{x}\in S^d$.
        \end{claim}
        \begin{claimproof}
         Let $i,j\in [d]$ be arbitrary. 
          If every $\bar{y} \in h(\bar{x})$ satisfies $y_i=y_j$, then $x_i\sim x_j$ and thus $g(x_i)=g(x_j)$.
          On the other hand, if $g(x_i)=g(x_j)$, then $x_i\sim x_j$ since $e'$ is an embedding.
          Thus, every $\bar{y} \in h(\bar{x})$ satisfies $y_i=y_j$.
          Moreover, for all $x_1,\dots,x_d\in S$, every $k$-ary relation $R^{\struct{S}'}$, and every $\iota\colon [k]\to[d]$, we have
        \begin{align*}
            h(\bar{x}) \in R_{\iota}^{\smash{\orbeq{\fpwr{\struct{B}}{d}}{_{\Aut(\struct{B})}}}} &\Leftrightarrow (x_{\iota(1)},\dots,x_{\iota(k)}) \in R^{\struct{S}'} \\
            & \Leftrightarrow \smash{([x_{\iota(1)}]_\sim,\dots,[x_{\iota(k)}]_\sim) \in R^{\Tilde{\struct{S}}}} \\
            &\Leftrightarrow (g(x_{\iota(1)}),\dots, g(x_{\iota(k)})) \in R^{\struct{B}} \\
            &\Leftrightarrow \orbeq{(g(x_1),\dots, g(x_d))}{\Aut(\struct{B})} \in R_{\iota}^{\smash{\orbeq{\fpwr{\struct{B}}{d}}{_{\Aut(\struct{B})}}}}. 
        \end{align*}
        Now the statement of the claim follows from the homogeneity of $\struct{B}$.
        \end{claimproof}

       \begin{claim}
           $g$ is a homomorphism from $\struct{S}$ to $\struct{A}$.
       \end{claim}
        \begin{claimproof}
            Let $R\in \tau$ and $(x_1,\dots,x_k)\in R^\struct{S}$ be arbitrary.
        Since $h$ is a homomorphism, we have $ h(x_1,\dots,x_d) \in R_{\mathrm{id}}^{\,\smash{\orbeq{\fpwr{\struct{A}}{d}}{_{\Aut(\struct{B})}}}} $ for all $x_{k+1},\dots,x_d\in S$. 
        Since every automorphism of $\struct{B}$ is also an automorphism of $\struct{A}$, by Claim~\ref{claim:orbits}, we have that $(g(x_1),\dots,g(x_k))\in R^{\struct{A}} $.
        \end{claimproof}  

        We continue with item~\ref{item:liftidentity}.
    It is not hard to show that $\fpwr{(\struct{S}^n)}{d} \simeq (\fpwr{\struct{S}}{d})^n$ (cf.~\cite[Cor.~22]{Mottet_2025}) for all $d,n\in \mathbb{N}$.
        Hence, every homomorphism $\smash{h\colon (\fpwr{\struct{S}}{d})^{n}\to \orbeq{\fpwr{\struct{A}}{d}}{\Aut(\struct{B})}}$ is also a homomorphism from $\fpwr{(\struct{S}^{n})}{d}$ to $\orbeq{\fpwr{\struct{A}}{d}}{\Aut(\struct{B})}$.
        Suppose that $h$ is a function symbol appearing in $\Sigma$ that is assigned an element of $\Pol(\fpwr{\struct{S}}{d},\orbeq{\fpwr{\struct{A}}{d}}{\Aut(\struct{B})})$. 
        By item~\ref{item:liftpol}, there is a homomorphism $g\colon \struct{S}^{n}\to \struct{A}$ such that 
        \begin{align*}
            h\begin{pmatrix}\bar{s_1} \\\vdots\\\bar{s_d}\end{pmatrix}=\orbeq{\begin{pmatrix}g(\bar{s_1})\\\vdots\\ g(\bar{s_d}) \end{pmatrix}}{\Aut(\struct{B})}   
        \end{align*}
        for all $\bar{s_1},\dots,\bar{s_d} \in S^{n}$.
        For each function symbol in $\Sigma$, we fix such a homomorphism. 
        Suppose now that $h_1,h_2\in  \Pol(\fpwr{\struct{S}}{d},\orbeq{\fpwr{\struct{A}}{d}}{\Aut(\struct{B})})$ witness the satisfaction of the identity 
        \begin{align*}
            h_1(x_{i_1},\dots,x_{i_{n}})\approx h_2(x_{i_{n+1}},\dots,x_{i_{n+k}}),
        \end{align*}
        where all $i_\ell\leq m$ for some $m\in\mathbb{N}$, and let $g_1:\struct{S}^{n}\to \struct{A}$ and $g_2:\struct{S}^{k}\to \struct{A}$ be the associated homomorphisms.

        Let $\bar{s}_1,\dots,\bar{s}_{m}$ be arbitrary tuples over $S$ of arity $k$ large enough so that, for every  $\bar{s}\in S^{m}$, there exists $\ell \in [k]$ such that $\bar{s}$ equals the $\ell$-th row in $(\bar{s}_1,\dots,\bar{s}_{m})$ (viewed as a $k\times m$-matrix).
        Define the $m$-ary operations $\Tilde{g_1}$ and $\Tilde{g_2}$ by 
        \[
         \Tilde{g_1}(x_1,\dots,x_m) \coloneqq g_1(x_{i_1},\dots,x_{i_n}) \qquad \text{and}  \qquad \Tilde{g_2}(x_1,\dots,x_m) \coloneqq g_2(x_{i_{n+1}},\dots,x_{i_{n+k}}).
        \] 
        Then  $\Tilde{g_1}(\bar{s}_1,\dots,\bar{s}_{m})$ and $\Tilde{g_2}(\bar{s}_1,\dots,\bar{s}_m)$ satisfy the same equalities and relations w.r.t.~$\struct{B}$ on all $d$-subtuples.
        Since all relations of $\struct{B}$ have arity $\leq d$ and $\struct{B}$ is homogeneous, $\Tilde{g_1}(\bar{s}_1,\dots,\bar{s}_{m})$ and $\Tilde{g_2}(\bar{s}_1,\dots,\bar{s}_m)$ lie in the same orbit.
        Thus, there are $\alpha_1,\alpha_2\in\Aut(\struct{B})$ such that
        \begin{align*}
            \alpha_1\circ\Tilde{g_1}(\bar{s}_1,\dots,\bar{s}_{m})=\alpha_2\circ\Tilde{g_2}(\bar{s}_1,\dots,\bar{s}_{m}).
        \end{align*}
        This is equivalent to
        \begin{align*}
            \alpha_1\circ g_1(s_{i_1},\dots,s_{i_n})=\alpha_2\circ g_2(s_{i_{n+1}},\dots,s_{i_{n+k}})
        \end{align*}
        for all $s_1,\dots,s_{m}\in S$. Thus, $\Sigma$ is satisfied in $\Pol(\struct{S},\struct{A})$ modulo $\Aut(\struct{B})$. 
    \end{proof}

Now we can finally prove Proposition~\ref{prop:sandwiches}.
\begin{proof}[Proof of Proposition~\ref{prop:sandwiches}] Clearly, for every finite substructure $\struct{S}$ of $\struct{A}$, we have $ \fpwr{\struct{S}}{d}  \rightarrow \fpwr{\struct{A}}{d} \rightarrow \orbeq{\fpwr{\struct{A}}{d}}{\Aut(\struct{B})}$.
Let $S\subseteq A$ be arbitrary with $3 \leq |S| < \infty$.
Suppose that there is a finite cheese $\struct{D}$ of $(\fpwr{\struct{S}}{d},\orbeq{\fpwr{\struct{A}}{d}}{\Aut(\struct{B})})$ that does not pp-construct $\struct{K}_3$. 
Then, by Theorem \ref{thm:barto12}, there exists a cyclic operation $f$ in $\Pol(\struct{D})$ of some arity $n$. 
Take any homomorphisms $g\colon \fpwr{\struct{S}}{d} \rightarrow \struct{D}$ and $h\colon \struct{D} \rightarrow \orbeq{\fpwr{\struct{A}}{d}}{\Aut(\struct{B})}$. 
Then we get a cyclic $n$-ary $f'\in \Pol(\fpwr{\struct{S}}{d},\orbeq{\fpwr{\struct{A}}{d}}{\Aut(\struct{B})})$ by setting 
\[
    f'(x_1,\dots,x_n) \coloneqq h\circ f(g(x_1),\dots,g(x_n)).
\]
  By Lemma~\ref{lemma:lifting}\eqref{item:liftidentity} there is some $n$-ary $g\in\Pol(\struct{S},\struct{A})$ that is pseudo-cyclic modulo $\Aut(\struct{B})$. 
By Lemma~\ref{lemma:pseudocyclic}, this operation is cyclic. Take pairwise distinct elements $s_1,s_2,s_3\in S$. 
We now show that the cyclicity of $g$ contradicts the fact that it preserves $\neq$. 
To this end, we consider the following three cases:

\medskip \emph{$n\equiv 0  \bmod 3$:} $ 
    g(s_1,s_2,s_3,\dots,s_1,s_2,s_3) = g(s_2,s_3,s_1,\dots,s_2,s_3,s_1) $
    
    \medskip \emph{$n\equiv 1  \bmod 3$:} 
    $ 
    g(s_1,s_2,s_3,\dots,s_1,s_2,s_3,s_1) = g(s_3,s_1,s_2,\dots,s_3,s_1,s_1,s_2)
    $ 
    
    \medskip \emph{$n\equiv 2  \bmod 3$:}
    $ 
    g(s_1,s_2,s_3,\dots,s_1,s_2,s_3,s_1,s_2) = g(s_2,s_3,s_1,\dots,s_2,s_3,s_1,s_2,s_1)
    $ 
    
    \medskip
In all three cases, $g$ does not preserve $\neq$.
 Hence, the  assumption that such $\struct{D}$ exists was wrong, and consequently $\PCSP(\fpwr{\struct{S}}{d},\orbeq{\fpwr{\struct{A}}{d}}{\Aut(\struct{B})})$ is not finitely tractable.
\end{proof}

 The prerequisites of Proposition~\ref{prop:sandwiches} are fairly general but clearly not sufficient to cover the entire scope of Conjecture~\ref{conj:bodirsky_pinsker} (cf.~Section~\ref{section:preprocessing}). Given an arbitrary reduct $\struct{A}$ of a finitely bounded homogeneous structure $\struct{B}$, one does not generally have an inequality relation on $\struct{A}$ or a linear order on $\struct{B}$; this is where Theorem~\ref{thm:removing_algebraicity} becomes useful.
On the one hand, removing algebraicity from $\struct{A}$ using Theorem~\ref{thm:removing_algebraicity} allows us to expand our structure by the binary inequality relation without fundamentally changing the complexity of its CSP.
On the other hand, removing algebraicity from $\struct{B}$ allows us to take the generic superposition of $(\mathbb{Q};<)$ and $\struct{B}$ (see Proposition~\ref{prop:generic_superpositions}), which can be seen as expanding $\struct{B}$ by a linear order without losing the property of being finitely bounded homogeneous.

\begin{proposition}[Section~2.3.6~in~\cite{bodirsky2021complexity}] \label{prop:generic_superpositions} Let $\struct{A}_1, \struct{A}_2$ be  countable homogeneous structures without algebraicity over disjoint finite relational signatures $\tau_1,\tau_2$.
Then there exists an up to isomorphism unique countable homogeneous $(\tau_1\cup\tau_2)$-structure $\struct{A}_1\astlarge \struct{A}_2$, called the \emph{generic superposition} of $\struct{A}_1, \struct{A}_2$, such that any finite $(\tau_1\cup\tau_2)$-structure embeds into  $\struct{A}_1\astlarge \struct{A}_2$ if and only if its $\tau_i$-reduct embeds into  $\struct{A}_i$ for both $i\in\{1,2\}$. Moreover, if $\struct{A}_1,\struct{A}_2$ are both finitely bounded, then so is $\struct{A}_1\astlarge \struct{A}_2$.
\end{proposition} 

Finally, we prove Theorem~\ref{thm:main_theorem_sandwiches}. 
%
On the way, we also explain which part of the proof needs to be adjusted in order to obtain a proof of Theorem~\ref{thm:main_theorem_sandwiches_basic}.

\begin{proof}[Proof of Theorem~\ref{thm:main_theorem_sandwiches}]
Let $\struct{A}$ be a non-contractible reduct of a finitely bounded homogeneous structure $\struct{B}$. 
    
   First, we move to the structures without algebraicity $\blowup{\struct{A}}$ and $\ordblowup{\struct{B}}$ from Theorem~\ref{thm:removing_algebraicity}; recall that $\struct{\blowup{\struct{A}}}$ pp-constructs $\struct{K}_3$ if and only if $\struct{A}$ does and that $\CSP(\struct{A})$ and $\CSP\blowup{\struct{A}}$ are Datalog-interreducible.
    Next, we take the generic superposition of $\smash{\ordblowup{\struct{B}}}$ with $(\mathbb{Q};<)$, which exists by Proposition~\ref{prop:generic_superpositions}.
    We set \[\smash{\widehat{\struct{B}}\coloneqq  \ordblowup{\struct{B}} \astlarge (\mathbb{Q};<)}\] and define $\widehat{\struct{A}}$ as the reduct of $\widehat{\struct{B}}$ to the signature of $\blowup{\struct{A}}$. 
    For Theorem~\ref{thm:main_theorem_sandwiches_basic}, we would instead keep $\widehat{\struct{B}}\coloneqq \struct{B}$.  
    
     By Proposition~\ref{prop:generic_superpositions}, we have that $\widehat{\struct{A}}$ and $\blowup{\struct{A}}$ are isomorphic and $\widehat{\struct{B}}$ is finitely bounded homogeneous.
    But now $\widehat{\struct{A}}$ and $\widehat{\struct{B}}$ satisfy the prerequisites of Proposition~\ref{prop:sandwiches}.
    Let $\mathcal{J}$ be the composition of the Datalog reduction $\mathcal{I}$ from $\struct{A}$ to $\blowup{\struct{A}}\cong \widehat{\struct{A}}$ from Theorem~\ref{thm:removing_algebraicity} and the pp-interpretation of $\fpwr{\widehat{\struct{A}}}{d}$ from $\widehat{\struct{A}}$ and let $\mathcal{J}'$ be the composition of the pp-interpretation and Datalog reduction running in the opposite direction.
    The fact that $\mathcal{J}$ and $\mathcal{J}'$ are Datalog-reductions from $\CSP(\struct{A})$ to $\CSP(\struct{A}')$ and vice versa follows directly from Theorem~\ref{thm:removing_algebraicity} and Proposition~\ref{prop:basic_properties_full_powers}\eqref{item:fpwrppconstr}. 
    
    We select a finite substructure $\struct{S}$ of $\struct{A}$ with $|S|\geq 3$ large enough so that every finite structure which has a homomorphism to a substructure of $\struct{A}$ of size $\leq n$ also has a homomorphism to $\struct{S}$ (this can be done since $\struct{A}$ has finite signature). Now let $\widehat{\struct{S}}$ be the substructure of $\widehat{\struct{A}}$ induced by $\{0\}\times S$.
    We claim that, for $d$ chosen as in Proposition~\ref{prop:sandwiches}, the structures
    \begin{itemize}
        \item $\struct{A}'\coloneqq \fpwr{\widehat{\struct{A}}}{d}$,
        \item $\struct{B}'\coloneqq \fpwr{\widehat{\struct{B}}}{d}$,
        \item $\struct{S}_1\coloneqq  \fpwr{\widehat{\struct{S}}}{d}$, and
        \item $ \struct{S}_2\coloneqq \orbeq{\fpwr{\widehat{\struct{A}}}{d}}{\Aut(\widehat{\struct{B}})}$
    \end{itemize} 
    witness items~\ref{item:sandwichthm5}--\ref{item:sandwichthm3} of the theorem. 

    Item~\ref{item:sandwichthm5} follows immediately from the definition of full powers.

    For item~\ref{item:sandwichthm6}, note that if $\struct{X}$ maps homomorphically to a substructure of $\struct{A}$ of size $\leq n$, then $\struct{X}\rightarrow \struct{S}$.
    Since the Datalog reduction $\mathcal{I}$ from Theorem~\ref{thm:removing_algebraicity} is just the expansion of a structure by empty relations $\neq$ and $E$, we have $\mathcal{I}(\struct{X})\rightarrow \widehat{\struct{S}}$.
    Since every homomorphism between two structures clearly induces a homomorphism between their corresponding full powers, we have $\mathcal{J}(\struct{X})=\fpwr{\mathcal{I}(\struct{X})}{d}\rightarrow  \fpwr{\widehat{\struct{S}}}{d}= \struct{S}_1$.
    On the other hand, if $\struct{X}\nrightarrow \struct{A}$, then $\mathcal{I}(\struct{X})\nrightarrow \widehat{\struct{A}}$ by Theorem~\ref{thm:removing_algebraicity}.
    By Lemma~\ref{lemma:lifting}\eqref{item:liftpol}, if $\fpwr{\mathcal{I}(\struct{X})}{d}\rightarrow \struct{S}_2$, then also $\mathcal{I}(\struct{X})\rightarrow\widehat{\struct{A}}$, hence $\mathcal{J}(\struct{X})=\fpwr{\mathcal{I}(\struct{X})}{d}\nrightarrow \struct{S}_2$. 
    
    Item~\ref{item:sandwichthm4} follows Proposition~\ref{prop:sandwiches}.
    We remark that, if we have kept  $\widehat{\struct{B}}\coloneqq \struct{B}$ for the purpose of proving Theorem~\ref{thm:main_theorem_sandwiches_basic}, then we instead get that $\struct{A}$ and $\struct{A}'$ pp-interpret each other due to Proposition~\ref{prop:basic_properties_full_powers}\eqref{item:fpwrppconstr}.
    
    Item~\ref{item:sandwichthm1} follows from Theorem~\ref{thm:removing_algebraicity}\eqref{AppconstructK3} and Proposition~\ref{prop:basic_properties_full_powers}\eqref{item:fpwrppconstr}. 
    
    Item~\ref{item:sandwichthm2} follows from Theorem~\ref{thm:removing_algebraicity}\eqref{Amodelcomplcore} and Proposition~\ref{prop:basic_properties_full_powers}\eqref{item:fpwrmodcompl}.

    Item~\ref{item:sandwichthm3} follows from Theorem~\ref{thm:removing_algebraicity}\eqref{item:BhomogRamsey}, Proposition~\ref{prop:basic_properties_full_powers}\eqref{item:fpwrramsey}, and~\cite[Thm.~1.5]{bodirsky2014new}. 
\end{proof}

\subsection{Proof of Theorem~\ref{thm:temporal_situation}} 
The main object of our focus is $(\mathbb{Q};X)$, where
\[X=\{(x,y,z) \in \mathbb{Q}^3 \mid x=y<z \vee y=z<x \vee z=x<y\}.\]
In~\cite[Sect.~4.2]{bodirsky2022descriptive}, it was shown that $\CSP(\mathbb{Q};X)$ is inexpressible in FPC.
 The idea is that $\CSP(\mathbb{Q};X)$ can be reformulated as a decision problem for systems of mod-2 equations:

\medskip    
\setlength{\tabcolsep}{0pt}    
\hspace{1em}\begin{tabular}{ll}  
  INSTANCE: & \,  A finite homogeneous system $S$ of mod-2 equations of length 3. \\ 
      QUESTION: & \, Does each non-empty subset of $S$ have a solution other than $\bar{0}$?  
\end{tabular}   

\medskip
The above reformulation was obtained in~\cite[Prop.~4.12]{bodirsky2022descriptive} as a direct consequence of the correctness of the polynomial-time algorithm for~$\CSP(\mathbb{Q};X)$ from~\cite{bodirsky2022descriptive} (see Figure~\ref{algo:QX}). 
The algorithm runs a loop where, in each step, it eliminates the potentially minimal elements of the remaining variables (those that are assigned $1$ after interpreting the relational constraints as mod-2 equations) and inductively builds a satisfying assignment.

\begin{figure}[ht!]
\renewcommand{\algorithmiccomment}[1]{// #1} 
  \hrule 
  \smallskip
    \begin{algorithmic} 
\STATE{INPUT: An instance $\struct{A}$ of $\CSP(\mathbb{Q};X)$}
\STATE{OUTPUT: \texttt{true} or \texttt{false}}
\IF{$A=\emptyset$}
\RETURN  \texttt{true} 
\ENDIF 
\FOR{$x \in A$}    
\STATE{\COMMENT{interpret each $(x_i,x_j,x_k)\in X^{\struct{A}}$  as $x_i+x_j+x_k=0$}}
\IF{$\struct{A}\cup \{x=1\}$ mod 2 has a solution $s$}
\RETURN divide-and-conquer$(\struct{A}[s^{-1}(0)])$
\ENDIF
\ENDFOR 
\RETURN \texttt{false} 
\end{algorithmic} 
 \smallskip
  \hrule
 \caption{A divide-and-conquer algorithm for $\CSP(\mathbb{Q};X)$}\label{algo:QX}
\end{figure}

A similar reformulation can be obtained for other relations pp-definable in $(\mathbb{Q};X)$, e.g., equations $x_1+\cdots + x_i=0$ in the case of the following higher-ary generalizations:
\[
X_i = \{(x_1,\dots, x_i) \in \mathbb{Q}^i \mid \ |\{j\in [i] \mid x_j=\min(x_1,\dots, x_i)\}|=0 \bmod 2  
\}.
\]
 
It is well-known that the class of all satisfiable systems of mod-2 equations cannot be defined in FPC~\cite{atserias2009affine}. 
However, the fact that $\CSP(\mathbb{Q};X)$ is not solvable in FPC does not immediately follow from the aforementioned result about mod-2 equations.
Instead, one needs to specifically prove the inexpressibility in FPC for overdetermined systems of equations, i.e., those whose homogeneous companion only has the trivial solution $\bar{0}$. 
Such systems are then only one step away from fixed-image  fooling instances of $\CSP(\mathbb{Q};X)$ for FPC.
The following auxiliary result was obtained in~\cite{bodirsky2022descriptive}. 

 \begin{lemma}[Proof of Theorem~4.23 in~\cite{bodirsky2022descriptive}] \label{lemma:fooling}
    For every $k\in \mathbb{N}$, there exist structures $\struct{A}_k,\struct{B}_k$ in the signature consisting of symbols $R_i$ with arities $i$ for $i\in \{2,3,4,5\}$ with $\struct{A}_k \equiv_{k} \struct{B}_k$ and such that the following holds if $R_i(x_1,\dots, x_i)$ is interpreted as the equation $x_1+\cdots+x_i = 0 \bmod 2$:
    \begin{enumerate}
        \item \label{item:foolingold1} $\struct{A}_k$ has a solution mapping in each equation at least one variable to $1$;
        \item \label{item:foolingold2} $\struct{B}_k$ only has the trivial solution $\bar{0}$.  
    \end{enumerate}
 \end{lemma}
Since $\CSP(\mathbb{Q};X_2,X_3,X_4,X_5)$ admits fixed-image fooling instances for FPC and is pp-definable from $(\mathbb{Q};X)$~\cite{bodirsky2022descriptive}, it follows from Theorem~\ref{thm:main_theorem_sandwiches_basic} that $(\mathbb{Q};X)$ pp-constructs a finite-domain PCSP inexpressible in FPC.
This was essentially already shown in~\cite[Prop.~38]{Mottet_2025}, except that the author did not comment on how to push the fixed-image fooling instances from $\CSP(\mathbb{Q};X_2,X_3,X_4,X_5)$ to $\CSP(\mathbb{Q};X)$ or the fact that the right-hand side of the (half infinite) PCSP can be chosen finite if one additionally applies the full-power construction. 
We now give a full proof of Theorem~\ref{thm:temporal_situation}.

\begin{proof}[Proof of Theorem~\ref{thm:temporal_situation}] Let $\struct{B}$ be a first-order reduct of $(\mathbb{Q};<)$. By~\cite[Thm.~1.3]{bodirsky2022descriptive}, we have two cases.
If $\struct{B}$ is expressible in FPC, then it cannot pp-construct any finite-domain PCSP inexpressible in FPC, because pp-constructions between (P)CSP-templates give Datalog-reductions between the (P)CSPs.
Otherwise, $\struct{B}$ pp-constructs $\struct{K}_3$ or $(\mathbb{Q};X)$.
In the former case,  the statement follows immediately because $\CSP(\struct{K}_3)$ is inexpressible in FPC.
In the latter case, recall that $(\mathbb{Q};X)$ pp-defines $(\mathbb{Q};X_2,X_3,X_4,X_5)$.
Then the statement follows from the fact that $\CSP(\mathbb{Q};X_2,X_3,X_4,X_5)$ admits fixed-image fooling instances together with Theorem~\ref{thm:main_theorem_sandwiches_basic} for $n$ chosen large enough. \end{proof}

\section{Case study: phylogeny CSPs} \label{section:phylo}

In the present section, we introduce phylogeny CSPs and give a proof of Theorem~\ref{thm:temporal_phylo_pp_constr}.

\subsection{Phylogeny CSPs}
Phylogeny problems are motivated by evolutionary biology.
Every species uniquely stems from a prior species, and all species have a common ancestor.
A natural question in the area is how to reconstruct the full phylogenetic tree using only observations about recent species, i.e., the leaves of the tree.
For the leaves $x, y, z$ of a rooted tree, we write $x|yz$ if the youngest common ancestor of $y$ and $z$ lies strictly below the youngest common ancestor of $x$, $y$ and $z$. 
The set of all leaves of a rooted tree together with this ternary relation is called the leaf structure of the tree. 

The basic \emph{phylogeny (decision) problem} asks whether, for a given set of constraints of the form
$x_i|x_jx_k$ over variables $V$, there exists a mapping $s$ from $V$ to the leaves of some rooted tree such that $s(x_i)|s(x_j)s(x_k)$ holds for every constraint $x_i|x_jx_k$.
It is known that this problem can be solved in quadratic time using a simple divide-and-conquer algorithm~\cite{bodirsky2017complexity}.
Matters get more difficult, however, when we allow Boolean combinations as constraints.
For example, consider the ternary relation $N$ defined by the formula $(x|yz \vee z|xy)$. 
The problem whether there exists a rooted tree whose leaves satisfy a given set of constraints of
the form $N(x_i,x_j,x_k)$ is NP-hard~\cite{bodirsky2017complexity}. The class of all phylogeny problems is obtained by considering all sets of relations which are specified by Boolean combinations of formulas of the form $(x|yz)$ or $(x = y)$.

 The key to a successful complexity classification for phylogeny problems was the observation that they can be viewed as CSPs of $\omega$-categorical structures; in fact, they fall within the scope of Conjecture~\ref{conj:bodirsky_pinsker}.
 Given a finite %
 rooted tree $\struct{T}$ (i.e.\;a directed graph without loops that has at most one vertex without incoming edges (the root) and whose undirected graph reduct is acyclic), we denote by $L(\struct{T})$ the set of all its leaves (vertices without outgoing edges); for our purposes, we may assume that all trees have branching factor two.
 On $L(\struct{T})$, we then consider the ternary relation $x|yz$ from above and call $(L(\struct{T});|\,)$ the \emph{leaf structure} of $\struct{T}$.
 It is well-known that $\struct{T}$ is uniquely determined by its leaf structure 
 $(L(\struct{T});|\,)$ up to isomorphism~\cite[Thm.~3]{steel1992complexity}.
 The class of the leaf structures of all finite rooted trees forms a Fra\"i{ss}\'{e} class~\cite[Prop.~7]{bodirsky2016reductsofc}, and hence there exists an up to isomorphism unique homogeneous structure $\struct{L}$ whose age consists of all possible leaf structures.\footnote{In literature on model theory, this structure is commonly known as the homogeneous $C$-relation.} 
 One can moreover show that $\struct{L}$ is finitely bounded~\cite[Sect.~5.1.2]{bodirsky2021complexity}:
 \begin{align*}
     \forall x,y,z,u \big( \neg (x|yz \wedge y|xz)  
     \wedge (x|yz \Rightarrow x|zy) \wedge  (x|yz \Rightarrow x|uz \vee u|yz) & \\ {} \wedge (x\neq y \Rightarrow x|yy) \wedge 
     (x|yz \vee xy|z \vee xz|y \vee x=y=z) & 
     \big)
 \end{align*} 
 We denote the domain of $\struct{L}$ by $\mathbb{L}$, and write $\struct{L}=(\mathbb{L};|\,)$.
 
 The basic phylogeny problem can equivalently be phrased as $\CSP(\mathbb{L};|\,)$, and arbitrary phylogeny problems correspond to CSPs of first-order reducts of $(\mathbb{L};|\,)$. 
 It is useful to view $\mathbb{L}$ as the set of all $\{0,1\}$-sequences with only finitely many non-zero entries and where $xy|z$ holds if and only if there exists $n\in \mathbb{N}$ such that the first index where $x$ and $y$ differ is $n$ and the first index where $x$ or $y$ differ from $z$ is less than $n$.
 In the following, we are interested in the specific case of the four-ary relation $Y$ defined by the formula 
 \begin{align}
    Y\coloneqq(xy|wz \vee xw|yz \vee xz|yw). \label{eq:even_split}
 \end{align} 
 Here, the notation $L_1|L_2$ for sets of leaves $L_1,L_2$ naturally generalizes the original ternary relation $|$.
 Intuitively, it states that the youngest common ancestors of $L_1$ and $L_2$ are incomparable.
 Formally,  $L_1|L_2$ if and only if $x|yz$ for all $x\in L_i$ and $y,z\in L_{3-i}$ ($i\in \{1,2\})$.
 Note that, for every finite $F\subseteq \mathbb{L}$ with $|F|\geq 2$, there exist unique non-empty subsets  $F_0,F_1\subseteq F$ with $F=F_0\cup F_1$ and $F_0|F_1$; these subsets represent the leaves of the left and the right branches of the rooted binary tree encoded by $\struct{L}[F]$.
 We call $\{F_0,F_1\}$ the \emph{split-partition} of $F$. 
 Since $(\mathbb{L};|)$ and $(\mathbb{L};Y)$ are first-order interdefinable using quantifier-free formulas and $(\mathbb{L};|)$ is homogeneous, it follows that $(\mathbb{L};Y)$ is homogeneous as well.

A \emph{split} of a tuple $(t_1,\dots, t_n)$ over $\mathbb{L}$ is any $\{0,1\}$-tuple $(s_1,\dots, s_n)$ with the property that, for all $i,j,k\in [n]$, we have $t_i|t_jt_k$ whenever $s_i \neq s_j=s_k$~\cite{bodirsky2017complexity} (note that besides the constant tuples there is at most one split up to swapping 0's and 1's).
Note that the relation $Y$ consists of all $4$-tuples over $\mathbb{L}$ with a non-constant even split in $\struct{L}$, i.e., the split contains a non-zero even number of $0$s and $1$s.
 The formula in eq.~\eqref{eq:even_split} and more generally all formulas defining sets of tuples with only even splits fall within the \emph{affine Horn} fragment of phylogeny languages as defined in~\cite{bodirsky2017complexity}.
 By the results obtained therein, the CSPs of all affine Horn phylogeny languages, and in particular $\CSP(\mathbb{L};Y)$, are solvable in polynomial time, using an algorithm based on Gaussian elimination and local consistency checking. 
 
 What we show here is that there exists a relation $Y'$ pp-definable in $(\mathbb{L};Y)$ such that $\CSP(\mathbb{L};Y,Y')$ is inexpressible in FPC and that this is witnessed by fixed-image fooling instances.
 We then use this fact to show the existence of a finite-domain PCSP template $(\struct{B}_1,\struct{B}_2)$ pp-constructible from $(\mathbb{L};Y)$ such that $\PCSP(\struct{B}_1,\struct{B}_2)$ is inexpressible in FPC.
 Finally, we prove that there is no finite-domain CSP template with these properties.
 
 \subsection{A tractable phylogeny CSP not in FPC} \label{sec:phylo_not_in_FPC}
 We show that $\CSP(\mathbb{L};Y)$ is inexpressible in FPC.
 Our proof of FPC-inexpressibility %
 builds on the proof of FPC-inexpressibility for $\CSP(\mathbb{Q};X)$ from~\cite[Sect.~4.2]{bodirsky2022descriptive}.
A detailed inspection of the polynomial-time algorithm for $\CSP(\mathbb{L};Y)$ from~\cite[Figure~3]{bodirsky2017complexity} (rephrased in Figure~\ref{algo:LY} below) reveals that $\CSP(\mathbb{L};Y)$ admits a similar reformulation as $\CSP(\mathbb{Q};X)$: 

\medskip    
\setlength{\tabcolsep}{0pt}    
\hspace{1em}\begin{tabular}{ll}  
  INSTANCE: & \,  A finite homogeneous system $S$ of mod-2 equations of length 4. \\ 
      QUESTION: & \, Does each non-empty subset of $S$ have a non-constant solution?  
\end{tabular}    

\medskip    
Instead of eliminating minimal elements in each step as in the algorithm in Figure~\ref{algo:QX}, we divide the instance into the ``left half" and ``right half" of a corresponding solution tree.
We can easily extend Lemma~\ref{lemma:fooling} to its analogue tailored to this problem.

\begin{figure}[ht!]
\renewcommand{\algorithmiccomment}[1]{// #1}
  \hrule 
\smallskip 
    \begin{algorithmic} 
\STATE{INPUT: An instance $\struct{A}$ of $\CSP(\mathbb{L};Y)$}
\STATE{OUTPUT: \texttt{true} or \texttt{false}}
\IF{$A=\emptyset$}
\RETURN  \texttt{true} 
\ENDIF 
\FOR{$x_0,x_1 \in A$}    
\STATE{\COMMENT{interpret each $(x_i,x_j,x_k,x_{\ell})\in Y^{\struct{A}}$  as $x_i+x_j+x_k+x_{\ell}=0$}}
\IF{$\struct{A}\cup \{x_0=0,x_1=1\}$ mod 2 has a solution $s$}
\RETURN  $\text{divide-and-conquer}(\struct{A}[s^{-1}(0)])\wedge \text{divide-and-conquer}(\struct{A}[s^{-1}(1)])$  
\ENDIF
\ENDFOR 
\RETURN \texttt{false} 
\end{algorithmic} 
\smallskip 
  \hrule  
    \caption{A divide-and-conquer algorithm for $\CSP(\mathbb{L};Y)$}\label{algo:LY}
\end{figure}

 \begin{lemma} \label{lemma:fooling2}
    For every $k\in \mathbb{N}$, there exist structures $\struct{A}_k,\struct{B}_k$ in the signature consisting of symbols $R_i$ with arities $i$ for $i\in \{4,6\}$ with $\struct{A}_k \equiv_{k} \struct{B}_k$ and such that the following holds if $R_i(x_1,\dots, x_i)$ is interpreted as the equation $x_1+\cdots+x_i = 0 \bmod 2$:
    \begin{enumerate}
        \item \label{item:fooling1} $\struct{A}_k$ has a solution taking both values in each equation;
        \item \label{item:fooling2} $\struct{B}_k$ only has a constant solution. 
    \end{enumerate} 
 \end{lemma}
 \begin{proof}
 Let $\struct{A}'_k,\struct{B}'_k$ be the structures from Lemma~\ref{lemma:fooling} (formerly $\struct{A}_k,\struct{B}_k$).
 We obtain new structures $\struct{A}_k,\struct{B}_k$ by extending their domain by a fresh element $e$ and specifying the relations as follows:
 
\smallskip
\begin{itemize}
    \item $(x_1,x_2,e,e)\in R_4^{\struct{A}_k}$ ($R_4^{\struct{B}_k}$) if $(x_1,x_2)\in R_2^{\struct{A}'_k}$ ($R_2^{\struct{B}'_k}$); 
    \item $(x_1,x_2,x_3,e)\in R_4^{\struct{A}_k}$ ($R_4^{\struct{B}_k}$) if $(x_1,x_2,x_3)\in R_3^{\struct{A}'_k}$ ($R_3^{\struct{B}'_k}$);
    \item $(x_1,x_2,x_3,x_4,e,e)\in R_6^{\struct{A}_k}$ ($R_6^{\struct{B}_k}$) if $(x_1,x_2,x_3,x_4)\in R_4^{\struct{A}'_k}$ ($R_4^{\struct{B}'_k}$);
    \item $(x_1,x_2,x_3,x_4,x_5,e)\in R_6^{\struct{A}_k}$ ($R_6^{\struct{B}_k}$) if $(x_1,x_2,x_3,x_4,x_5)\in R_5^{\struct{A}'_k}$ ($R_5^{\struct{B}'_k}$).
\end{itemize}
\smallskip

It is easy to see that $\struct{A}_k \equiv_{k} \struct{B}_k$ holds given that $\struct{A}'_k \equiv_{k} \struct{B}'_k$; the Duplicator may simply extend the original bijections by $e \mapsto e$ in order to obtain a winning strategy.

For~\eqref{item:fooling1}, let $s$ be an assignment witnessing Lemma~\ref{lemma:fooling}\eqref{item:foolingold1}.
Then extending $s$ by $e\mapsto 0$ witnesses the condition in item~\ref{item:fooling1}.

For~\eqref{item:fooling2}, suppose that there is a non-constant solution $s$ to $\struct{B}_k$.
Since the equations in $\struct{B}_k$ are of even length, also $\bar{1}-s$ is a solution to $\struct{B}_k$; hence, without loss of generality, we may assume that $s(e)=0$.
But then the restriction of $s$ to $B'_k$ provides a non-trivial solution to $\struct{B}_k$, contradicting Lemma~\ref{lemma:fooling}\eqref{item:foolingold2}. Hence, item~\eqref{item:fooling2} holds as well.
\end{proof}
Consider the $6$-ary relation $Y'$ defined in $(\mathbb{L};Y)$ by the pp-formula
\[
\exists h \big( Y(x_1,x_2,x_3,h) \wedge Y(h,x_4,x_5,x_6)\big).
\]
 
\begin{theorem} \label{thm:phylo}
    $\CSP(\mathbb{L};Y,Y')$ is not definable in FPC. Moreover, this is witnessed by fixed-image fooling instances.  
\end{theorem} 
\begin{proof} 
We claim that $Y$ and $Y'$ consist of all tuples over $L$ of arity $4$ and $6$, respectively, which have a non-constant even split in $\struct{L}$.
In the case of $Y$, this follows directly from the definition.
For $Y'$, we show that $ (s_1,\dots,s_6)\in Y'$ if and only if there exists a permutation $\pi$ such that $s_{\pi(1)}s_{\pi(2)}| s_{\pi(3)}s_{\pi(4)}s_{\pi(5)}s_{\pi(6)}.$

``$\Rightarrow$'' We prove this direction by contradiction.
First, suppose that $s_1|s_2\cdots s_6$ holds. 
Then there is some $h \in \mathbb{L}$ such that $Y(s_1, s_2, s_3, h)$ and $Y(h, s_4, s_5, s_6)$ hold in $\mathbb{L}$. 
    Now $s_1h|s_2s_3$ must hold, and hence also $s_1h|s_2s_3s_4s_5s_6$.
    But then we get $h|s_4s_5s_6$, a contradiction to the definition of $Y$. 
    Secondly, suppose that $s_1s_2s_3|s_4s_5s_6$ holds; without loss of generality, we have $s_1s_2|s_3h$ and $hs_4|s_5s_6$. 
    But this implies $s_1s_2s_3h|s_4s_5s_6$ as well as $s_1s_2s_3|hs_4s_5s_6$, which contradict each other. Finally, assume that $s_1s_2s_4|s_3s_5s_6$. It follows that $s_1s_2|s_3h$; thus, $s_4|hs_5s_6$, which is a contradiction as well. 
    The above case distinction is exhaustive since $Y$ is symmetric. 

``$\Leftarrow$'' By the symmetry of $Y$, there are only two cases to consider.
First, suppose that $s_1s_2|s_3s_4s_5s_6$. Without loss of generality, we have $s_4|s_5s_6$. Due to the homogeneity of $(L; |\,)$, there exists some $h \in \mathbb{L}$ such that $hs_4|s_5s_6$. 
    Clearly, this implies $s_1s_2|s_3h$. Secondly, consider the case that $s_1s_4|s_2s_3s_5s_6$. The homogeneity of $(\mathbb{L}; |\,)$ provides some $h\in \mathbb{L}$ such that $s_1s_4h|s_2s_3s_5s_6$. Hence, $s_1h|s_2s_3$ and $s_4h|s_5s_6$.

This concludes the proof of our claim; given this characterization of $Y$ and $Y'$ via non-constant even splits, we can now prove the theorem.

For every $k\in \mathbb{N}$, let $\struct{A}_k,\struct{B}_k$ be the structures from Lemma~\ref{lemma:fooling2}.
Then $\struct{A}_k \equiv_{k} \struct{B}_k$.
By Lemma~\ref{lemma:fooling2}\eqref{item:fooling2}, $\struct{B}_k$ does not have any non-constant solution when viewed as a system of mod 2 equations.
Since the algorithm in Figure~\ref{algo:LY} is correct for $\CSP(\mathbb{L};Y,Y')$ (after transforming a $\CSP(\mathbb{L};Y,Y')$-instance into a $\CSP(\mathbb{L};Y)$-instance according to the pp-definition of $Y'$), we have that $\struct{B}_k \nrightarrow (\mathbb{L};Y,Y')$.
On the other hand, by Lemma~\ref{lemma:fooling2}\eqref{item:fooling1}, there exists a solution $s$ such that both $\struct{A}_k[s^{-1}(0)])$ and $\struct{A}_k[s^{-1}(1)]$ contain no equations.
Again, it follows immediately from the correctness of the algorithm in Figure~\ref{algo:LY} that $\struct{A}_k \rightarrow (\mathbb{L};Y,Y')$.
On the other hand, a homomorphism $h\colon \struct{A}'_k \rightarrow (\mathbb{L};Y,Y')$ can also be obtained explicitly by selecting arbitrary $a,b\in \mathbb{L}$ with $a\neq b$ and setting $h(x)\coloneqq a$ if $s(x)=0$ and $b$ otherwise.
Note that the cardinality of the image of $h$ is bounded by two.
The statement of the theorem for $\CSP(\mathbb{L};Y,Y')$ now follows from Theorem~\ref{thm:pebble_games}.  
\end{proof}

\subsection{Transfer of logical inexpressibility results}\label{section:inexpressibilitytransfer}
Theorem~\ref{thm:phylo} combined with Theorem~\ref{thm:main_theorem_sandwiches_basic} shows that, similarly to $(\mathbb{Q};X)$, the phylogeny template $(\mathbb{L};Y)$ pp-constructs a finite-domain PCSP that is not expressible in FPC.
Note that, since pp-constructions provide Datalog-reductions between (P)CSPs, we can actually interpret this as the reason why $(\mathbb{Q};X)$ and $(\mathbb{L};Y)$ are not expressible in FPC. 
The inexpressibility of finite-domain CSPs in FPC is characterized by the ability to pp-construct a non-trivial finite Abelian group~\cite[Thm.~1.1]{bodirsky2022descriptive}; here, a finite Abelian group $(G;+,0)$ is typically represented by a relational structure with domain $G$ and relations
\[
\{
(x_1,\dots,x_i)\in G^i \mid x_1+\cdots+ x_i = a
\}
\]
for all $i\in [3]$ and $a\in G$.

The above correspondence does not extend to the CSPs in the scope of Conjecture~\ref{conj:bodirsky_pinsker}~\cite[Thm.~1.5]{bodirsky2022descriptive}.
In fact, it already fails within the first-order reducts of $(\mathbb{Q};<)$, witnessed by the structure $(\mathbb{Q};X)$.
By~\cite[Prop.~7.4]{bodirsky2022descriptive}, the polymorphism clone of $(\mathbb{Q};X)$ satisfies a certain height-1 condition called \emph{$(3+3)$-terms}, originally introduced by Ol\v{s}\'{a}k~\cite{olvsak2021maltsev}. 
Operations $f,g,$ and $h$ of arities 6, 3, and 3, respectively, are called \emph{$(3+3)$-terms}~\footnote{In our work, we do not require $(3+3)$-terms to be \emph{idempotent}, which was the case in~\cite{olvsak2021maltsev}.} if they satisfy
\begin{align*}
f(y,x,x,x,x,x) &\approx g(y,x,x), \\
f(x,y,x,x,x,x) &\approx g(x,y,x), \\
f(x,x,y,x,x,x) &\approx g(x,x,y), \\
f(x,x,x,y,x,x) &\approx h(y,x,x), \\
f(x,x,x,x,y,x) &\approx h(x,y,x), \\
f(x,x,x,x,x,y) &\approx h(x,x,y).
\end{align*}
It is easy to show that the polymorphism clone of a non-trivial finite Abelian group cannot have $(3+3)$-terms; see the paragraph below Theorem~1.7 in~\cite{olvsak2021maltsev}.
Since $\Pol(\mathbb{Q};X)$ contains $(3+3)$-terms and pp-constructions preserve satisfiability of height-1 identities, we get that $(\mathbb{Q};X)$ cannot pp-construct any non-trivial finite Abelian group~\cite[Cor.~7.5]{bodirsky2022descriptive}. 

We will show that the above described phenomenon is not isolated, and that $(\mathbb{L};Y)$ is another template within the scope of Conjecture~\ref{conj:bodirsky_pinsker} whose CSP is not definable in FPC and yet it does not pp-construct any non-trivial finite Abelian group.
\begin{theorem} \label{thm:spiesstheorem}
    The polymorphism clone of $(\mathbb{L};Y)$ contains $(3+3)$-terms.
\end{theorem}
\begin{proof}
Let $(S_i)_{i\in\mathbb{N}}$ be a chain of finite subsets of $\mathbb{L}$ with $\bigcup_{i\in\mathbb{N}} S_i=\mathbb{L}$, i.e., so that $S_i \subseteq S_{i+1}$ for every $i\in \mathbb{N}$. 
 \begin{claim} \label{claim:3+3}
     For every $i\in \mathbb{N}$, there exist polymorphisms $f_i$, $g_i$, $h_i$ from $(S_i;Y\cap S^4_i)$ to $ (\mathbb{L};Y)$ witnessing the $(3+3)$-terms condition.
 \end{claim}
 Fix $i\in \mathbb{N}$.
 We prove Claim~\ref{claim:3+3} by induction on $|S_i|$. We will in fact use a stronger induction hypothesis than the original statement, which is captured by Claim~\ref{3plus3termsforX0X1}. 
\begin{claim}\label{3plus3termsforX0X1} 
Let $\struct{F}$, $\struct{F}_0$, $\struct{F}_1$ be finite substructures of $\struct{L}$ such that $\{F_0,F_1\}$ is the split-partition of $F$. 
Suppose that, for both $i\in \{0,1\}$,  $\Pol((F_j;Y\cap F_j^4),(\mathbb{L};Y))$ contains injective $(3+3)$-terms $f_i,g_i,h_i$ with $g_i(F_i^3)\subseteq f_i(F_i^6)$ and $h_i(F_i^3)\subseteq f_i(F_i^6)$. 
Then $\Pol((F;Y\cap F^4),(\mathbb{L};Y))$ contains injective $(3+3)$-terms $f,g,h$ with $g(F^3)\subseteq f(F^6)$ and $h(F^3)\subseteq f(F^6)$. Moreover, there exist $\alpha_0,\alpha_1\in\Aut(\mathbb{L};Y)$ such that $f|_{F_i^6}=\alpha_i\circ f_i,g|_{F_i^3}=\alpha_i\circ g_i,$ and $h|_{F_i^3}=\alpha_i\circ h_i$.
\end{claim}
\begin{claimproof}  
We may assume that $f_0(F_0^6)|f_1(F_1^6)$.
 Indeed, otherwise, by the homogeneity of $\Aut(\mathbb{L};Y)$, there exists an automorphism $\alpha\in \Aut(\mathbb{L};Y)$ such that $f_0(F_0^6)|\alpha\circ f_1(F_1^6)$.
 In this case,  $f_1'\coloneqq \alpha\circ f_1$, $g_1'\coloneqq \alpha\circ g_1$ and $h_1'\coloneqq \alpha\circ h_1$ also satisfy the prerequisites of this lemma, and we can continue with these operations instead.
 
 We first define the $6$-ary operation $f$, show that it is a polymorphism, and then use it to define the ternary operations $g$ and $h$.
  Before we define $f$, we need one auxiliary definition. We define the \emph{code} of a tuple $\bar{x} =(x_1,\dots,x_k)\in (F_0\cup F_1)^k$ with respect to $F_0$ and $F_1$ as the unique $\{0,1\}$-tuple $\code{\bar{x}}=(\code{x_1},\dots,\code{x_k})$ with the property that 
$x_i\in F_{\code{x_i}}$ for every $i\in [k]$. 
Note that if $F_0|F_1$, then codes of tuples are essentially splits where the values $0$ and $1$ explicitly denote the left and the right branch (or vice versa). 
 Now, consider a rooted balanced binary tree of height 6; we will refer to its $2^6$ many leaves as \emph{bubbles}, and each bubble will be represented by a tuple $\bar{b}=(b_1,\dots,b_6)\in\{0,1\}^6$.
 We define $f$ by extending each bubble $\bar{b}$ by a binary tree and assigning the leaves in this tree to tuples $\bar{x}\in F^6$ in such a way that $\code{\bar{x}}=\bar{b}$.
 In other words, we will map tuples with the same code w.r.t.\ $F_0$ and $F_1$ to the same bubble.
 We make a particular choice of branching when constructing $\struct{T}$, specified by Table~\ref{tab:branching}.
  \begin{table}[H]
      \centering \renewcommand{\tabcolsep}{0.5em}
      \begin{tabular}{l|c|c}    
       \toprule    level & left half  &  right half  \\ \midrule 
          1st   & $b_1=b_2$ & $b_1\neq b_2$ \\
             2nd    & $b_2=b_3$ & $b_2\neq b_3$ \\
       3rd    & $b_4=b_5$ & $b_4\neq b_5$ \\
        4th     & $b_5=b_6$ & $b_5\neq b_6$ \\
         5th   & $b_3=b_4$ & $b_3\neq b_4$ \\
         6th    & $b_1=0$ & $b_1=1$ \\ \bottomrule
      \end{tabular}
      \bigskip 
      \caption{An assignment of leaves to bubbles. Here, the $0$- and the $1$-entries do not necessarily specify branching to the left and the right   (or vice versa).}
      \label{tab:branching}
  \end{table}
   We illustrate this choice in Figure~\ref{fig:definitiontreef}. 
\begin{center}
\begin{figure}[H]
\resizebox{\textwidth}{!}{
    \begin{forest}
for tree={
  grow'=south,
  parent anchor=south,
  child anchor=north,
  l=8mm,
  s=4mm,
  draw,
  inner sep=4pt,
  font=\tiny,
  line width=0.3mm
}
[root,draw=white
  [$b_1\neq b_2$
    [$b_2\neq b_3$
      [$b_4\neq b_5$[{$\vdots$},draw=white][{$\vdots$},draw=white]]
      [{$b_4=b_5$}
        [$b_5\neq b_6$
          [$b_3\neq b_4$
            [101001, circle]
            [010110, circle]
          ]
          [{$b_3=b_4$}
            [101110, circle]
            [010001, circle]
          ]
        ]
        [{$b_5=b_6$}
          [$b_3\neq b_4$
            [101000, circle]
            [010111, circle]
          ]
          [{$b_3=b_4$}
            [101111, circle]
            [010000, circle]
          ]
        ]
      ]
    ]
    [{$b_2=b_3$}[{$\vdots$},draw=white][{$\vdots$},draw=white]]
  ]
  [{$b_1=b_2$}
    [{$\vdots$},draw=white][{$\vdots$},draw=white]]
]
\end{forest}}
\centering
    \caption{An illustration of the definition-tree of $f$. The figure covers the subtree of all codes $(b_1,\dots,b_6)$ such that $b_1\neq b_2$, $b_2\neq b_3$, and $b_4=b_5$.}
    \label{fig:definitiontreef}
\end{figure}
\end{center}
Let us now define the subtrees extending the bubbles.
If $\bar{b}=(i,\dots,i)$ for some $i\in\{0,1\}$, then we extend $\bar{b}$ by a tree such that its leaf structure is isomorphic to $f_i(F_i^6)$. 
For all other bubbles, we want to have the image under $f$ \emph{lexicographically dominated}:
That means for all $\bar{x}=(x_1,\dots,x_{\ell}),\bar{y}=(y_1,\dots,y_{\ell}),\bar{z}=(z_1,\dots,z_{\ell})\in F^{\ell}$ such that $\code{\bar{x}}=\code{\bar{y}}=\code{\bar{z}}\notin \{(0,\dots,0),(1,\dots,1)\}$ the following holds: if there exists $k\in [\ell]$ such that $x_i=y_i=z_i$ for every $i<k$ while $x_k|y_kz_k$, then $f(\bar{x})|f(\bar{y})f(\bar{z})$.  
To this end, we extend the bubble $\bar{b}=(b_1,\dots,b_6)$ by a tree whose leaf structure is isomorphic to $\struct{F}_{b_1}$, then each of its leaves by a tree whose leaf structure is isomorphic to $\struct{F}_{b_2}$, etc.
This way, we get a tree that has $|F_{b_1}|\cdot \cdots \cdot |F_{b_6}|$ many leaves and, for the definition of $f$, we assign to each tuple $\bar{x}=(x_1,\dots,x_6)$ with $\code{\bar{x}}=\bar{b}$ the leaf that is reached by going to $x_1$ in the first level, then to $x_2$ in the subtree extending $x_1$, etc. It can be easily checked that an operation defined in this way is lexicographically dominated.
 
Now consider the leaf structure $(L(\struct{T});|) $ of this tree; it can be embedded into $\struct{L}=(\mathbb{L};|)$. Since $Y$ is defined 
 in $(\mathbb{L};|)$ using a quantifier-free formula, namely the formula in eq.~\eqref{eq:even_split}, there exists an embedding from $\struct{S} \coloneqq (L(\struct{T});Y\cap L(\struct{T})^4 )$ into $(\mathbb{L};Y)$.
 This embedding defines a function $f\colon F^6\to \mathbb{L}$. Clearly, for both $i\in\{0,1\}$, we have $f|_{F_i^6}=\alpha_i\circ f_i$ for some $\alpha_i \in \Aut(\mathbb{L};Y)$.  
Since $f_0$ and $f_1$ are injective, so is $f$.

\medskip\noindent{\itshape Claim: $f$ is a polymorphism.}

\medskip

Let  $\bar{x},\bar{y},\bar{z},\bar{w}\in F^6$ be arbitrary such that
$\struct{F} \models Y(x_i,y_i,z_i,w_i)$  for every $i\in [6]$.
        We have to show that 
         \begin{align}
      (f(\bar{x}),f(\bar{y}),f(\bar{z}),f(\bar{w}))\in Y^{\struct{L}}. \label{eq:claimf}
       \end{align}   
        
        First, suppose that  $\code{\bar{x}}=\code{\bar{y}}=\code{\bar{z}}=\code{\bar{w}}$.
        If the codes are constant, then~\eqref{eq:claimf}  follows from $f_0$ or $f_1$ being a polymorphism.
        Otherwise, we get~\eqref{eq:claimf} because the split relation $|$ in the image is determined by the projection $(x_k,y_k,z_k,w_k)$ to the first argument $k\in [6]$ where the tuples $\bar{x},\bar{y},\bar{z},\bar{w}$ differ.

        Second, suppose that $\code{\bar{x}},\code{\bar{y}},\code{\bar{z}},\code{\bar{w}}$ are not all equal.
        Then we consider the first branching rule that divides the tuples into the left and right half of the corresponding subtree (Table~\ref{tab:branching}).
        E.g, it could be the branching step dividing between $b_4=b_5$ and $b_4\neq b_5$ for $\code{\bar{x}}$ and $\code{\bar{y}}$, in which case $\code{x_4} =\code{x_5} $ and $\code{y_4} \neq \code{y_5} $.
        Without loss of generality, we assume that $\code{x_4} =0$ and $\code{y_4} =0$ ($0$ and $1$ are treated symmetrically).
        Since $F_0|F_1$ and the relation $Y$ is charaterized by non-constant even splits, we have that either $\code{z_4} =\code{z_5} $ and $\code{w_4} \neq \code{w_5} $ or $\code{z_4} \neq \code{z_5} $ and $\code{w_4} = \code{w_5} $.
        Hence, two of the codes will branch to the left and two to the right. 
        Therefore, we get~\eqref{eq:claimf}.
        All other cases in the first five branching rules work analogously.

        If all the tuples are only divided by the last branching rule, then again the characterization of $Y$ by non-constant even splits yields that two tuples will branch to the left and two to the right, yielding~\eqref{eq:claimf}. Thus, $f$ is a polymorphism.

\medskip Now, we define $g$ and $h$.  Let $\bar{x}=(x_1,x_2,x_3)\in F^3$ be arbitrary.   We distinguish the following two cases. 
Suppose that $\code{\bar{x}}\in \{(0,0,0),(1,1,1)\}$.
Recall that $f|_{F_i^6}=\alpha_i\circ f_i$ for both $i\in\{0,1\}$.
Then, for both $i\in\{0,1\}$ and $\bar{x}\in F_i^3$, we set  
  \begin{align*}
    g(\bar{x})\coloneqq \alpha_i\circ g_i(\bar{x}) \quad \textnormal{ and } \quad h(\bar{x})\coloneqq \alpha_i\circ h_i(\bar{x}).
  \end{align*}
Otherwise $\code{\bar{x}} \notin \{(0,0,0),(1,1,1)\}$.
Since $\code{\bar{x}}$ is non-constant, it contains either two 0s or two 1s. Let $x_*$ be the second entry in $\bar{x}$ from $F_i$, where $i$ appears twice in $\code{\bar{x}}$.\footnote{Here we could have chosen the first entry as well.} We set
    \begin{align*}
     g(\bar{x})\coloneqq f(x_1,x_2,x_3,x_*,x_*,x_*) \quad \textnormal{ and } \quad h(\bar{x})\coloneqq f(x_*,x_*,x_*,x_1,x_2,x_3).
    \end{align*}
        Clearly, we have $g(F^3)\subseteq f(F^6)$ and $h(F^3)\subseteq f(F^6)$. Moreover, since $g_0$, $g_1$, $h_0$, $h_1$, and $f$ are injective, so are $g$ and $h$.

\medskip
\noindent{\itshape Claim: $g$ is a polymorphism.} 

\medskip  

    Let  $\bar{x},\bar{y},\bar{z},\bar{w}\in F^3$ be arbitrary such that $\struct{F} \models Y(x_i,y_i,z_i,w_i)$  for every $i\in [3]$.
   We need to show that  
   \begin{align}
   (g(\bar{x}),g(\bar{y}),g(\bar{z}),g(\bar{w}))\in Y^{\struct{L}}. \label{eq:claimg}
   \end{align}
   To this end, we distinguish three cases.

   First, suppose that $\code{\bar{x}}=\code{\bar{y}}=\code{\bar{z}}=\code{\bar{w}}$.
   If the code is constant, then~\eqref{eq:claimg} follows follows from $g_i$ being a polymorphism.
   Otherwise~\eqref{eq:claimg} follows since $g$ inherits being lexicographically dominated from $f$.

   Next, suppose that  $\code{\bar{y}},\code{\bar{z}},\code{\bar{w}}$ are all equal to $\code{\bar{x}}$ or its complement mod 2.  
            If $\code{\bar{x}}$ is constant, then since $F_0|F_1$, there must be two tuples from $F_0^3$ and two from $F_1^3$ among $\bar{x},\bar{y},\bar{z},\bar{w}$. Since $g(F_i^3)\subseteq f(F_i^6)$ for both $i\in\{0,1\}$ and $f(F_0^6)|f(F_1^6)$, we get~\eqref{eq:claimg}. 
            Otherwise $\code{\bar{x}}$ is non-constant. If $\code{\bar{y}}=\code{\bar{x}}$, then also 
            \[
            (\code{y_1} ,\code{y_2} ,\code{y_3} ,\code{y_*} ,\code{y_*} ,\code{y_*} )=(\code{x_1} ,\code{x_2} ,\code{x_3},\code{x_*} ,\code{x_*} ,\code{x_*} ).
            \]
            If $\code{\bar{y}}\neq\code{\bar{x}}$, then 
            \[
(\code{y_1} ,\code{y_2} ,\code{y_3},\code{y_*} ,\code{y_*},\code{y_*})=(\code{x_1} ,\code{x_2} ,\code{x_3} ,\code{x_*} ,\code{x_*} ,\code{x_*} ) + (1,\dots, 1) \bmod 2.
            \]
            The same clearly holds for $\bar{z}$ and $\bar{w}$ extended by $(z_*,z_*,z_*)$ and $(w_*,w_*,w_*)$, respectively.
            In the definition of $f$, all extended tuples agree on the first 5 branching steps.
            Since $Y$ is characterized by non-constant even splits and $F_0|F_1$, the codes of the extended tuples contain two pairs.
            Hence, also the images of the extended tuples under $f$ which correspond to the images of $\bar{x},\bar{y},\bar{z},\bar{w}$ under $g$ yield an even split.
            In sum, we get~\eqref{eq:claimg}.   
  
           In the last remaining case, there exist $i,j\in [3]$ such that, without loss of generality, $\code{x_i} =\code{x_j} $ and $\code{y_i} \neq\code{y_j}$. 
           Then, since $$\struct{F} \models  Y(x_i,y_i,z_i,w_i) \wedge Y(x_j,y_j,z_j,w_j),$$ 
             either $\code{z_i} =\code{z_j}$ and $\code{w_i} \neq \code{w_j} $, or $\code{z_i} \neq\code{z_j}$ and $\code{w_i} = \code{w_j} $.
            Since $g(F_i^3)\subseteq f(F_i^6)$ for both $i\in\{0,1\}$ and since the image of a tuple $\bar{v}\in F^3\backslash(F_0^3\cup F_1^3)$ under $g$ lies in the bubble $(\code{v_1} ,\code{v_2} ,\code{v_3} ,\code{v_*} ,\code{v_*} ,\code{v_*} )$ in the definition tree of $f$ (see Figure~\ref{fig:definitiontreef}), we only need to check whether the bubbles corresponding to $\bar{x},\bar{y},\bar{z},\bar{w}$ satisfy $Y$. But it can be easily checked that one of the first two branching steps will split the codes evenly. 
            Hence, we get~\eqref{eq:claimg}. 
        The case distinction is exhaustive, and hence $g$ is a polymorphism.

\medskip\noindent{\itshape Claim: $h$ is a polymorphism.}

\medskip

       The arguments are the same as for $g$ except for the last case. 
       First, note that $h(F_i^3)\subseteq f(F_i^6)$ and that all tuples $\bar{x}$ with a non-constant code $\code{\bar{x}}$ are mapped to some element in the bubble corresponding to $(\code{x_*},\code{x_*},\code{x_*},\code{x_1},\code{x_2},\code{x_3})$ in the definition tree of $f$ (see Figure~\ref{fig:definitiontreef}). Hence, the images also satisfy $Y$, since they fall into the branch specified by $b_1=b_2$ and $b_2=b_3$ and then end up being split evenly at the third or fourth branching step by the same arguments as for $g$.

\medskip\noindent{\itshape Claim: $f,g,h$ satisfy the $(3+3)$-terms condition.}

\medskip 
 
        Let $x,y\in F$ be arbitrary. If $x,y\in F_i$ for $i\in\{0,1\}$, then 
        \[
        g(y,x,x)=\alpha_i\circ g_i(y,x,x)=\alpha_i\circ f_i(y,x,x,x,x,x)=f(y,x,x,x,x,x).
        \]
        Otherwise, we have $\code{y}\neq\code{x}$, hence, $x$ is the second entry in $(y,x,x)$ with code $\code{x}$. Thus, we also have
        $  
        g(y,x,x)=f(y,x,x,x,x,x).
        $  
        All other cases can be treated analogously. 
        
        \medskip 
        
    This finishes the proof of Claim~\ref{3plus3termsforX0X1}. 
\end{claimproof}

 We can now prove Claim~\ref{claim:3+3}.
If $|S_i|=1$, then we choose the operations $f_i,g_i,h_i$ that all map the only possible input $(x,\dots,x)$ to $x$.
Now, suppose that $f_i,g_i,h_i$ satisfying the stronger induction hypothesis from Claim~\ref{3plus3termsforX0X1} have already been defined on each subset of $S_i$ of size $<|S_i|$.
We set $F\coloneqq S_i$ and consider the split-partition $\{F_0,F_1\}$ thereof. 
Then Claim~\ref{3plus3termsforX0X1} immediately yields the desired operations on $S_i$.
We now show that we can impose an additional requirement on said operations.

 \begin{claim} Let $f_i,g_i,h_i$ for $i\in\mathbb{N}$ be constructed as above. Then for every $i\in \mathbb{N}$, there exists $\alpha_{i+1}\in \Aut(\mathbb{L};Y)$ such that $\alpha_{i+1}\circ f_{i+1}|_{S_{i}}=f_{i}$.
\end{claim}
            \begin{claimproof}
                Fix an arbitrary $i\in\mathbb{N}$. By the homogeneity of $(\mathbb{L};Y)$, it suffices to show that, for all $\bar{x},\bar{y},\bar{z},\bar{w}\in  S_i^6$, we have 
                \begin{align} 
                 (\mathbb{L};Y) \models Y(f_i(\bar{x}),f_i(\bar{y}),f_i(\bar{z}),f_i(\bar{w})) \Leftrightarrow Y(f_{i+1}(\bar{x}),f_{i+1}(\bar{y}),f_{i+1}(\bar{z}),f_{i+1}(\bar{w}))  \label{eq:iff}
                \end{align}
                To this end, let $M$ be the set of all entries appearing in $\bar{x},\bar{y},\bar{z},$ and $\bar{w}$. 
                Without loss of generality, we may assume  that $|M|>1$; otherwise the claim is trivially true.
                Let $\{M_0,M_1\}$ and $\{F_0,F_1\}$ be the split-partitions of $M$ and $F\coloneqq S_i$, respectively. 
                If $M\subseteq F_0$ or $M\subseteq F_1$, then we only need to look at the recursive definition of $f_i$ and $f_{i+1}$ on $F_0^6$ or $F_1^6$. 
                Hence, we may without loss of generality assume that $M_0\subseteq F_0$ and $M_1\subseteq F_1$; if $M$ is not fully contained in some $F_i$, we always have $M\cap F_0|M\cap F_1$.
                We may also assume that the split-partition $\{F'_0,F'_1\}$ of $F'\coloneqq S_{i+1}$ satisfies  $F_i\subseteq F_i'$ for both $i\in\{0,1\}$.
                Otherwise $S_i\subseteq F_0'$ or $S_i\subseteq F_1'$.
                In this case, in what follows, we only need to inspect the definition of $f_{i+1}$ on $(F_0')^6$ or $(F_1')^6$.
                
                Note that the splits between the bubbles associated to $\code{\bar{x}},\code{\bar{y}},\code{\bar{z}},\code{\bar{w}}$ (interpreted as codes marking which entries lie in $F_0$ and which in $F_1$ or, equivalently, which lie in $F_0'$ and which in $F_1'$) are exactly the same for $f_i$ and $f_{i+1}$. 
                Thus, if not all of the codes are the same, the splits between the images of the tuples under $f_i$ and $f_{i+1}$ agree since they are determined by the position of the bubbles in Figure~\ref{fig:definitiontreef}.
                Otherwise, if all the codes are equal, they fall into the same bubble. Since this must be a bubble with a non-constant code, the splits between the images of the tuples under $f_i$ and $f_{i+1}$ agree since they are determined by  lexicographic domination.
                Hence, eq.~\eqref{eq:iff} indeed holds.
            \end{claimproof}
    Note that following the same arguments and since $g_i$ and $h_i$ are defined via $f_i$, we also get $\alpha_{i+1}\circ g_{i+1}|_{S_{i}}=g_{i}$ and $\alpha_{i+1}\circ h_{i+1}|_{S_{i}}=h_{i}$.
    Finally, by setting $f_0'\coloneqq f_0$, $g_0'\coloneqq g_0$, $h_0'\coloneqq h_0$, and 
    \[
        f_i' \coloneqq \alpha_1\circ\dots\circ\alpha_i \circ f_i,   \qquad 
        g_i' \coloneqq \alpha_1\circ\dots\circ\alpha_i \circ g_i,  \qquad 
        h_i' \coloneqq \alpha_1\circ\dots\circ\alpha_i \circ h_i,
     \]
    we get chains of polymorphisms $(f_i')_{i\in\mathbb{N}}$, $(g_i')_{i\in\mathbb{N}}$, and $(h_i')_{i\in\mathbb{N}}$. Now \[f\coloneqq  \bigcup\nolimits_{i\in\mathbb{N}}f_i',\qquad g\coloneqq \bigcup\nolimits_{i\in\mathbb{N}}g_i',\qquad h\coloneqq \bigcup\nolimits_{i\in\mathbb{N}}h_i' \] are the required $(3+3)$-terms in $\Pol(\mathbb{L};Y)$. 
\end{proof}

We can now prove Theorem~\ref{thm:temporal_phylo_pp_constr}. 
\begin{proof}[Proof of Theorem~\ref{thm:temporal_phylo_pp_constr}]
The concrete template is $(\mathbb{L};Y)$. 

For item~\eqref{item:finite_reason_1}, recall that $Y'$ is pp-definable in $(\mathbb{L};Y)$. 
By Theorem~\ref{thm:phylo}, $\CSP(\mathbb{L};Y,Y')$ is inexpressible in FPC, and this is witnessed by fixed-image fooling instances.
Consequently, the PCSP provided by Theorem~\ref{thm:main_theorem_sandwiches_basic} from $\CSP(\mathbb{L};Y,Y')$ for $n$ large enough is not expressible in FPC.
By construction, said PCSP template is pp-constructible from $(\mathbb{L};Y)$. 

For item~\eqref{item:finite_reason_2}, we use the fact that, by Theorem~\ref{thm:spiesstheorem}, the polymorphism clone of $(\mathbb{L};Y)$ contains $(3+3)$-terms.
As in the proof of~\cite[Cor.~7.5]{bodirsky2022descriptive}, we conclude that $(\mathbb{L};Y)$ does not pp-construct any non-trivial finite Abelian group.
Then the claim follows because, for a finite structure $\struct{B}$, we have that $\CSP(\struct{B})$ is expressible in FPC if and only if $\struct{B}$ does not pp-construct any non-trivial finite Abelian group (\cite[Thm.~1.1]{bodirsky2022descriptive}).
\end{proof}

Recall that while definability in FPC and the ability to pp-construct a non-trivial finite Abelian group complement each other in the realm of finite-domain CSPs,  there are infinite structures within the scope of Conjecture~\ref{conj:bodirsky_pinsker}  who satisfy neither the former nor the latter. Yet, it is still possible that every structure within that scope whose CSP is not definable in FPC pp-constructs a finite-domain PCSP that is not definable in FPC.

\begin{openquestion}\label{q:FPC}
Does every structure in the scope of Conjecture~\ref{conj:bodirsky_pinsker} whose CSP is not definable in FPC pp-construct a finite-domain PCSP that is not definable in FPC?
\end{openquestion}

\subsection{The quest for uniform algorithms}

Theorem~\ref{thm:phylo} combined with Theorem~\ref{thm:main_theorem_sandwiches_basic} shows that, analogously to $(\mathbb{Q};X)$, also $(\mathbb{L};Y)$ pp-constructs a finite-domain PCSP that is not expressible in FPC.
As discussed in Section~\ref{sec:phylo_not_in_FPC},  $\CSP(\mathbb{Q};X)$ and $\CSP(\mathbb{L};Y)$ are similar decision problems, in the sense that the classic divide-and-conquer algorithms for solving them are similar.
Recently, uniform algorithms for all tractable temporal CSPs, in particular  for $\CSP(\mathbb{Q};X)$, were presented in~\cite{Mottet_2025}; they are based on a combination of polynomial-time \emph{sampling} and either \emph{singleton arc-consistency} or \emph{singleton affine integer programming}. Unfortunately, the  correctness proofs for these algorithms still use the ad-hoc framework of free sets, and in contrast to linear orders, it is not clear how to efficiently sample over leaf structures: for every $n\in \mathbb{N}$, there are exponentially many non-isomorphic trees with $n$ leaves (but only a single linear order of length $n$).
Hence, despite the fact that $\CSP(\mathbb{Q};X)$ and $\CSP(\mathbb{L};Y)$ are similar, it is unclear how to efficiently solve the latter using an algorithm based on numeric relaxation.

 \begin{openquestion}\label{q:phylogeny}
Can tractable phylogeny CSPs, and hence also the PCSPs produced from them via Theorems~\ref{thm:main_theorem_sandwiches} and~\ref{thm:main_theorem_sandwiches_basic}, be solved using a uniform algorithm based on local consistency checking and numeric relaxation --- similarly to temporal CSPs? 
\end{openquestion}

\section{Conclusion and outlook}\label{section:conclusion} 
 
In the present article, we have given two significant  simplifications of the scope of Conjecture~\ref{conj:bodirsky_pinsker} -- Theorems~\ref{thm:removing_algebraicity} and~\ref{thm:polinjective}.
We believe this will have rich consequences for the approach to the conjecture itself, but also for other applications; we offer one concrete example,  Theorem~\ref{thm:main_theorem_sandwiches}.
Below we offer some further more detailed discussion of our results.

\subsection{Is everything Ramsey?}
In the proof of Theorem~\ref{thm:main_theorem_sandwiches}, we use Theorem~\ref{thm:removing_algebraicity} to ensure that the obtained finite-domain PCSPs are not finitely tractable, which we could otherwise only show in certain specific cases.
The precise contribution of Theorem~\ref{thm:removing_algebraicity} is that it allows us to select:
\begin{itemize}
    \item a finite substructure whose polymorphisms preserve the inequality and 
    \item a finite factor by an arbitrary subgroup of automorphisms preserving a linear order and acting with finitely many orbits. 
\end{itemize} 

Regarding the former, while the inequality is often not even the only possible choice of a relation that would make our argument work, some structures in the scope of Conjecture~\ref{conj:bodirsky_pinsker} provably require a non-trivial modification before they can be used as a cheese of a non-finitely tractable PCSP; see Section~\ref{section:preprocessing}.

The latter is often possible without any modifications of the original structure -- it is conjectured that every finitely bounded homogeneous structure has a finitely bounded homogeneous ordered Ramsey expansion~\cite[Prob.~27]{bodirsky_pinsker_tsankov}; see also~\cite{the2014survey}.

In the proof of Theorem~\ref{thm:removing_algebraicity} we modify templates in order to reduce the scope of Conjecture~\ref{conj:bodirsky_pinsker} to structures without algebraicity. Motivated by this, we ask if something similar can be achieved with regards to the existence of finitely bounded homogeneous Ramsey expansions, a question weaker than \cite[Conj.~11.1.1]{bodirsky2021complexity} and another concretization of Main~Question~\ref{mq1}.
 \begin{openquestion}\label{q:Ramsey}
Is Conjecture~\ref{conj:bodirsky_pinsker} equivalent (under suitable reductions) to its restriction to templates with a finitely bounded homogeneous Ramsey expansion?
\end{openquestion}
 
 For details about structural Ramsey theory, we refer the reader to~\cite{hubivcka2025twenty} where the authors also  provide a concrete example of a finitely bounded homogeneous structure with no algebraicity for which no Ramsey expansion is known (Sect.~8.1.1).

\subsection{Topology seems irrelevant}
  
  In the context of infinite-domain CSPs, the influence of topological properties of the polymorphism clones of $\omega$-categorical structures has received much attention, with various results claiming their relevance~\cite{bodirsky_relevant_2019} or irrelevance~\cite{barto_pinsker_journal}: while generally, whether or not an $\omega$-categorical structure $\struct{A}$ pp-constructs $\struct{K}_3$ (or pp-interprets another $\omega$-categorical structure $\struct{A}'$) does depend on such topological properties~\cite{barto2018wonderland,Topo-Birk,uniformbirkhoff,  bodirsky2021projective}, in many situations it does not (see, e.g.,~\cite{EJMPP,marimon2025guidetopologicalreconstructionendomorphism}). 
    By moving from any $\omega$-categorical model-complete core template to one without algebraicity, Theorem~\ref{thm:removing_algebraicity} allows us to restrict Conjecture~\ref{conj:bodirsky_pinsker} to a class of topologically well-behaved structures. Namely, for those recent research suggests that often the algebraic structure determines the topological one; at least this is a fact for the endomorphism monoid, where the relevant topology (\emph{pointwise convergence}) can then be defined from the algebraic structure of the monoid alone~\cite{PINSKER_SCHINDLER_2023}.
     For $\omega$-categorical structures with algebraicity, on the other hand, various examples~\cite{evans_counterexample_1990, bodirsky_counterexample_2018,PINSKER_SCHINDLER_2023,gonzalez2026minimalintrinsictopologiesmonoids} show that this need not be the case. This also implies that there is no hope of lifting isomorphisms between endomorphism monoids or polymorphism clones of $\omega$-categorical structures to their counterparts constructed in Theorem~\ref{thm:removing_algebraicity}.

\subsection{Necessity of cheese preprocessing} \label{section:preprocessing}

In our proof of Theorem~\ref{thm:main_theorem_sandwiches}, we apply   Proposition~\ref{prop:sandwiches} to a structure which we previously ``preprocessed'' using Theorem~\ref{thm:removing_algebraicity} (removing algebraicity) and generic superpositions. 
It is natural to ask how strong Proposition~\ref{prop:sandwiches} is on its own: can a result similar to Theorem~\ref{thm:main_theorem_sandwiches} be proved directly using this proposition,  perhaps under a  more general structural assumption, e.g.~that  $\struct{A}$ is a model-complete core?
Here the answer seems to be negative.

\begin{example} 
Consider the model-complete core  structure \[\struct{A}=(\mathbb{Q};<,\{(x,y,z) \in \mathbb{Q}^{3} \mid x\geq y \text{ or } x\geq z\})\] within the scope of  Conjecture~\ref{conj:bodirsky_pinsker} ($\struct{A}_2$ in Example~\ref{ex:two_ex}); its CSP is  polynomial-time tractable~\cite{ComplOfTempCSPs}.
Let $\struct{S}$ be any finite structure for which there exists a homomorphism $h\colon \struct{S}\rightarrow \fpwr{\struct{A}}{d}$, and let $\struct{B}$ be an arbitrary $\omega$-categorical expansion of $\struct{A}$.
Let $A_S$ be the set of the elements in $A$ that appear as an entry of some element of $h(S)$, and let  $\struct{A}_S$ be the substructure of $\struct{A}$ on $A_S$.
Since $\smash{\fpwr{\struct{A}}{d}\to \orbeq{\fpwr{\struct{A}}{d}}{\Aut(\struct{B})}}$, the structure $\fpwr{\struct{A}}{d}_S$ is a finite cheese for $ \PCSP(\struct{S},\orbeq{\fpwr{\struct{A}}{d}}{\Aut(\struct{B})}) $.

 We claim that $ \CSP(\fpwr{\struct{A}}{d}_S) $ is tractable, and hence this PCSP is finitely tractable.  
 To see this, recall from Example~\ref{ex:two_ex} that the $n$-ary minimum operation is a polymorphism of $\struct{A}$ for every $n\geq 2$. This operation is \emph{conservative}:
 $\min\{x_1,\dots,x_n\}\subseteq \{x_1,\dots, x_n\}$ holds for all $x_1,\dots,x_n$.
 Hence, its restriction to $A_S^2$ induces a polymorphism of $\struct{A}_S$, which itself induces a polymorphism of $\fpwr{\struct{A}}{d}_S$ through its component-wise action.
 Since this operation is cyclic, $\fpwr{\struct{A}}{d}_S$ does not pp-construct $\struct{K}_3$ \cite[Thm.~1.4]{barto2018wonderland}, and hence we are done by Theorem~\ref{thm:finite_domain_CSP}.
 This issue cannot be fixed simply by taking an expansion of $\struct{A}$ by the inequality relation $\neq$ because $\CSP(\struct{A};\neq)$ is NP-complete~\cite{ComplOfTempCSPs}.
\end{example}

\subsection{Algorithmic properties of sandwiches} An interesting open question is whether the PCSP templates generated by Theorem~\ref{thm:main_theorem_sandwiches} can be solved by some uniform algorithm.
Many known PCSP problems admitting infinite tractable cheeses can be solved by numeric relaxation algorithms such as \emph{BLP} (the basic linear programming relaxation), \emph{AIP} (the basic affine integer relaxation) or \emph{BLP+AIP} (a combination of both) (see, e.g., \cite{barto2021algebraic,BLP_AIP}).
The proof of Proposition~\ref{prop:sandwiches} shows that none of the PCSP templates produced by Theorem~\ref{thm:main_theorem_sandwiches} admits a cyclic polymorphism.
Since solvability of PCSPs by BLP is characterized by admitting symmetric polymorphisms of all arities~\cite[Thm.~7.9]{barto2021algebraic}, which would also be cyclic, BLP does not solve any of those PCSPs.
Regarding the combination BLP+AIP, we remark that replacing $\blowup{\struct{A}}$ and $\smash{\ordblowup{\struct{B}}}$ in the proof of Theorem~\ref{thm:main_theorem_sandwiches} by 
$\blowupi{\struct{A}}$ and $\ordblowupi{\struct{B}}$ ensures that there is no cyclic and no 2-block symmetric operation in $\Pol(\struct{S},\blowupi{\struct{A}})$.
In fact, we get the following much stronger statement, which follows directly from the fact that we take an expansion by $I_4$. 
\begin{proposition}[cf.~proof of Lemma~7.5.1 in~\cite{bodirsky2021complexity}] \label{prop:lifting}
       Let $\Sigma$ be any finite height-1 condition that cannot be satisfied by any set of essentially injective operations from an at least 2-element set $S_1$ to a countably infinite set $S_2$. Then, for every at least 2-element substructure $\struct{S}$ of $\blowupi{\struct{A}}$, we have that   $\Pol(\struct{S},\blowupi{\struct{A}})$ does not satisfy $\Sigma$.  
\end{proposition}
Similarly to cyclic identities, also the 2-block symmetric identities can be lifted from $\smash{\Pol(\fpwr{\struct{S}}{d},\orbeq{\fpwr{\widehat{\struct{A}}}{d}}{G})}$ to $\Pol(\struct{S},\widehat{\struct{A}})$, where 
\begin{itemize}
    \item $\widehat{\struct{A}}\coloneqq\blowupi{\struct{A}}$ and 
    \item $G\coloneqq \Aut\big(\ordblowupi{\struct{B}} \astlarge (\mathbb{Q};<)\big).$
\end{itemize} 
Thus, since solvability of a finite-domain PCSP by BLP+AIP is characterized by the existence of 2-block symmetric polymorphisms \cite[Thm.~4]{BLP_AIP}, we get that $\PCSP(\fpwr{\struct{S}}{d},\orbeq{\fpwr{\widehat{\struct{A}}}{d}}{G})$ is not solvable by BLP+AIP.
For details, we refer the reader to~\cite[Thm.~2]{Mottet_2025}, where this idea was used in concrete cases of non-finitely tractable PCSPs. 
It follows that we can exclude BLP and BLP+AIP as potential algorithms solving all tractable PCSPs produced by Theorem~\ref{thm:main_theorem_sandwiches}. 
As demonstrated in~\cite[Lem.~33]{Mottet_2025}, the above idea can in fact be extended to any uniform algorithm such that solvability of finite-domain PCSPs by this algorithm is captured by a set of height-1 identities specified by non-trivial permutations of variables.
Most known uniform algorithms for finite-domain PCSPs are either captured by such height-1 conditions or their algebraic description is rather complicated~\cite{ciardo2023hierarchies}.
However, the lifting method we and~\cite{Mottet_2025} use provably cannot be extended to all height-1 conditions covered by Proposition~\ref{prop:lifting}.
To see this, note that the Ol\v{s}\'{a}k identities 
\[f(y,x,x,x,y,y)\approx f(x,y,x,y,x,y)\approx f(x,x,y,y,y,x)\] 
cannot be satisfied by essentially injective operations but every finite-domain PCSP whose template does not have an Ol\v{s}\'{a}k polymorphism must be NP-hard~\cite[Cor.~6.3]{barto2021algebraic}. 
This means that the Ol\v{s}\'{a}k identities cannot possibly be prevented in our PCSPs without losing solvability in polynomial time (unless P=NP).

\subsection{Algorithmic properties of infinite-domain CSPs}

Bodirsky and Rydval~\cite[Thm.~1.6]{bodirsky2022descriptive} showed that the expressibility of CSPs from Conjecture~\ref{conj:bodirsky_pinsker} in Datalog is not captured by any set of identities for polymorphisms.
Our Theorem~\ref{thm:polinjective} shows that the complexity classes of CSPs in Conjecture~\ref{conj:bodirsky_pinsker} closed under Datalog-reductions can only be captured by identities satisfiable by essentially injective operations.
It might seem that these results essentially rule out the usefulness of height-1 conditions in the infinite-domain setting.
However, a recently announced result of Zhuk~\cite{zhuk2025singleton} points in the exact opposite direction: there might be a way to link the tractability of infinite-domain CSPs to height-1 identities\textemdash by considering modified versions of the notion of a polymorphism.

 Zhuk~\cite{zhuk2025singleton} shows that the applicability of the uniform algorithms from~\cite{Mottet_2025} to temporal CSPs is always witnessed by the so-called \emph{palette totally symmetric} identities, a certain type of height-1 identities that cannot be satisfied by essentially injective operations.  
But how can this be true, say, for the template $(\mathbb{Q};\neq, I_4)$, which only has essentially injective polymorphisms?
The important detail we omitted is that these identities only need to be satisfied in a certain \emph{palette-variant} of the polymorphism clone~\cite{zhuk2025singleton}, not the polymorphism clone itself.
In some cases, the polymorphism clone provably does not contain any palette totally symmetric operations; the structure $(\mathbb{Q};\neq, I_4)$ is an example of this phenomenon~\cite{zhuk2025singleton}.
Understanding Zhuk's palette-polymorphisms might help us in further development of the algebraic approach to infinite-domain CSPs.

\bibliographystyle{abbrv}
\bibliography{references}

\end{document}